\documentclass[11pt,letterpaper]{article}

\usepackage[margin=1in]{geometry}
\usepackage[utf8]{inputenc}
\usepackage{amsmath}
\usepackage{amsthm}
\usepackage{xcolor}
\usepackage{amssymb}
\usepackage{algorithm}
\usepackage[noend]{algpseudocode}
\usepackage{mathtools}
\usepackage{paralist}
\usepackage{cite} 
\usepackage[colorlinks=true, citecolor=blue, linkcolor=magenta]{hyperref}
\usepackage{tikz}
\usepackage{tkz-tab}
\usepackage{caption}
\usepackage{subcaption}
\usepackage{thm-restate}

\usetikzlibrary{backgrounds}
\tikzset{
	edge/.style={very thick, gray},
	medge/.style={decorate,very thick,decoration={snake}},
	aedge/.style={very thick,dashed,black},
	dedge/.style={thick,->},
	availedge/.style={thick,blue},
	vertex/.style={shape=circle,thick,draw,node distance=3em}
}
\DeclareMathOperator*{\argmin}{arg\,min}
\newcommand{\glue}{\mathsf{Glue}}
\newcommand{\contractvsglue}{\mathsf{\text{Contract-vs-Glue}}}
\newcommand{\eps}{\varepsilon}
\newcommand{\alg}{\mathsf{ALG}}
\newcommand{\Deg}{\mathsf{deg}}
\newcommand{\opt}{\mathsf{opt}}
\newcommand{\OPT}{\mathsf{OPT}}
\newcommand{\ecss}{2\text{-ECSS}}
\newcommand{\savings}{\mathsf{savings}}
\newcommand{\credit}{\mathsf{credit}}
\newcommand{\newcredit}{\mathsf{newcredit}}
\newcommand{\load}{\mathsf{load}}

\newcommand{\reducea}{\mathsf{Reduce}}
\newcommand{\divideT}{\mathsf{Divide}}
\newcommand{\dividea}{\mathsf{Divide0}}
\newcommand{\divideb}{\mathsf{Divide1}}
\newcommand{\dividec}{\mathsf{Divide2}}
\newcommand{\divided}{\mathsf{Divide3}}
\newcommand{\dividee}{\mathsf{Divide4}}
\newcommand{\dividef}{\mathsf{Divide5}}
\newcommand{\combine}{\mathsf{Combine}}
\newcommand{\combinea}{\mathsf{Combine0}}
\newcommand{\combineb}{\mathsf{Combine1}}
\newcommand{\combinec}{\mathsf{Combine2}}
\newcommand{\combined}{\mathsf{Combine3}}
\newcommand{\combinee}{\mathsf{Combine4}}
\newcommand{\combinef}{\mathsf{Combine5}}

\usepackage[textsize=scriptsize,color=red,textwidth=1.8cm]{todonotes}

\newtheorem{theorem}{Theorem}
\newtheorem{lemma}{Lemma}
\newtheorem{definition}{Definition}
\newtheorem{proposition}{Proposition}

\newtheorem{fact}{Fact}
\newtheorem{remark}{Remark}
\theoremstyle{remark}

\title{Matching Augmentation via Simultaneous Contractions}

\author{
Mohit Garg\thanks{Department of Computer Science and Automation, Indian Institute of Science, Bengaluru, India.  \ mohitgarg@iisc.ac.in} 
\and 
Felix Hommelsheim\thanks{Faculty of Mathematics and Computer Science, University of Bremen, Germany. \{fhommels,nmegow\}@uni-bremen.de}  
\and 
Nicole Megow\footnotemark[2]
}
\date{\today}

\begin{document}

\maketitle
\begin{abstract}

We consider the matching augmentation problem (MAP), where a matching of a graph needs to be extended into a $2$-edge-connected spanning subgraph by adding the minimum number of edges to it. We present a polynomial-time algorithm with an approximation ratio of $13/8 = 1.625$ improving upon an earlier $5/3$-approximation. The improvement builds on a new $\alpha$-approximation preserving reduction for any $\alpha\geq 3/2$ from arbitrary MAP instances to well-structured instances that do not contain certain forbidden structures like parallel edges, small separators, and contractible subgraphs. We further introduce, as key ingredients, the technique of repeated simultaneous contractions and provide improved lower bounds for instances that cannot be contracted.
\end{abstract}

\section{Introduction}
In the  \texttt{Matching Augmentation Problem~(MAP)}, we are given an undirected graph $G$, where each edge $e\in E(G)$ has a weight in $\{0,1\}$, and all the zero-weight edges form a matching. The task is to compute a minimum weight $2$-edge-connected spanning subgraph ($\ecss$) of $G$, which is a connected graph $(V(G),F)$ with $F \subseteq E(G)$ that remains connected on deleting an arbitrary edge.

\texttt{MAP} is a fundamental problem in the field of {\em survivable network design} and is known to be MAX-SNP-hard with several better-than-$2$ approximation algorithms~\cite{CheriyanCDZ21,CheriyanDGKN20,BamasDS22}. 
Prior to this work, the best-known approximation ratio for \texttt{MAP} was $\frac 5 3$, achieved by Cheriyan et al.~\cite{CheriyanCDZ21}.

Both~\cite{CheriyanCDZ21,CheriyanDGKN20} provide combinatorial algorithms for \texttt{MAP}, where the approximation ratios are achieved by comparing the outputs against the minimum-cardinality $2$-edge-cover ($D_2$). A {\em $2$-edge-cover} of an undirected graph is a spanning subgraph in which each node has a degree of at least~$2$, but it may not be connected. Thus, computing a $D_2$ is a relaxation of \texttt{MAP}. In contrast to solving \texttt{MAP}, a $D_2$ can be computed exactly in polynomial time by extending Edmonds' matching algorithm~\cite{edmonds1965paths}. 
A weaker approximation result for \texttt{MAP} by Bamas et al.~\cite{BamasDS22} follows a very different approach: the output is compared against another lower bound on an optimal \texttt{MAP} solution, obtained by solving a linear programming relaxation, 
the so-called Cut LP. The integrality gap of the Cut LP is at least~$\frac{4}{3}$~\cite{BamasDS22}.

\medskip 
\noindent 
\textbf{Our result.} 
We present a polynomial-time algorithm for \texttt{MAP} with an approximation ratio of $\frac{13}{8}=1.625$, improving the previous best ratio of $5/3$. 

\begin{theorem}\label{thm:main}
	There is a polynomial-time $\frac{13}{8}$-approximation algorithm for \texttt{MAP}.
\end{theorem}

This improvement builds on a new $\alpha$-approximation preserving reduction for any $\alpha\geq 3/2$ from arbitrary \texttt{MAP} instances to well-structured instances that do not contain certain forbidden structures like parallel edges, small separators, and contractible subgraphs. We further introduce, as key ingredients, the technique of repeated simultaneous contractions and provide improved lower bounds for instances that cannot be contracted.

\medskip 
\noindent 
\textbf{Further related work.} 
\texttt{MAP} sits between the minimum unweighted $\ecss$ and the minimum weighted $\ecss$ problems.
For the minimum weighted $\ecss$ problem, improving known $2$-approximations~\cite{Jain01,AgrawalKR95,WilliamsonGMV95,KhullerV94} is a major open problem. Whereas for the unweighted case, in a recent breakthrough, Garg et al.~\cite{GargGJ22} provided a $1.326$-approximation, improving the earlier $\frac 4 3$-approximations~\cite{HunkenschroderV19,SeboV14}.

Research on the minimum weighted $\ecss$ problem assuming that the set of zero-weight edges in the input graph has a certain structure such as
forest, spanning tree, matching, or disjoint paths has received a lot of attention recently. A general variant is the \texttt{Forest Augmentation Problem (FAP)}, where the edges in the input graph have $0$/$1$ edge weights. For \texttt{FAP}, only recently, Grandoni et al.~\cite{GrandoniAT22} obtained a $1.9973$-approximation, breaking the approximation barrier of $2$.
A famous special case of FAP is the unweighted \texttt{Tree Augmentation Problem (TAP)} where the zero-weight edges in the input graph form a single connected component. In a long line of research, several better-than-$2$ approximations have been achieved for \texttt{TAP}~\cite{Adjiashvili19,CecchettoTZ21,CheriyanG18,CheriyanG18a,CohenN13,EvenFKN09,Fiorini0KS18,GrandoniKZ18,Nutov21,KortsarzN18,KortsarzN16,Nagamochi03,TraubZ21,TraubZ22}, culminating in a $1.393$-approximation by Cecchetto et al.~\cite{CecchettoTZ21}.

Notice that 
\texttt{MAP} is somewhat orthogonal to \texttt{TAP} in terms of connectivity as it has many small connected components as input instead of a single big one. Understanding both the extreme cases well, \texttt{TAP} and \texttt{MAP}, seems promising for making further progress for \texttt{FAP}.  

\medskip
\noindent
\textbf{Organization of the rest of the paper.}
Section~\ref{sec:techOverview} contains preliminaries and a high-level overview of our work along with some important definitions.
Section~\ref{sec:mainPreprocessing} and Section~\ref{sec:mainAlgorithm} consist of the description of our reduction and algorithm, respectively, along with the corresponding theorem and lemma statements which we prove in the appendix.
Using these theorems and lemmas, in Section~\ref{sec:mainTheoremProof} we prove Theorem~\ref{thm:main}. 
In Section~\ref{sec:conclusion}, we conclude with final remarks, pointing out the bottleneck for improving our algorithm. 
 
 We have included detailed proofs of various lemmas in Appendices~\ref{sec:running-time}-\ref{sec:Glue}, which makes our write-up lengthy. 
 A lot of material, especially in Appendices~\ref{sec:running-time}, \ref{sec:bridgeCovering}, and~\ref{sec:Glue} are standard, but formally necessary; the new innovations are mainly contained in Appendices~\ref{sec:preprocessing}, \ref{sec:ComputingSpecialConfiguration}, and~\ref{sec:lower-bound}.
 While some proofs admit a case analysis, no single proof has too many cases. 
 We have tried to keep the exposition reader-friendly and make the proofs easily verifiable at the expense of making the write-up a bit lengthy; a terser style might have saved some pages.

\section{Technical overview}\label{sec:techOverview}
\subsection{Preliminaries}
We use standard notation for graphs. We consider weighted undirected graphs where each edge has a weight of either $0$ or $1$. A {\em MAP instance}  consists of a graph $G$ such that the zero-weight edges of $G$ form a matching, and the task is to compute a minimum weight
$2$-edge-connected spanning subgraph ($\ecss$) of $G$, which is a connected subgraph $(V(G),F)$ which remains connected on deletion of an arbitrary edge.
Without loss of generality, we may assume that the input graph $G$ is $2$-edge-connected; this can be checked in polynomial time by testing for each edge whether its deletion results in a disconnected graph. 

Given a set of edges $F \subseteq E(G)$, $||F||$ denotes the weight of $F$, i.e., the number of unit edges in $F$. With slight abuse of notation, we denote the weight of a subgraph $H$ of $G$ by $||H||$. Thus, $|| H || = ||E(H)||$.
Given a MAP instance $G$, let $\OPT(G)$ represent a $2$-edge-connected spanning subgraph of $G$ of minimum total weight $\opt(G):=||\OPT(G)||$. When $G$ is clear from the context we sometimes omit $G$.

Whenever we speak of components of a graph we refer to its {\em connected} components.

\subsection{Algorithmic template and the previous $\frac 5 3$-approximation}

The algorithm and analysis of Cheriyan et al.~\cite{CheriyanCDZ21}, for obtaining a $\frac 5 3$-approximate solution for MAP, exemplifies the general template for obtaining combinatorial algorithms for $2$-edge-connected spanning subgraphs used in several works\cite{HunkenschroderV19,EvenFKN09,KortsarzN16,CheriyanDGKN20}. We first explain this template, by giving an overview of the algorithm of Cheriyan et al.~\cite{CheriyanCDZ21}, and
then in the next subsection highlight our approach where we alter this template to achieve our improvement.

 The algorithm consists of two parts. The first part is a preprocessing step which constitutes a $\frac 5 3$-approximation preserving reduction from arbitrary MAP instances to well-structured instances that do not contain certain forbidden structures. This is a key element of their work, which helps them to improve upon an earlier $\frac 7 4$-approximation by getting rid of certain hard instances. 

In the second part, they handle instances that do not contain any of the forbidden structures through a {\it discharging scheme}. Their algorithm starts by computing a minimum $2$-edge-cover, $D_2$ (in polynomial time). Additionally, all the zero-edges are added to the $D_2$, so that the edges not in the $D_2$ are all unit-edges.  Now, since a $\ecss$ is a $2$-edge-cover, $||D_2||$ lower bounds $\opt$, the weight of the minimum weight $\ecss$. To output a $(1+c)$-approximate solution (for $c=\frac 2 3$), one has $(1+c)||D_2||$ {\it charge} to work with. 
This charge is used to buy the edges of the $D_2$ and a charge of $c$ is distributed to each of the unit edges of the $D_2$. Now, they incrementally transform this $D_2$ into a $\ecss$ by adding edges to it. For each edge that is added, a charge of $1$ is used up from the available charge (which is taken from nearby edges), i.e., their $D_2$ incrementally evolves into a $\ecss$ at the expense of {\it discharging}. This is an oversimplified view of their actual algorithm. In reality, sometimes they even delete edges in the process which results in gaining charge.

We briefly describe the two steps involved in transforming the $D_2$ into a $\ecss$, namely {\it bridge covering} and {\it gluing}. Note that a $D_2$ can have several connected components. Some of these components can be $2$-edge-connected, whereas some might have bridges (i.e., deleting those edges will result in increasing the number of components). The first step is to {\it cover} all the bridges one by one. Given a bridge, they add edges so that the bridge becomes part of a cycle; as a side effect, multiple components might merge into one. At the end of the bridge-covering step, their graph has only $2$-edge-connected components.
They ensure that after using up the charge for buying the edges in the process, each component with at least $3$ unit-edges has at least a charge of $2$ leftover. 
Components with exactly $2$ unit-edges (cycles of lengths $3$ and $4$) 
keep the initial charge of $2c = \frac 4 3$. 

Next, in the gluing step, the components are merged into a single component using the leftover charge in the components resulting in a feasible solution. To 
see how this might be done, momentarily assume that all components have at least $3$ unit-edges, i.e., having a leftover charge of at least $2$. Here, one can simply contract each component into a single node, find a cycle in the contracted graph, and buy all the edges in that cycle. 
This will result in merging all the components corresponding to the nodes in the cycle into a single component.
To be able to repeat such merges, we need to ensure that we have a leftover charge of at least $2$ in this newly formed component. 
For a cycle of length $k\geq 2$, initially there was a charge of at least $2k$ in the corresponding components, and we need to buy exactly $k$ edges. 
Thus, after the merge, the leftover charge in the new component is at least $2k-k=k\geq 2$, maintaining the charge invariant.
Repeating this process eventually leads to a feasible solution.
Unfortunately, this idea is not guaranteed to work if there are components with $2$ unit-edges; a cycle with $2$ nodes corresponding to such components will have a total charge of $2\times \frac 4 3 =\frac 8 3 $, and after buying the $2$ edges in the cycle, we will be left with a charge of only $\frac 2 3$. On repeating such merges, the graph will run out of charge before the gluing finishes. 
To handle such small components one needs to delete edges to gain charge.

\subsection{Highlights of our approach and innovations for the $\frac{13}{8}$-approximation}
Our approach follows the same broad framework explained above with some key innovations.
Formal definitions will be given in later sections.

\paragraph*{Preprocessing.} For all $\alpha \geq \frac 3 2$, we provide an $\alpha$-approximation preserving reduction from arbitrary MAP instances to {\it structured graphs}. Our list of forbidden structures subsumes the one 
by Cheriyan et al.~\cite{CheriyanCDZ21} and consists of parallel edges, cut vertices, small separators, and {\it contractible subgraphs}.

Contractible subgraphs are a general form of contractible cycles as considered in~\cite{HunkenschroderV19}. 
A $2$-edge-connected subgraph $H$ of a graph $G$ is contractible if each $\ecss(G)$ includes at least $\frac 1 \alpha||H||$ unit-edges from $G[V(H)]$. 
Since we are interested in only an $\alpha$-approximate solution, we may contract  $V(H)$ into a single node, solve the problem on the contracted graph, and add the edges of~$H$ to the solution without 
any loss in the approximation ratio. As an example, suppose a $6$-cycle in $G$ of weight 6 has 2 antipodal vertices that have 
degree 2 in $G$. Then, $\OPT(G)$ must include the $4$ edges incident on these two vertices. As the cycle costs 6, and $\OPT(G)$ is guaranteed to pick at least weight $4$ from the subgraph induced on the vertices of the cycle, for a $\frac 3 2$-approximation, it suffices to buy all 
edges of this cycle, contract it and solve the reduced problem. We can detect all contractible subgraphs with constant-size vertex sets in polynomial time and remove them during preprocessing. Interestingly, some intricate structures considered in~\cite{CheriyanCDZ21} are simply contractible~subgraphs.

We further exclude several small separators, which is crucial for our bridge covering and gluing steps as we have less charge at our disposal. 
Given a separator, we split the graph into two or three parts, recursively solve the problem on the smaller parts, and then combine the solutions arguing that the approximation ratio is preserved. 
If each of these parts has at least 5 vertices, this step is relatively straightforward. But for parts containing at most $4$ vertices, the argument becomes significantly more challenging, in particular, since we are aiming for a better guarantee than previous work. 
Handling the small separators forms a substantial part of our work consisting of several innovations. 
In particular, given a separator that splits the graph into two parts, the structure of the interaction of the separator with the parts is exploited in carefully constructing the subproblems. 
Here, we sometimes introduce {\it pseudo-edges}, representing possible connections via the other part, and suitably remove them during the combining step.
Our reduction, which works for any $\alpha\geq \frac 3 2$, might be useful for future works.
We only need $\alpha=\frac {13}8$ for our main result.

\paragraph*{Bridge covering.}
Empowered by a stronger preprocessing, we can rule out more structures in the input graph, which enables us to obtain a bridgeless $2$-edge-cover of $G$, even for an approximation ratio of $(1+c)$, for $c= \frac{5}{8}$.
In fact, our bridge-covering works even for $c = \frac{3}{5}$, so it might also be useful for future works.
At the end of the bridge-covering step, we have the following charges left in the $2$-edge-connected components: $2$ in the {\it large} components (containing $4$ or more unit-edges), $3c= \frac {15}8 < 2$ in the {\it medium} components (containing $3$ unit-edges) and $2c=\frac{5}4 \ll 2$ in the {\it small} components (containing $2$ unit-edges).

\paragraph*{Gluing.}
In the gluing step, we are able to merge all the medium components into large components even though medium components have strictly less than $2$ charge. We are also able to handle some small components that have a particular configuration by deleting edges and gaining charge. 
Unfortunately, we were unable to handle all the small components as they have a minuscule charge. 
In the end, we are left with a {\em special} configuration that has only large (with charge $\geq 2$) and small components (with charge $\frac 5 4$) which cannot be merged.

\paragraph*{Two-edge-connecting special configurations.}

The small components of the special configuration originate in the initial $D_2$ and could not be merged. 
So we ask the following question. How {\it close} are these small components to $\OPT$ restricted to the vertices of the small components?
Intuitively, if they are close, we should be able to do something algorithmically as we are roughly doing what $\OPT$ is doing on this part of the graph. 
Otherwise, if they are not close, we should be able to argue that $\OPT$ does much worse than what the $D_2$ does on this part of the graph, giving us an improved lower bound. Our main conceptual innovation is in articulating a notion of closeness and making this intuition work. 

\paragraph*{Method of simultaneous contractions.}
We now describe our measure of closeness. Let $G$ be our input structured graph and let $H_1,\cdots, H_{s}$ be the small components of the special configuration obtained. 
We count the total number of unit-edges bought by $\OPT$ from the following subgraphs $G[V(H_1)],\cdots,G[V(H_{s})]$. If this number is more than $\frac 8{13}$ times the number of unit-edges in the small components of our special configuration, which is precisely $2s$, we say that the small components are close to $\OPT$. 
Otherwise, they are not close. 

Observe when the small components are indeed close,  on average, each $H_i$ is contractible, preserving an approximation ratio of $\frac {13} 8$. 
Thus, algorithmically, we can simultaneously contract each $V(H_i)$ into a distinct single node, solve the problem on the reduced instance (which can be done recursively, as contracting vertices into nodes decreases the size of the graph), and add the edges of the small components to the solution, without incurring a loss in approximation. 

 When the small components are not close to $\OPT$,  i.e., the $H_i$'s are not {\it simultaneously contractible}, we rely on the gluing step of Cheriyan et al.~\cite{CheriyanCDZ21} using a charge of $\frac 4 3$ per small component instead of our original charge of $\frac 5 4$, increasing our cost. Our improvement, in this case, comes from improving the lower bound.
 
 Note that it is not possible for us to check in polynomial time whether the small components are simultaneously contractible or not; so what should we do -- contract, or use the gluing algorithm of Cheriyan et al.~\cite{CheriyanCDZ21}? We do both and return the solution with a smaller weight and argue that in either scenario the algorithm performs well.

 \paragraph*{Improved lower bound.} In the case when the small components are not simultaneously contractible, $\OPT$ picks at most $\frac 8{13}\cdot 2s$ unit-edges from the $G[V(H_i)]$'s put together. Thus, at least $2s - \frac {16}{13}s = \frac{10}{13}s$ unit-edges are not picked from within the small components. We show that for each unit-edge not picked by $\OPT$ from this part, $\OPT$ buys on average at least $1+\frac 1{12}$  edges that go between different small components. To argue this, we crucially use the fact that the special configurations have a restricted structure, as certain merges are not possible in it. Thus, we show that in total the number of unit-edges used by $\OPT$ on the vertices of the small components is at least $\frac 8 {13}\cdot 2s + (1+\frac 1{12})\frac{10}{13}s=\frac {161}{78}s$, which is strictly more than $2s$, which is the number of unit-edges used by the $D_2$ on this part. Finally, through an elegant argument, we are able to use this improved lower bound on $\OPT$ restricted only to the vertices of the small components to show that it compensates for the increased cost incurred during gluing the small components.

\subsection{Important definitions}
We 
give some definitions that we need for the presentation of our algorithm.
\begin{definition}[$f(.)$]
Given a MAP instance $G$, let $$f(G) = \max\{\frac {13} 8 \cdot \opt(G) -2, \opt(G)\}.$$ 
\end{definition}
For a MAP instance $G$, we will compute a $\ecss$ of $G$ with weight at most $f(G)$. Observe that the `$-2$' term gives us a slightly better bound than claimed, which we crucially exploit in our preprocessing.
\begin{definition}[size of a graph $s(.)$]
Given a graph $G$, its {\bf size} is $s(G)=10\cdot |V(G)|^2 + |E(G)|.$
\end{definition}
We will show that the running time of our algorithm is upper bounded by a polynomial in the size of the input graph.

\begin{definition}[notation graph contraction]
Given a graph $G$ and a set of vertices $T\subseteq V(G)$, $G/T$ denotes the graph obtained from $G$ after contracting all the vertices in $T$ into a single vertex. More generally, given disjoint vertex sets $T_1,\cdots T_k\subseteq V(G)$, $G/\{T_1,\cdots,T_k\}$ denotes the graph obtained from $G$ after contracting vertices of each $T_i$ into single vertices.

Note that edges in $G$ and $G/\{T_1,\cdots,T_k\}$ are in one-to-one correspondence. Given a subgraph $H$ of the contracted graph, we use $\hat H$ to refer to the subgraph of $G$ containing precisely those edges that correspond to the edges of $H$.
\end{definition}

\begin{definition}[contractible subgraphs]
\label{structuredGraphDef:contractible}
Let $\alpha \geq 1$ and $t \geq 2$ be fixed constants. Given a $2$-edge-connected graph $G$, a collection of vertex-disjoint $2$-edge-connected subgraphs $H_1, H_2, ...,$ $H_k$ of $G$ is called \emph{$(\alpha, t , k)$-contractible} if $2 \leq |V(H_i)| \leq t$ for every $i \in [k]$ and every $\ecss$ of $G$ contains at least $\frac{1}{\alpha} ||\bigcup_{i \in [k]} E(H_i)||$ unit-edges from  $\bigcup_{i \in [k]} E(G[V(H_i)])$.
\end{definition}

In our preprocessing, we will remove all $(\frac {13}8,12,1)$-contractible subgraphs, which we simply refer to as contractible subgraphs. Later, when considering a special configuration with $n_s$ small components, we will work with a $(\frac {13}8, 4,n_s)$-contractible collection of small components; we will refer to the special 
configuration simply as $\frac{13}8$-simultaneously contractible.

\subsection{Algorithm overview}
Here, we give a brief overview of our main algorithm.
\begin{itemize}
    \item[Step 1:] Preprocessing: We apply our reduction to obtain a collection of subproblems of MAP on structured graphs (Section~\ref{sec:mainPreprocessing}). We then assume that we are given some structured graph~$G$.
    \item[Step 2:] Bridge covering: We compute a $D_2$ in polynomial-time and apply bridge covering to obtain an \emph{economical} bridgeless 2-edge-cover $H$ -- a bridgeless 2-edge-cover of low cost (Section~\ref{subsex:main:bridge-covering}).
    \item[Step 3:] Special configuration: Given $H$, we compute a special configuration $S$ of $G$ (Section~\ref{subsec:main:computing-special}).
    \item[Step 4:] Contract vs.~glue: We compute two feasible solutions $S_1$ and $S_2$: $S_2$ is obtained by applying the algorithm of~\cite{CheriyanCDZ21} to $G$ and $S$ (Section~\ref{subsec:main:feasible-for-special}); $S_1$ is obtained by calling 
    Step~1 for~$G_S$, 
    which arises from $G$ by contracting each small component of $S$ to a single vertex. Finally, we output $\argmin \{||S_1||,||S_2||\}$.
\end{itemize}

\section{Preprocessing}\label{sec:mainPreprocessing}

We show that, for purposes of approximating MAP with any approximation ratio at least $\frac 3 2$, it suffices to consider MAP instances that do not contain certain forbidden configurations. 
These configurations are cut vertex, parallel edge, contractible subgraph, $S_0$, $S_1$, $S_2$, $S_{\{3,4\}}$, $S_3$, $S_4$, $S_5$, $S_6$, $S'_3$, $S'_4$, $S'_5$, and $S'_6$. 
The formal definitions of these structures are provided in Appendix~\ref{sec:preprocessing}. Each of these configurations is referred to as a {\bf type} and is of constant size. A MAP instance with at least~$20$ vertices that does not contain any of these forbidden configurations is termed as {\bf structured}. 

We briefly describe some of the types that we forbid in a structured graph. Apart from cut vertex, parallel edge, and contractible subgraph, the other forbidden structures we consider can be broadly divided into two categories: (a) ``Path-like''-separators and (b) ``Component-like''-separators. Path-like separators are certain paths which when removed from the input graph disconnects it. 
Forbidding these structures in the structured graph is mostly used in the bridge-covering step. 
Roughly speaking, the absence of  these structures helps us in finding sufficient credit while covering some path (consisting of bridges) between 2-edge connected blocks of the 2-edge-cover.
Component-like structures, on the other hand, are  certain 2-edge-connected subgraphs that when removed from the input graph disconnects it. 
Their absence from structured graphs is mainly exploited in the gluing step; it allows us to find certain cycles through some small components which help us gain credit that is needed for the gluing.


The reduction from MAP instances to structured graphs is given by the following algorithm where we assume
$\alg$ is an algorithm that works on structured graphs. Our reduction is essentially a divide-and-conquer algorithm. 
It searches for a forbidden configuration and if it detects one, it divides the problem into a few subproblems (at most 3) of smaller sizes, solves them recursively, and then combines the returned solutions into a solution for the original instance. 
In case there are no forbidden configurations in the input (the input is structured), it calls $\alg$ to solve the problem.

\begin{algorithm}
\caption{Preprocessing}\label{alg:reduce}
\begin{algorithmic} 
\Function{Reduce}{$G$} 
\If{$G$ is simple and $|V(G)| \leq 20 $} \Return $\opt(G)$. \Comment{\textcolor{gray}{by brute force}} \State 
	\EndIf
	
	\State Look for a forbidden configuration in $G$ in the following type order: 
	\State cut vertex, parallel edge, contractible subgraph, $S_0$, $S_1$, $S_2$, $S_{\{3,4\}}$, $S_k$, $S'_k$ for $k\in\{3,4,5,6\}$. 
	\State Stop immediately on detecting a forbidden configuration. 
	
	\State
	
	\If{a forbidden configuration is detected} 
	\State Call it F and let T be the type of $F$.
	
	\State $(H_1,H_2, H_3)=\divideT_T(G, F)$.
	\Comment{\textcolor{gray}{$H_2$ and/or $H_3$ are always empty for certain types}}
	
	\State $H_i^*=\reducea(H_i)$ for all $i\in\{1,2,3\}$.
	
	\State \Return $\combine_T(G,H_1^*,H_2^*,H_3^*)$.
	\EndIf
	
	\Comment{\textcolor{gray}{$G$ is now structured}}
	
\State \Return $\alg(G)$.
\EndFunction
\end{algorithmic}
\end{algorithm}

In the above algorithm, $\divideT_T$ and $\combine_T$ are subroutines that are defined in Appendix~\ref{sec:preprocessing}, which also contains proofs of the following lemmas.

\begin{restatable}{lemma}{lempolytimesubs}
\label{lem:polytimesubs}
For all types $T$, $\divideT_T$ and $\combine_T$ are polynomial time algorithms. Furthermore, given a MAP instance $G$ and a type $T$, one can check in polynomial time whether $G$ contains a forbidden configuration of type $T$.
\end{restatable}

\begin{restatable}{lemma}{lemdivideCombine}
\label{lem:divideCombine}
Given a MAP instance $G$ and a forbidden configuration $F$ that appears in $G$ of type $T$ from the list $L=$(cut vertex, parallel edge, contractible subgraph, $S_0$, $S_1$, $S_2$, $S_k$, $S_k'$, $k \in \{3, 4, 5, 6\}$) such that
 $G$ does not contain any forbidden configuration of a type that precedes $T$ in the list $L$ and  $\divideT_T(G,F)=(H_1,H_2, H_3)$, then the following statements hold:
 \begin{itemize}
     \item[(i)] for each $i\in\{1,2, 3\}$ $H_i$ is a MAP instance, 
     \item[(ii)] $s(H_1) + s(H_2) + s(H_3) < s(G)$, and 
     \item[(iii)] if for each $i\in\{1,2,3\}$, $H_i^*$ is a $\ecss$ of $H_i$ such that $||H_i^*|| \leq f(H_i)$, then \linebreak  $\combine_T(G,H_1^*,H_2^*,H_3^*)$ is a $\ecss$ of $G$ such that $||\combine_T(G,H_1^*,H_2^*,H_3^*)||\leq f(G)$.
 \end{itemize}
\end{restatable}

Using the above lemma, we can establish the following result: If for all structured graphs $G$, $\alg(G)$ is a 
$\ecss$ of $G$ such that $||\alg(G)|| \leq f(G)$, then for all MAP instances $G$, $\reducea(G)$ is a $\ecss$ of $G$ such that $||\reducea(G)|| \leq f(G)$.

We will produce an {\it admissible} $\alg$ that calls $\reducea$ on a smaller instance. Since $\reducea$ and $\alg$ call each other, we need a slightly stronger result, which is obtained using an induction argument.  We first define an admissible algorithm.

\begin{definition}[admissible]
An algorithm $\alg$ is {\bf admissible} if the following holds. If for all MAP instances $G$ with $s(G)\leq t$,  $\reducea(G)$ is a $\ecss$ of $G$ such that $||\reducea(G)|| \leq f(G)$,
then for all structured graphs $G$ such that $s(G)=t+1$, $\alg (G)$ is a $\ecss$ of $G$ such that $|| \alg(G) || \leq f(G)$. Furthermore, if $T(s)$ denotes the running time of $\alg$ for structured graphs of size $s$, and $T'(s)$ denotes the running time of $\reducea$ on MAP instances of size $s'$, then 
$T(s)\leq T'(s-1) + poly(s)$. 
\end{definition}

Our main results in this section are the following.
\begin{theorem}\label{thm:reduce}
If $\alg$ is an admissible algorithm, then for all MAP instances $G$, $\reducea(G)$ is a $\ecss$ of $G$ such that $||\reducea(G)||\leq f(G)$.
\end{theorem}

\begin{theorem}\label{thm:runningtime}
If $\alg$ is an admissible algorithm, then $\reducea$ runs in polynomial time.
\end{theorem}

The proofs of the above two results follow from a straightforward application of Lemmas~\ref{lem:polytimesubs} and~\ref{lem:divideCombine} and are included in Appendix~\ref{sec:running-time}.
Now, if we can find an admissible $\alg$, Theorem~\ref{thm:reduce} and Theorem~\ref{thm:runningtime} immediately imply Theorem~\ref{thm:main}. In the next subsections, we exhibit an admissible $\alg$.

\section{Algorithm for structured graphs}\label{sec:mainAlgorithm}
We exhibit an admissible algorithm $\alg$ that takes as input a structured graph $G$ and outputs a $\ecss$ of $G$ with weight at most $f(G)$.
Our algorithm has three main steps. First, we compute an `economical' bridgeless $2$-edge-cover of $G$. Then, with the aid of this $2$-edge-cover, we compute a `special' configuration of $G$. 
We two-edge-connect the special configuration in two ways and return the solution with minimum weight. We define the relevant terms and explain these steps below.

\begin{algorithm}
\caption{Main algorithm for structured graphs}\label{alg:main-structured}
\begin{algorithmic} 
\Function{$\alg$}{$G$} 
\Comment{\textcolor{gray}{$G$ is structured}} 
	\State $H = \mathsf{\text{economical bridgeless $2$-edge-cover}}(G)$
	
	\State $S = \mathsf{\text{special configuration}}(G, H)$
	
	\State $R = \contractvsglue(G,S)$
	\State \Return $R$
\EndFunction
\end{algorithmic}
\end{algorithm}
\subsection{Computing an economical bridgeless $2$-edge-cover}
\label{subsex:main:bridge-covering}
Given a structured graph $G$, we first compute an economical bridgeless $2$-edge-cover of it. Before we define an economical bridgeless $2$-edge-cover, we need to first define {\bf small, medium,} and {\bf large}, which is used to categorize a $2$-edge-connected subgraph of $G$ based on its weight.

\begin{definition}[small, medium, large]For a weighted graph $G$, we call a $2$-edge-connected subgraph $H$ of $G$ {\bf small} if $||H||\leq 2$, {\bf medium} if $||H||=3$, and {\bf large} if $||H||\geq 4$.
\end{definition}
Note that for structured graphs, the only possible small components are cycles of length $3$ or $4$ with exactly $2$ unit-edges, and medium components are cycles of length $3$, $4$, $5$, or $6$ with exactly~$3$~unit-edges.

\begin{definition}
A bridgeless 2-edge-cover $H$ of a graph $G$ is {\bf economical} if all the zero-edges of~$G$ are in $H$, $||H|| \leq \frac{13}{8} \cdot ||D_2(G)|| - 2 n_\ell - \frac{15}{8} n_m - \frac{5}{4} n_s$, where $n_\ell, n_m$, and $n_s$ are the number of large, medium, and small components of $H$, respectively. Furthermore, there exists a $D_2$ of $G$ such that each small component of $H$ is a small component of the $D_2$.
\end{definition}

Our main result for this subsection is as follows.

\begin{restatable}{theorem}{thmbridgecover}
\label{thm:bridgeCover}
Given a structured graph $G$, we can compute an economical bridgeless $2$-edge-cover of $G$ in polynomial time.
\end{restatable}

To compute an economical bridgeless $2$-edge-cover of $G$, we first find a $D_2$ of $G$ and include all the zero-edges in it. Next, we transform this $D_2$ into a `canonical' $D_2$, which is defined in Appendix~\ref{sec:bridgeCovering}.
Then, we cover the bridges of the canonical $D_2$ to get an `economical'  bridgeless $2$-edge-cover. 
All of this can be done in polynomial time. The details 
with 
the proof of Theorem~\ref{thm:bridgeCover} are 
in Appendix~\ref{sec:bridgeCovering}.


\subsection{Computing a special configuration}
\label{subsec:main:computing-special}
Next, given an economical bridgeless $2$-edge cover $H$ of a structured graph $G$, we compute a `special' configuration of $G$.
A special configuration is a bridgeless $2$-edge-cover that satisfies certain additional properties. In particular, it does not contain any medium components.

\begin{definition}[Special configuration]
Given a structured graph $G$, we say $H$ is a special configuration of $G$ if 
\begin{itemize}
    \item[(i)] $H$ is an economical bridgeless $2$-edge-cover of $G$,
    \item[(ii)] $H$ does not contain medium-size components,
    \item[(iii)] $G/H$ does not contain good cycles
    \item[(iv)] $G/H$ does not contain open $3$-augmenting paths, and
    \item[(v)] $H$ does not contain a small to medium merge or small to large merge.
\end{itemize}
\end{definition}

The terms and notation used in conditions (iii)-(v) are formally defined in Appendix~\ref{sec:ComputingSpecialConfiguration}.
Without going into details, we briefly describe the structures defined in (iii)-(v).
A good cycle is a simple cycle $C$ in $G/H$ that contains either a) two large components, b) one large component and one small component containing a shortcut, or c) two small components each containing a shortcut.
Here, we say a small component $S$ is shortcut w.r.t.\ $C$ if in $G[V(S)]$ there exists a Hamiltonian path from $u$ to $v$ of weight $1$, where $u$ and $v$ are the vertices incident to $C$ when $S$ is expanded.
Hence, there is a unit-edge in the small component $S$ that is redundant for the 2-edge-connectivity of $H$ after we add the edges of $C$ to it.
An open $3$-augmenting path is a simple path $P$ in $G/H$ through 4 small components such that for each of the two interior small components there is a shortcut w.r.t.\ $P$.
Finally, a small to medium merge (or small to large merge) is a set of 3 small components $S_1, S_2, S_3 \in H$ such that in $G' = G[V(S_1 \cup S_2 \cup S_3)]$ there exists a set of edges that form two medium components (or one large component) of weight precisely 6 spanning $V(G')$.

Essentially, these conditions restrict the structure of special configurations.
For example, in the graph $G/H$ (the graph obtained by contracting the various components of $H$ into single nodes), there is no cycle that has $2$ or more nodes corresponding to large components.
In particular, a special configuration contains at least one small component or is already feasible. The restricted structure of special configurations will be crucially exploited while proving an improved lower bound on $\OPT(G)$.

From an economical bridgeless 2-edge-cover $H$, we obtain a special configuration by repeatedly searching for the four forbidden structures (properties (ii)-(v) above) and buying and selling certain edges such that we turn $H$ into an economical bridgeless 2-edge cover $H'$ with fewer components.
One can show that searching for such structures can be done in polynomial time.

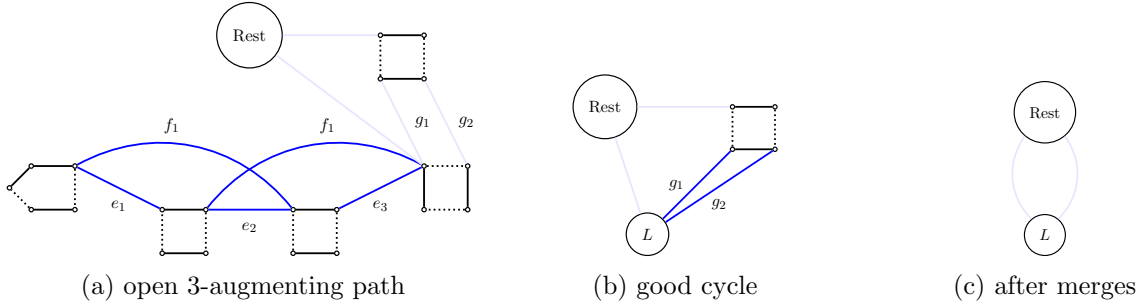
\begin{figure}[t]
 \begin{subfigure}[b]{0.39\textwidth}
        \centering
        \resizebox{\linewidth}{!}{
        \begin{tikzpicture}
	
				\tikzstyle{hvertex}=[thick,circle,inner sep=0.cm, minimum size=1mm, fill=white]
				\tikzstyle{overtex}=[thick,circle,inner sep=0.cm, minimum size=1mm, fill=white, draw=black]
				\tikzstyle{lvertex}=[thick,circle,inner sep=0.cm, minimum size=15mm, fill=white, draw=black]

				\node[overtex] (v11) at (0.5,1.5) {};
				\node[overtex] (v12) at (1,2) {};
				\node[overtex] (v13) at (2,2) {};
				\node[overtex] (v14) at (2,1) {};
				\node[overtex] (v15) at (1,1) {};
				
				\node[overtex] (v21) at (4,0) {};
				\node[overtex] (v22) at (4,1) {};
				\node[overtex] (v23) at (5,1) {};
				\node[overtex] (v24) at (5,0) {};
				
				\node[overtex] (v31) at (7,0) {};
				\node[overtex] (v32) at (7,1) {};
				\node[overtex] (v33) at (8,1) {};
				\node[overtex] (v34) at (8,0) {};

				\node[overtex] (v41) at (10,1) {};
				\node[overtex] (v42) at (10,2) {};
				\node[overtex] (v43) at (11,2) {};
				\node[overtex] (v44) at (11,1) {};
				
				\node[overtex] (v51) at (9,4) {};
				\node[overtex] (v52) at (9,5) {};
				\node[overtex] (v53) at (10,5) {};
				\node[overtex] (v54) at (10,4) {};
				
				\node[lvertex] (l1) at (6,5) {Rest};
				
				\draw[edge, - , black] (v11) to  (v12);
				\draw[edge, - , black] (v12) to  (v13);
				\draw[edge, - , black, dotted] (v13) to  (v14);
				\draw[edge, - , black] (v14) to  (v15);
				\draw[edge, - , black, dotted] (v15) to  (v11);
				
				\draw[edge, - , black, dotted] (v21) to  (v22);
				\draw[edge, - , black] (v22) to  (v23);
				\draw[edge, - , black, dotted] (v23) to  (v24);
				\draw[edge, - , black] (v24) to  (v21);
				
				\draw[edge, - , black, dotted] (v31) to  (v32);
				\draw[edge, - , black] (v32) to  (v33);
				\draw[edge, - , black, dotted] (v33) to  (v34);
				\draw[edge, - , black] (v34) to  (v31);
				
				\draw[edge, - , black] (v41) to  (v42);
				\draw[edge, - , black, dotted] (v42) to  (v43);
				\draw[edge, - , black] (v43) to  (v44);
				\draw[edge, - , black, dotted] (v44) to  (v41);
				
				\draw[edge, - , black, dotted] (v51) to  (v52);
				\draw[edge, - , black] (v52) to  (v53);
				\draw[edge, - , black, dotted] (v53) to  (v54);
				\draw[edge, - , black] (v54) to  (v51);
				
				\draw[edge, - , blue] (v13) to node[label= below:\textcolor{black}{$e_1$}] {} (v22);
				\draw[edge, - , blue] (v23) to node[label= below:\textcolor{black}{$e_2$}] {} (v32);
				\draw[edge, - , blue] (v33) to node[label= below:\textcolor{black}{$e_3$}] {} (v42);
				
				\draw[edge, - , blue] (v13) to [bend left=40]  node[label= above left:\textcolor{black}{$f_1$}] {} (v32);
				\draw[edge, - , blue] (v42) to [bend right=40]  node[label= above right:\textcolor{black}{$f_1$}] {} (v23);
				\draw[edge, - , blue!10!white] (l1) to (v42);
				\draw[edge, - , blue!10!white] (l1) to (v52);
				\draw[edge, - , blue!10!white] (v51) to node[label= right:\textcolor{black}{$g_1$}] {} (v42);
				\draw[edge, - , blue!10!white] (v54) to node[label= right:\textcolor{black}{$g_2$}] {} (v43);
			\end{tikzpicture}
			}
        \caption{open $3$-augmenting path}
        \label{fig:ipco:1-1}
    \end{subfigure}
  \begin{subfigure}[b]{0.29\textwidth}
    \centering
        \resizebox{.6\linewidth}{!}{
            \begin{tikzpicture}
            
				\tikzstyle{hvertex}=[thick,circle,inner sep=0.cm, minimum size=1mm, fill=white]
				\tikzstyle{overtex}=[thick,circle,inner sep=0.cm, minimum size=1mm, fill=white, draw=black]
				\tikzstyle{lvertex}=[thick,circle,inner sep=0.cm, minimum size=10mm, fill=white, draw=black]
				\tikzstyle{rvertex}=[thick,circle,inner sep=0.cm, minimum size=15mm, fill=white, draw=black]

				\node[lvertex] (l1) at (1,0) {$L$};
				\node[rvertex] (r1) at (0,3) {Rest};
				
				\node[overtex] (v11) at (3,2) {};
				\node[overtex] (v12) at (3,3) {};
				\node[overtex] (v13) at (4,3) {};
				\node[overtex] (v14) at (4,2) {};
				
				\draw[edge, - , black, dotted] (v11) to  (v12);
				\draw[edge, - , black] (v12) to  (v13);
				\draw[edge, - , black, dotted] (v13) to  (v14);
				\draw[edge, - , black] (v14) to  (v11);
				
				\draw[edge, - , blue] (v11) to node[label= left:\textcolor{black}{$g_1$}] {} (l1);
				\draw[edge, - , blue] (v14) to node[label= below:\textcolor{black}{$g_2$}] {} (l1);
				
				\draw[edge, - , blue!10!white] (l1) to (r1);
				\draw[edge, - , blue!10!white] (r1) to (v12);
			\end{tikzpicture}
			}
        \caption{good cycle}   
        \label{fig:ipco:1-2}
    \end{subfigure}
    \begin{subfigure}[b]{.29\textwidth}
        \centering
        \resizebox{.25\linewidth}{!}{
            \begin{tikzpicture}
				\tikzstyle{hvertex}=[thick,circle,inner sep=0.cm, minimum size=1mm, fill=white]
				\tikzstyle{overtex}=[thick,circle,inner sep=0.cm, minimum size=1mm, fill=white, draw=black]
				\tikzstyle{lvertex}=[thick,circle,inner sep=0.cm, minimum size=10mm, fill=white, draw=black]
				\tikzstyle{rvertex}=[thick,circle,inner sep=0.cm, minimum size=15mm, fill=white, draw=black]

				\node[lvertex] (l1) at (0,0) {$L$};
				\node[rvertex] (r1) at (0,3) {Rest};
				
				\draw[edge, - , blue!10!white] (l1) to [bend left=40] (r1);
				\draw[edge, - , blue!10!white] (l1) to [bend right=40] (r1);
			\end{tikzpicture}
			}
		 \caption{after merges}
        \label{fig:ipco:1-3}
    \end{subfigure}
\caption{An example of obtaining a special configuration.} 
\label{fig:ipco:1}
\end{figure} 

We briefly explain this process by an example: Figure~\ref{fig:ipco:1-1}. The black edges correspond to the economical bridgeless $2$-edge-cover $H$, where the dotted edges are of weight $0$. The (bold and faint) blue edges are edges of $G$ that are not in $H$. The $3$ blue edges $e_1$, $e_2$, and $e_3$ form an open $3$-augmenting path in $G/H$ (as it can `shortcut' the two black unit edges adjacent to $e_2$).
By the properties of structured graphs, one can show that the blue edges $f_1$ and $f_2$ must exist, and hence four  components of $H$ can be merged into one large component by buying all the 5 bold blue edges and selling the 2 black unit edges adjacent to $e_2$, which are `shortcut' by the open 3-augmenting path, to obtain Figure~\ref{fig:ipco:1-2}; as initially the $4$ components incident to the blue edges have a credit of $9\times \frac 5 8 \geq 5$, we buy $5$ blue and sell $2$ black edges to have a final credit of at least $2$. 
One can show that this step is tight for our analysis with an approximation ratio of $13/8$.
Now, in Figure~\ref{fig:ipco:1-2} the blue edges $g_1$ and $g_2$ form a good cycle in $G/H$ (as it can be merged into a single large component). Initially, the credits in the large and small components that are part of this good cycle have a credit of $2+2\times \frac 5 8\geq 3$. We buy the edges $g_1$ and $g_2$ and sell the unit edge adjacent to both $g_1$ and $g_2$ to form a single $2$-edge-connected component having a credit of at least $3-2+1=2$, and thus the good cycle merges into a large component as shown in Figure~\ref{fig:ipco:1-3}.
In general, we show the following theorem,
which is proved in Appendix~\ref{sec:ComputingSpecialConfiguration}.

\begin{restatable}{theorem}{thmspecialconfig}
\label{thm:specialConfig}
Given a structured graph $G$ and an economical bridgeless $2$-edge-cover of it, we can compute a special configuration of $G$ in polynomial time.
\end{restatable}

\subsection{Two-edge-connecting special configurations}
\label{subsec:main:feasible-for-special}
Finally, we present the last part of our algorithm, which we call `$\contractvsglue$', that converts a special configuration into a $2$-edge-connected graph. Recall, a special configuration is an economical bridgeless $2$-edge-cover of a structured graph that contains only small and large components and satisfies certain additional properties. Our algorithm computes two solutions and returns the one with a lower weight. The first solution is obtained by contracting the small components into single nodes and recursively computing the solution on the contracted graph (this is done by calling $\reducea$ on the contracted graph) and then adding the edges in the small components to the solution after expanding it back. The second solution is obtained by following the `Gluing Algorithm' of Cheriyan et al.~\cite{CheriyanCDZ21}, which we call `$\glue$', and reproduce it below for completeness' sake.

\begin{algorithm}
\caption{$\contractvsglue$}\label{alg:contractvsglue}
\begin{algorithmic} 
\Function{Contract-vs-glue}{$G$, $S$} 
\Comment{\textcolor{gray}{$G$ is structured, $S$ is a special configuration of $G$}}
\If {$S$ is a $\ecss$ of $G$}
\Return{$S$}
\EndIf
\State Let $H_1, \cdots, H_k$ be the small components of $S$. \Comment{\textcolor{gray}{now $S$ must have small components}}
\State $G_1^* = \reducea(G/\{H_1,\cdots,H_k\})$
\State $S_1 = (V(G), E(\hat{G_1^*})\cup \bigcup_{i\in[k]} E(H_i))$

\State $S_2 = \glue(G,S)$

\Return{$\argmin \{||S_1||,||S_2||\}$}

\EndFunction
\end{algorithmic}
\end{algorithm}

We will be using the following lemmas to prove our main result.
\begin{lemma}\label{lem:simultaneousContract}
Let $G$ be a structured graph, $S$ be a special configuration of $G$ with small components $H_1,\cdots,H_k$, and let $G_1^*=\reducea(G/\{H_1,\cdots,H_k\})$.
If $G_1^*$ is a $2$-edge-connected spanning subgraph of $G/\{H_1,\cdots,H_k\}$, then $(V(G),E(\hat G_1^*)\cup \bigcup_{i\in[k]}E(H_i))$ is a $2$-edge-connected spanning subgraph of~$G$. 
Furthermore,
if $H_1,\cdots,H_k$ is $(\frac {13}8,4,k)$-contractible in $G$ and  $||G_1^*||\leq f(G/\{H_1,\cdots,H_k\})$, then $||(V(G), E(\hat G_1^*)\cup \bigcup_{i\in[k]}E(H_i))||\leq f(G)$.
\end{lemma}
The first statement in the above lemma is straightforward to see. 
The second part is obtained by specializing Lemma~33 of Appendix~\ref{sec:preprocessing} to our parameters.

The following lemma is proved implicitly in~\cite{CheriyanCDZ21}. 
For completeness, we give a full proof in Appendix~\ref{sec:Glue}.

\begin{restatable}{lemma}{lemglue}
\label{lem:glue}
Let $G$ be a structured graph and $S$ be a special configuration of $G$ with $n_\ell$ large and $n_s$ small components. Then, $\glue(G,S)$ is a $2$-edge-connected spanning subgraph of $G$ with
$||\glue(G,S)|| \leq ||S|| + 2 n_\ell + \frac 4 3 n_s - 2$.
\end{restatable}

Also, we prove the following lower bound result. Informally, if a $\ecss$ includes $t$ fewer edges from within the small components of a special configuration, it must include $t(1+\frac 1{12})$ edges going between different small components.
The proof is given in Appendix~\ref{sec:lower-bound}.

\begin{restatable}{lemma}{lemlowerbound}
\label{lem:lowerbound}
Let $G$ be a structured graph and $S$ be a special configuration of $G$ with $k$ small components: $H_1,\cdots,H_k$.
Let $R$ be any $2$-edge-connected spanning subgraph of $G$ such that $2k - \sum_{i\in[k]}||R[V(H_i)]||=t$, then $\sum_{i<j\leq k}e_R(V(H_i),V(H_j)) \geq (1 + \frac {1}{12})t$, where $e_R(A,B)$ represents the number of unit-edges going between vertex sets $A$ and $B$ in $R$.
\end{restatable}

Here, we give an intuition on how to prove Lemma~\ref{lem:lowerbound}.
Fix a structured graph $G$ together with a special configuration $S$ and a $2$-edge-connected spanning subgraph $R$ as specified in Lemma~\ref{lem:lowerbound}.
In order to simplify things, here we assume that each small component $H_i$ of $S$ is a cycle of length $4$ such that $G$ does not contain any of the diagonals of $H_i$.
Furthermore, let $H_1, H_2, \dots, H_\ell$ be the set of small components in~$S$.

An edge between two 
small components is called \textbf{crossing}, whereas an edge inside a small component is called \textbf{inside}.
Informally speaking, Lemma~\ref{lem:lowerbound} states that, on average, for each inside edge $e \in E(S)$ of some small component of $S$ that is not present in $R$, $E(R)$ has to contain at least $1 + \frac{1}{12}$ crossing edges.
First, one can show that the vertices incident to an inside edge $e$ that is not present in $E(R)$ cannot be adjacent to vertices of a large component 
of $S$, as otherwise this implies that $S$ contains a good cycle
, contradicting that $S$ is special.
Hence, 
each vertex incident to an inside edge $e$ that is not present in $E(R)$ must be incident to at least one crossing edge in $R$. 

In order to show that $R$ contains sufficiently many crossing edges, we define an assignment $\xi$ that distributes for each inside or crossing edge of $R$ a total charge of one to the small components $H_1, \dots, H_\ell$.
The sum over all charges of edges incident to some component~$H_i$ then defines the \emph{load} of the component~$H_i$. Note that, by construction, the total load over all small components is equal to the number of inside and crossing edges in $R$.

Each inside edge contributes one to the charge of its component, while each crossing edge  distributes a charge of one to the components incident to it:
if only one of the two unit edges of $E(S)$ adjacent to a crossing edge is shortcut (absent) in $R$, then the component with the shortcut edge receives a charge of one from that crossing edge. 
Otherwise, both incident components receive a charge of $\frac 1 2$.
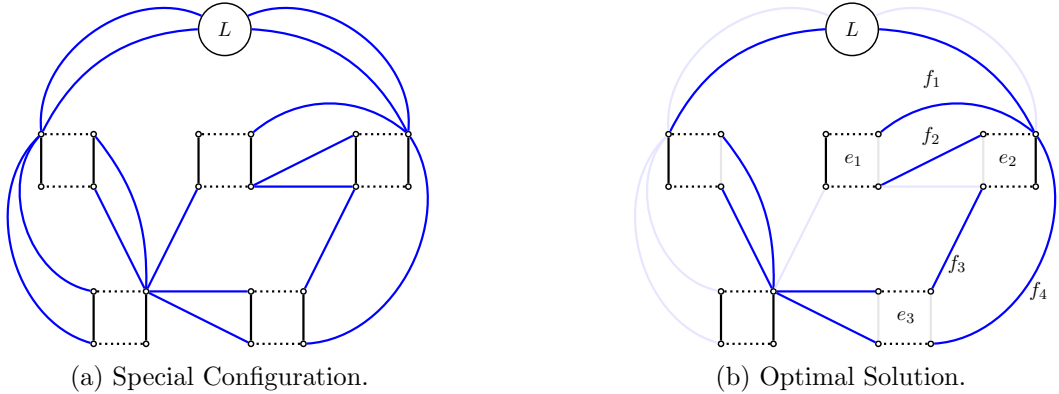
\begin{figure}[t]
	\subcaptionbox{\centering Special Configuration.\label{fig:ipco:2-1}}[0.5\textwidth]
	{
	
			 \resizebox{0.4\textwidth}{!}{
			\begin{tikzpicture}[rotate=-90]

				\tikzstyle{hvertex}=[thick,circle,inner sep=0.cm, minimum size=1mm, fill=white]
				\tikzstyle{overtex}=[thick,circle,inner sep=0.cm, minimum size=1mm, fill=white, draw=black]
				\tikzstyle{lvertex}=[thick,circle,inner sep=0.cm, minimum size=10mm, fill=white, draw=black]

				\node[overtex] (v11) at (2,0) {};
				\node[overtex] (v12) at (2,1) {};
				\node[overtex] (v13) at (3,1) {};
				\node[overtex] (v14) at (3,0) {};
				
				\node[overtex] (v21) at (2,3) {};
				\node[overtex] (v22) at (2,4) {};
				\node[overtex] (v23) at (3,4) {};
				\node[overtex] (v24) at (3,3) {};
				
				\node[overtex] (v31) at (2,6) {};
				\node[overtex] (v32) at (2,7) {};
				\node[overtex] (v33) at (3,7) {};
				\node[overtex] (v34) at (3,6) {};

				\node[overtex] (v41) at (5,1) {};
				\node[overtex] (v42) at (5,2) {};
				\node[overtex] (v43) at (6,2) {};
				\node[overtex] (v44) at (6,1) {};
				
				\node[overtex] (v51) at (5,4) {};
				\node[overtex] (v52) at (5,5) {};
				\node[overtex] (v53) at (6,5) {};
				\node[overtex] (v54) at (6,4) {};
				
				\node[lvertex] (l1) at (0,3.5) {$L$};
				
				\draw[edge, - , black, dotted] (v11) to  (v12);
				\draw[edge, - , black] (v12) to  (v13);
				\draw[edge, - , black, dotted] (v13) to  (v14);
				\draw[edge, - , black] (v14) to  (v11);
				
				\draw[edge, - , black, dotted] (v21) to  (v22);
				\draw[edge, - , black] (v22) to  (v23);
				\draw[edge, - , black, dotted] (v23) to  (v24);
				\draw[edge, - , black] (v24) to  (v21);
				
				\draw[edge, - , black, dotted] (v31) to  (v32);
				\draw[edge, - , black] (v32) to  (v33);
				\draw[edge, - , black, dotted] (v33) to  (v34);
				\draw[edge, - , black] (v34) to  (v31);
				
				\draw[edge, - , black, dotted] (v41) to  (v42);
				\draw[edge, - , black] (v42) to  (v43);
				\draw[edge, - , black, dotted] (v43) to  (v44);
				\draw[edge, - , black] (v44) to  (v41);
				
				\draw[edge, - , black, dotted] (v51) to  (v52);
				\draw[edge, - , black] (v52) to  (v53);
				\draw[edge, - , black, dotted] (v53) to  (v54);
				\draw[edge, - , black] (v54) to  (v51);

				\draw[edge, - , blue] (l1) to [bend right=30] (v11);
				\draw[edge, - , blue] (l1) to [bend right=60] (v11);
				\draw[edge, - , blue] (l1) to [bend left=30] (v32);
				\draw[edge, - , blue] (l1) to [bend left=60] (v32);
				
				\draw[edge, - , blue] (v12) to [bend left=20] (v42);
				\draw[edge, - , blue] (v13) to (v42);
				\draw[edge, - , blue] (v11) to [bend right=60] (v41);
				\draw[edge, - , blue] (v11) to [bend right=60] (v44);
				
				\draw[edge, - , blue] (v24) to (v42);
				\draw[edge, - , blue] (v23) to (v31);
				\draw[edge, - , blue] (v22) to [bend left=40] (v32);
				\draw[edge, - , blue] (v23) to (v34);
				
				\draw[edge, - , blue] (v32) to [bend left=60] (v53);
				\draw[edge, - , blue] (v34) to (v52);
				
				\draw[edge, - , blue] (v42) to (v51);
				\draw[edge, - , blue] (v42) to (v54);
			\end{tikzpicture}
			
			}
		}
%
		\subcaptionbox{\centering Optimal Solution. \label{fig:ipco:2-2}}[0.5\linewidth]
		{
		\resizebox{0.4\textwidth}{!}{
			\begin{tikzpicture}[rotate=-90]
				\tikzstyle{hvertex}=[thick,circle,inner sep=0.cm, minimum size=1mm, fill=white]
				\tikzstyle{overtex}=[thick,circle,inner sep=0.cm, minimum size=1mm, fill=white, draw=black]
				\tikzstyle{lvertex}=[thick,circle,inner sep=0.cm, minimum size=10mm, fill=white, draw=black]

				\node[overtex] (v11) at (2,0) {};
				\node[overtex] (v12) at (2,1) {};
				\node[overtex] (v13) at (3,1) {};
				\node[overtex] (v14) at (3,0) {};
				
				\node[overtex] (v21) at (2,3) {};
				\node[overtex] (v22) at (2,4) {};
				\node[overtex] (v23) at (3,4) {};
				\node[overtex] (v24) at (3,3) {};
				
				\node[overtex] (v31) at (2,6) {};
				\node[overtex] (v32) at (2,7) {};
				\node[overtex] (v33) at (3,7) {};
				\node[overtex] (v34) at (3,6) {};

				\node[overtex] (v41) at (5,1) {};
				\node[overtex] (v42) at (5,2) {};
				\node[overtex] (v43) at (6,2) {};
				\node[overtex] (v44) at (6,1) {};
				
				\node[overtex] (v51) at (5,4) {};
				\node[overtex] (v52) at (5,5) {};
				\node[overtex] (v53) at (6,5) {};
				\node[overtex] (v54) at (6,4) {};
				
				\node[lvertex] (l1) at (0,3.5) {$L$};
				
				\draw[edge, - , black, dotted] (v11) to  (v12);
				\draw[edge, - , black!10!white] (v12) to  (v13);
				\draw[edge, - , black, dotted] (v13) to  (v14);
				\draw[edge, - , black] (v14) to  (v11);
				
				\draw[edge, - , black, dotted] (v21) to  (v22);
				\draw[edge, - , black!10!white] (v22) to node[label= left:\textcolor{black}{$e_1$}] {} (v23);
				\draw[edge, - , black, dotted] (v23) to (v24);
				\draw[edge, - , black] (v24) to  (v21);
				
				\draw[edge, - , black, dotted] (v31) to  (v32);
				\draw[edge, - , black] (v32) to  (v33);
				\draw[edge, - , black, dotted] (v33) to  (v34);
				\draw[edge, - , black!10!white] (v34) to node[label= right:\textcolor{black}{$e_2$}] {} (v31);
				
				\draw[edge, - , black, dotted] (v41) to  (v42);
				\draw[edge, - , black] (v42) to  (v43);
				\draw[edge, - , black, dotted] (v43) to  (v44);
				\draw[edge, - , black] (v44) to  (v41);
				
				\draw[edge, - , black, dotted] (v51) to  (v52);
				\draw[edge, - , black!10!white] (v52) to  node[label= left:\textcolor{black}{$e_3$}] {} (v53);
				\draw[edge, - , black, dotted] (v53) to  (v54);
				\draw[edge, - , black!10!white] (v54) to  (v51);

				\draw[edge, - , blue] (l1) to [bend right=30] (v11);
				\draw[edge, - , blue!10!white] (l1) to [bend right=60] (v11);
				\draw[edge, - , blue] (l1) to [bend left=30] (v32);
				\draw[edge, - , blue!10!white] (l1) to [bend left=60] (v32);
				
				\draw[edge, - , blue] (v12) to [bend left=20] (v42);
				\draw[edge, - , blue] (v13) to (v42);
				\draw[edge, - , blue!10!white] (v11) to [bend right=60] (v41);
				\draw[edge, - , blue!10!white] (v11) to [bend right=60] (v44);
				
				\draw[edge, - , blue!10!white] (v24) to (v42);
				\draw[edge, - , blue] (v23) to node[label= above:\textcolor{black}{$f_2$}] {} (v31);
				\draw[edge, - , blue] (v22) to [bend left=40]  node[label= above left:\textcolor{black}{$f_1$}] {} (v32);
				\draw[edge, - , blue!10!white] (v23) to (v34);
				
				\draw[edge, - , blue] (v32) to [bend left=60] node[label= below:\textcolor{black}{$f_4$}] {}  (v53);
				\draw[edge, - , blue] (v34) to node[label= below:\textcolor{black}{$f_3$}] {}  (v52);
				
				\draw[edge, - , blue] (v42) to (v51);
				\draw[edge, - , blue] (v42) to (v54);
			\end{tikzpicture}
			
			}
		}
		\caption{Example: special configuration, optimal solution, and lower bound.}
		\label{fig:ipco:2}
		
	\end{figure} 
Consider for example Figure~\ref{fig:ipco:2-2}, where the bold edges represent $R$.
By the above assignment, the components incident to $f_3$ receive a charge of $\frac 1 2$ each from $f_3$, while only the component containing $e_3$ receives a charge of $1$ from $f_4$. 
The total load of the components (from left to right, top to bottom) then is $3,\frac 5 2, 2, 2$, and $\frac 7 2$, respectively.

From our assignment, one can easily argue that the load of each small component is $\geq 2$.
Furthermore, if there are no two shortcut edges that are adjacent to a crossing edge, 
then it clearly follows that Lemma~\ref{lem:lowerbound} holds.
In fact, in this case, we could replace the $1 + \frac{1}{12}$ by $2$ in the lemma -- a much stronger result.

Hence, we may assume that there are some shortcuts that share crossing edges, e.g.\ edges $e_1, e_2,$ and $e_3$ in Figure~\ref{fig:ipco:2-2}.
However, in this case, the edges $f_2$ and $f_3$ (which form an {\it open $2$-augmenting path}) cannot be extended to an open $3$-augmenting path 
(since $S$ is special); the edges $f_1$ and $f_4$ have to go back to the component containing $e_2$.
One can show that in this case (since there are also no good cycles or local merges), the average load of the components containing $e_1$, $e_2$, and $e_3$ is at least $\frac 5 2$. 
In the remaining case when there are no open 2-augmenting paths in $R$, using a similar argument we can also show that the average load of a component is at least $2 + \frac{1}{6}$.
This load assignment then implies the statement of Lemma~\ref{lem:lowerbound}.

\section{$\alg$ is admissible}\label{sec:mainTheoremProof}
As noted earlier, from Theorems~\ref{thm:reduce} and \ref{thm:runningtime}, it follows that if we can show $\alg$ is admissible, Theorem~\ref{thm:main} follows. Thus, we will now focus on proving that $\alg$ is admissible.
\begin{lemma}\label{lem:admissible}
$\alg$ is admissible.
\end{lemma}

Before we proceed with the proof, we develop some key definitions and propositions that will be used in the proof. Throughout this subsection, $G$ is a  structured graph and $S$ is a special configuration of $G$ with small components $H_1,\cdots,H_{n_s}$.

\begin{definition}[simultaneously-contractible]
We say $S$ is $\frac {13}8$-simultaneously contractible if the small components of $S$ are $(\frac {13}8, 4, n_s)$-contractible in $G$.
\end{definition}
\begin{definition}[$\OPT^L, \OPT^R, D_2^L, D_2^R$]
We partition the vertex set of $G$ in two sets: $V(G)=L\cup R$, where $L$ consists of the vertices in the large components of the special configuration $S$ and $R$ is the set of remaining vertices, i.e., the set of vertices in the small components of $S$. Let $\OPT^L$ be the edges of $\OPT(G)$ that have at least one endpoint incident on a vertex in $L$, and $\OPT^R$ be the remaining edges of $\OPT(G)$, i.e., the edges whose both endpoints are in $R$. $D_2^L$ and $D_2^R$ are defined analogously: $D_2^L$ is the set of edges of $D_2(G)$ that are incident on at least one vertex of $L$ and $D_2^R=E(D_2(G))\setminus D_2^L$. $\opt^L, \opt^R, d_2^L$, and $d_2^R$ are defined to be $||\OPT^L||, ||\OPT^R||, ||D_2^L||$, and $||D_2^R||$, respectively.
\end{definition}

The following relationships are immediate.
\begin{proposition}\label{prop:relations}
$$||\OPT(G)||:=\opt=\opt^L + \opt^R.$$
$$||D_2(G)||:=d_2 = d_2^L + d_2^R.$$
\end{proposition}

The following proposition is key to proving our bound.
\begin{proposition}\label{prop:key}
$$\opt^L \geq d_2^L$$
\end{proposition}

\begin{proof}
Assume for contradiction $\opt^L<d_2^L$. Observe $ \OPT^L\cup D_2^R$ forms a $2$-edge-cover of $G$, since each vertex of $L$ has at least $2$ edges incident on it from $\OPT^L$ (as $\OPT$ is a feasible $2$-ECSS of $G$) and each vertex of $D_2^R$ has $2$ edges incident on it from $D_2^R$ (as $D_2^R$ are the edges of $D_2$ restricted to the small components of $S$, which were originally small in $D_2$). But $$||\OPT^L\cup D_2^R|| = \opt^L + d_2^R < d_2^L + d_2^R=d_2,$$ which contradicts the fact that $D_2$ is a minimum $2$-edge-cover of $G$.
\end{proof}

Now, we are ready to prove that $\alg$ is admissible.

\begin{proof}[Proof of Lemma~\ref{lem:admissible}]
Fix a structured graph $G$. To show $\alg$ is admissible, we need to show two properties: 
(i)~$\alg(G)$ is a $2$-edge-connected spanning subgraph of $G$ with $||\alg(G)|| \leq f(G)$ under the assumption that $\reducea(G')$ is a $2$-edge-connected spanning subgraph of $G'$ with $||\reducea(G')||\leq f(G')$ for all MAP instances $G'$ of size strictly smaller than the size of $G$, and (ii)  $T(s(G)) \leq T'(s(G)-1) + poly(s)$, where $T$ is the running time of $\alg$ and $T'$ is the running time of $\reducea$. 
Note that (ii) follows from the fact that each of the three steps in $\alg$ takes polynomial time and in the final step, namely $\contractvsglue$, $\alg$ calls the subroutine $\reducea$ only once on a smaller graph. Thus, we will focus on proving (i) below.

Note that $\alg$ on input $G$ first computes a special configuration $S$ and then applies the algorithm $\contractvsglue$ on $(G,S)$. 
If $S$ is $\ecss$, $\contractvsglue$ returns $S$, and $||S||\leq \frac{13}8 d_2-2n_\ell - \frac{5}4n_s$,  where $n_\ell=1$ is the number of large components and $n_s=0$ is the number of small components in $S$ (since $S$ is an economical bridgeless $2$-edge-cover of $G$). Thus, $||\alg||\leq f(G)$.
Otherwise, $S$ must contain at least one small component as observed in Section~\ref{subsec:main:computing-special}. 

Let $H_1,\dots,H_{n_s}$ be the small components of $S$.  
$\contractvsglue$ on $(G,S)$ computes two solutions $S_1$ and $S_2$ and returns the one with lower weight. 
Recall $S_1$ is obtained by contracting the small components of $S$, calling $\reducea$ on it, and then expanding the contracted nodes and adding back the edges of the small components. 
$S_2$ is computed by calling the $\glue(G,S)$ subroutine. 
In either case, the output is guaranteed to be a $2$-edge-connected spanning subgraph of $G$ from Lemmas~\ref{lem:simultaneousContract} and~\ref{lem:glue}.

Now, to show $||\alg(G)||\leq f(G)$,
we have two cases based on whether the special configuration $S$ is $\frac {13}8$-simultaneously contractible in $G$.
If $S$ is a $\frac{13}8$-simultaneously contractible in $G$, then by invoking Lemma~\ref{lem:simultaneousContract} (whose precondition holds since the contracted graph has size strictly smaller than $G$ and then we have the guarantee that $||\reducea(G')||\leq f(G')$ for all MAP instances $G'$), we have $||\alg(G)||\leq ||S_1||\leq f(G)$ and we are done.

In the case $S$ is not $\frac{13}8$-simultaneously contractible in  $G$, we will first lower bound $\opt$ and then upper bound $||S_2||$ to show $||\alg(G)||\leq f(G)$. 

\paragraph*{Lower bound on $\opt$.}

From Propositions~\ref{prop:relations} and~\ref{prop:key} we have
$$\opt=\opt^L + \opt^R \geq d_2^L + \opt^R.$$
We now focus on lower bounding $\opt^R$. Recall $\OPT^R$ consists of edges whose both endpoints are contained in $\bigcup_{i\in[n_s]} V(H_i)$. We categorize the edges of $\OPT^R$ into two types.
\begin{itemize}
    \item An edge of $\OPT^R$ is \textbf{inside} if both its endpoints belong to the same $V(H_i)$ for some $i$.
    \item An edge of $\OPT^R$ is \textbf{crossing} if its endpoints lie in distinct $V(H_i)$ and $V(H_j)$ for some $i\neq j$. 
\end{itemize}
Since $S$ is not $\frac{13}8$-simultaneously contractible in $G$, the number of unit-edges that are inside is at most $\frac 8{13}\cdot 2 n_s$. Let us say the number of inside edges is exactly $t$ less than the number of unit-edges in the small components of $S$, i.e., $2n_s - t$, and this number is at most $\frac 8{13}\cdot 2n_s$.

Now, to lower bound the number of unit-edges that are crossing, we invoke Lemma~\ref{lem:lowerbound}, which states that the number of unit-edges going between $V(H_i)$ and $V(H_j)$ for all $i\neq j$ is at least $(1+\frac 1{12})$t.

Thus, we have the following lower bound for $\opt$.
\begin{align*}
\opt = \opt^L + \opt^R &
    \geq d_2^L + \opt^R = d_2^L + ||\text{inside}|| + ||\text{crossing}|| \\
&\geq d_2^L + (2n_s-t) + \left(1+\frac1{12}\right)t,
\end{align*}

where $2 n_s - t \leq \frac {8}{13}\cdot 2n_s$. The 
lower bound is minimized when $t$ is kept as small as possible, i.e., when $2 n_s - t = \frac {8}{13} \cdot 2n_s$, i.e., for $t= \frac 5{13}\cdot 2n_s$. Thus,

\begin{align*}
    \opt 
    &\geq d_2^L + \frac 8{13}\cdot 2 n_s + \left(1+\frac 1{12}\right)\frac 5{13}\cdot 2 n_s 
    = d_2^L + \left(\frac {16}{13}+\frac {10}{12}\right)n_s
    = d_2^L + \frac {161}{78}\cdot n_s.
\end{align*}

\paragraph*{Upper bound on $||\alg(G)||$.}

Since $S$ is an economical bridgeless $2$-edge-cover of $G$, we have $$||S|| \leq \frac {13}8\cdot d_2 - 2n_\ell - \frac 5 4 \cdot n_s,$$ where  $n_\ell$ and $n_s$ denote the number of large and small components of $S$, respectively. Also, from Lemma~\ref{lem:glue}, we have $$||S_2||=||\glue(G,S)|| \leq ||S|| + 2 n_\ell + \frac 4 3\cdot n_s - 2.$$

Combing the two bounds we obtain $$||S_2||\leq \frac{13}8\cdot d_2 + \frac 1{12}\cdot n_s-2.$$

Now we can split $d_2$ as $d_2^L + d_2^R$, and use the fact that $d_2^R=2n_s$ to obtain our bound.

\begin{align*}
    ||\alg(G)||&\leq ||S_2|| \leq\frac{13}8\cdot d_2 + \frac 1{12}\cdot n_s-2
    =\left(\frac {13}8 \cdot d_2^L + \frac{13}8\cdot 2 n_s\right)+\frac 1{12} \cdot n_s-2\\ 
    &= \frac {13}8 \cdot d_2^L + \frac {10}{3} \cdot n_s -2
    \leq \frac{13}8\left(d_2^L + \frac{160}{78}\cdot n_s\right) -2
     \leq \frac {13}{8}\cdot \opt -2\leq f(G),
\end{align*}
where the second last inequality follows from the lower bound on $\opt$ obtained above. 
\end{proof}


\section{Conclusion}\label{sec:conclusion}

 In this work, we presented a $\frac {13}8$-approximation for MAP, which is a fundamental problem in network design.  
While several of our steps also work for smaller approximation ratios, two of our steps are tight for $\frac{13}{8}$:
 First, constructing a special configuration is tight for $\frac{13}{8}$. 
 In particular, the merge involving $3$-augmenting paths, in the worst case, uses all the available credits. 
 On the other hand, such a merge could not be avoided, as their absence from special configurations helps us improve the lower bound later.
 Furthermore, the lower bound is tight.
 Hence, simply obtaining a better construction of a special configuration is not enough as one has to improve upon the lower bound as well.
 Finally, even if one can resolve these two issues, our approximation ratio would still be tight for $1.6$ at two places:
 First, constructing a bridgeless 2-edge-cover is tight for $1.6$, even though we believe that this result can be strengthened.
 Second, in the construction of special configuration, handling the medium components is also tight for precisely $1.6$.
 Hence, also here new ideas are needed in order to obtain an approximation ratio below $1.6$.
 
Our result builds on a new $\frac 3 2$-approximation preserving reduction to instances not containing certain structures including small separators and contractible subgraphs. Furthermore, we introduced the method of simultaneous contractions and improved lower bounds to achieve our main result. These techniques seem general and applicable to other problems in network design.


\bibliography{lib.bib}

\newpage
\section*{Appendix}
\appendix 
\section{Proofs of Theorems~\ref{thm:reduce} and \ref{thm:runningtime}}\label{sec:reductions}
\label{sec:running-time}
\begin{proof}[Proof of Theorem~\ref{thm:reduce}]
Let $\alg$ be an admissible algorithm. We will show that for all MAP instances~$G$, $\reducea(G)$ is a $\ecss$ of $G$ and $||\reducea(G)||\leq f(G)$ by strong induction on $s(G)$.

Base case: $s(G)=1$ implies $V(G)=\emptyset$, and we are trivially done as $\reducea$ returns $\emptyset$ in the first if statement.
Induction hypothesis (I.H.): Fix a $t\geq 1$. We assume that the statement holds for all MAP instances $G$ with $s(G)\leq t$.
Induction step: We now prove the statement holds for all MAP instances with size $t+1$. Fix a MAP instance $G$ such that $s(G)=t+1.$

If $G$ is simple and $|V(G)| \leq 20$, from the first if statement in $\reducea$, $\reducea(G) = \OPT(G)$ and we are trivially done.

Else, if $V(G)$ consists of a forbidden configuration, let $F$ be the forbidden configuration of type $T$ detected by $\reducea$ in its second if statement. Let $(H_1,H_2,H_3)=\divideT_T(G,F)$. Since $G$ does not contain a forbidden configuration of a type that precedes $T$ (as per the algorithm), from Lemma~\ref{lem:divideCombine} parts (i) and (ii), we conclude that for each $i\in\{1,2,3\}$, $H_i$ is a MAP instance and $s(H_i) <s(G)=t+1$. Applying I.H., for each $i\in\{1,2\}$, $H_i^* = \reducea(H_i)$ is a $\ecss$ of $H_i$ and $||\reducea(H_i)|| \leq f(G)$, which enables the invocation of Lemma~\ref{lem:divideCombine} part (iii) giving us $\combine_T(G,H_1^*,H_2^*,H_3^*)$ is a $\ecss$ of $G$ and $||\combine_T(G,H_1^*,H_2^*,H_3^*)||\leq f(G)$, which is what $\reducea$ returns in the second if statement, and we are done. 

Finally, if $G$ does not contain any forbidden configuration, $\reducea$ returns $\alg(G)$. Now, from I.H. we know that the statement holds for all MAP instances with size at most $t$ which when combined with our assumption of $\alg$ being  admissible gives us that $\alg(G)$ is a $\ecss$ of $G$ such that $||\alg(G)|| \leq f(G)$ as $s(G)=t+1$, and we are done.
\end{proof}

\begin{proof}[Proof of Theorem~\ref{thm:runningtime}]
Let $\alg$ be an admissible algorithm whose running time is given by $T(s)$ for structured graphs of size $s$. Let the running time of $\reducea$ be denoted by $T'(s)$ for MAP instances of size $s$. From the definition of admissible, we have ($\star$) $T(s)\leq T'(s-1) + c_1 s^{k_1}$ for some constants $c_1$ and $k_1$.

If $\reducea$ terminates at the first return statement, it takes constant time, say at most $c_2$. If $\reducea$ terminates at the second return statement, from Lemma~\ref{lem:polytimesubs}, we know that checking for forbidden configurations and running $\divideT$ and $\combine$ takes polynomial time, and let us say that altogether $\reducea$ takes time at most $c_3s^{k_2}$ for some constants $c_3$ and $k_2$, not counting the time taken for the recursive calls.

Let $c = c_1+c_2+c_3$ and $k=1 + \max\{k_1,k_2\}$. We show $T'(s) \leq cs^k$ by strong induction on $s$. Base case: $s=1$, $T'(1)
\leq c_2 \leq c \cdot 1^k$, and we are done. Assume, that the inequality holds for all sizes at most $s-1$, and we will show that it holds for size $s$. Fix an input instance of size $s$ and we will analyze the running time of $\reducea$ on that instance in three cases depending on which of the three return statements $\reducea$ terminates at. If the $\reducea$ terminates in the first return statement, $T'(s)\leq c_2 \leq c s^k$, and we are done. If $\reducea$ terminates in the second return statement, then $T'(s) \leq  c_3s^{k_2} + T'(s_1) + T'(s_2)+T'(s_3)$ for some $s_1$, $s_2$, and $s_2$ such that  $s_1 + s_2 +s_3\leq s-1$ from $Lemma~\ref{lem:divideCombine}$ part (ii). Applying the I.H., $T'(s) \leq  c_3s^{k_2} + cs_1^k + cs_2^k + cs_3^k\leq c_3s^{k_2}+c(s_1+s_2+s_3)^k\leq c_3s^{k_2}+c (s-1)^k\leq c(s^{k_2}+(s-1)^k)\leq c(s^{k-1}+(s-1)^k) \leq cs^k$, where the last inequality follows from the fact that $(1-\frac 1 s)^k \leq 1 - \frac 1 s\implies (s-1)^k \leq s^k - s^{k-1}$. Finally, if $\reducea$ terminates at the last return statement, $T'(s)\leq c_2s^{k_2} + T(s)\leq c_2s^k + T'(s-1) + c_1s^{k_1}$ from ($\star$). Applying I.H., $T'(s)\leq c_2s^{k_2}+c(s-1)^k+c_1s^{k_1}\leq c(s-1)^k + cs^{\max\{k_1,k_2\}}= c((s-1)^k+s^{k-1})\leq cs^k$, where the last inequality follows like before.
\end{proof}
\section{Preprocessing}\label{sec:preprocessing}

In order to define structured graphs, we first need definitions of certain forbidden configurations.
Throughout the section, we assume that $\alpha \geq 1.5$ is some fixed constant.
Furthermore, we assume that $G$ is the given MAP instance.

\begin{definition}[$S_0$, previously called zero-cost $S_2$ in \cite{CheriyanCDZ21}]
Given a graph $G$, an $S_0$ of $G$ is a zero-edge $e\in E(G)$ such that deleting the endpoints of $e$ from $G$ results in two graphs with vertex sets $V_1$ and $V_2$ that are not connected to each other
\end{definition}

\begin{definition}[$S_1$, previously called unit-cost $S_2$ in \cite{CheriyanCDZ21}]
Given a simple graph $G$, an $S_1$ of $G$ is a unit-edge  $uv \in E(G)$ such that
(i) deleting the vertices $u$ and $v$ from $G$ results in two graphs with vertex sets $V_1$ and $V_2$ that are not connected to each other (ii) $\opt(G[V(G) \setminus V_2] / \{u, v\}) \geq 3$, (iii) $u$ is incident to a zero-edge which  goes to $V_2$ such that $\opt(G[V(G) \setminus V_1] / \{u, v\}) \geq 4$.
\end{definition}

\begin{figure}
\centering
\begin{subfigure}{.5\textwidth}
  \centering
  \includegraphics[width=.7\linewidth]{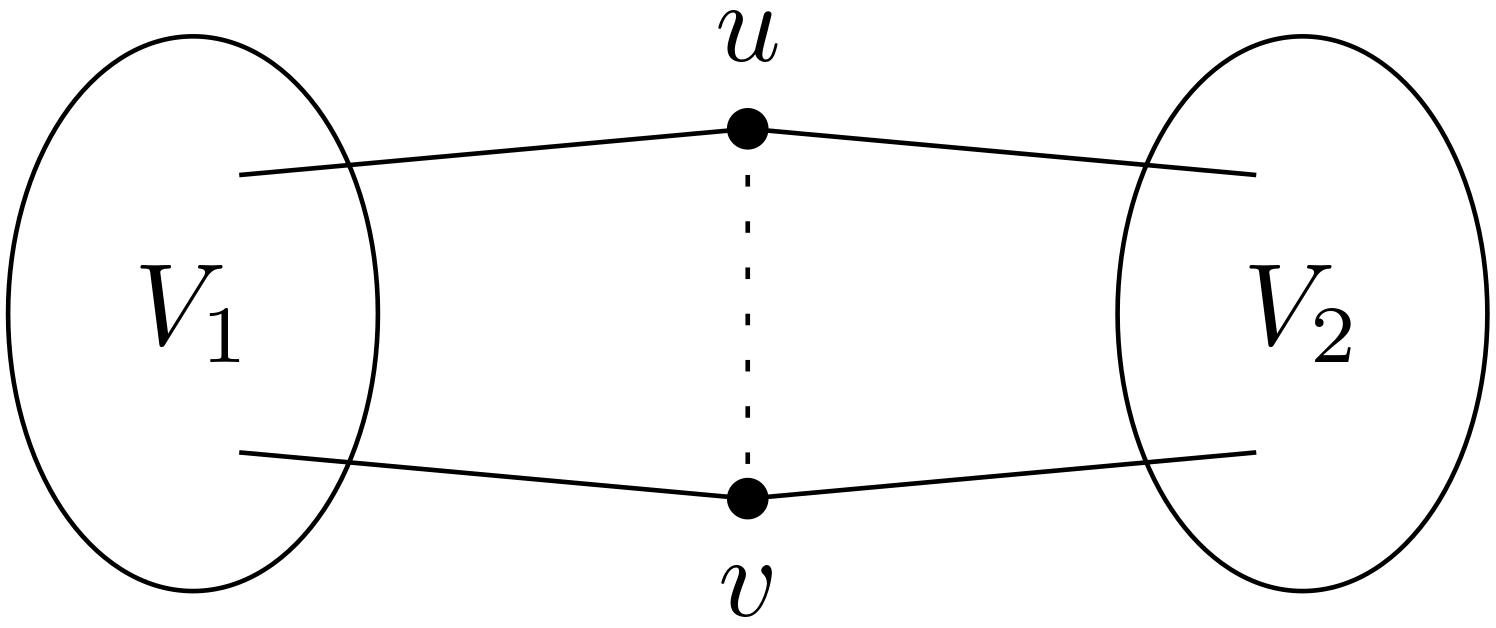}
\end{subfigure}%
\begin{subfigure}{.5\textwidth}
  \centering
  \includegraphics[width=.7\linewidth]{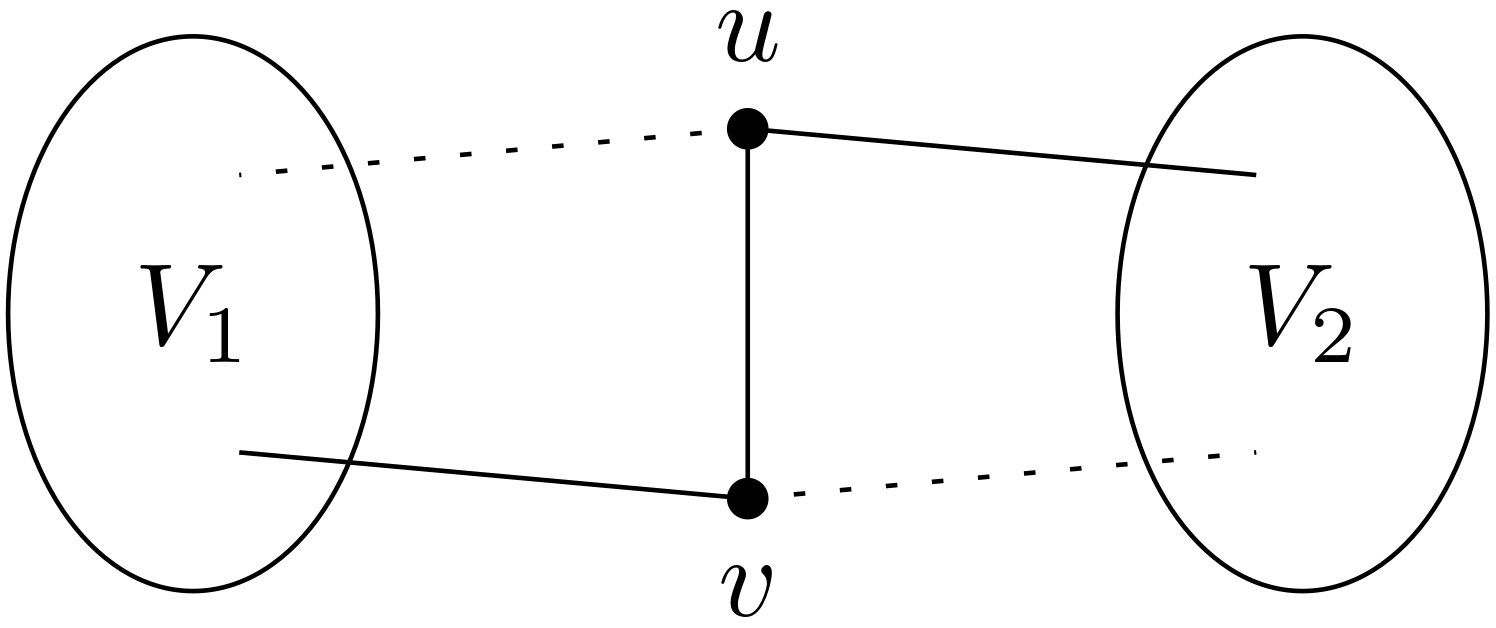}
\end{subfigure}
\caption{Illustration of an $S_0$ and $S_1$. Zero-edges are dotted, and unit edges are not dotted.}
\label{fig:S1-S2}
\end{figure}

\begin{figure}
\centering
\begin{subfigure}{.5\textwidth}
  \centering
  \includegraphics[width=.7\linewidth]{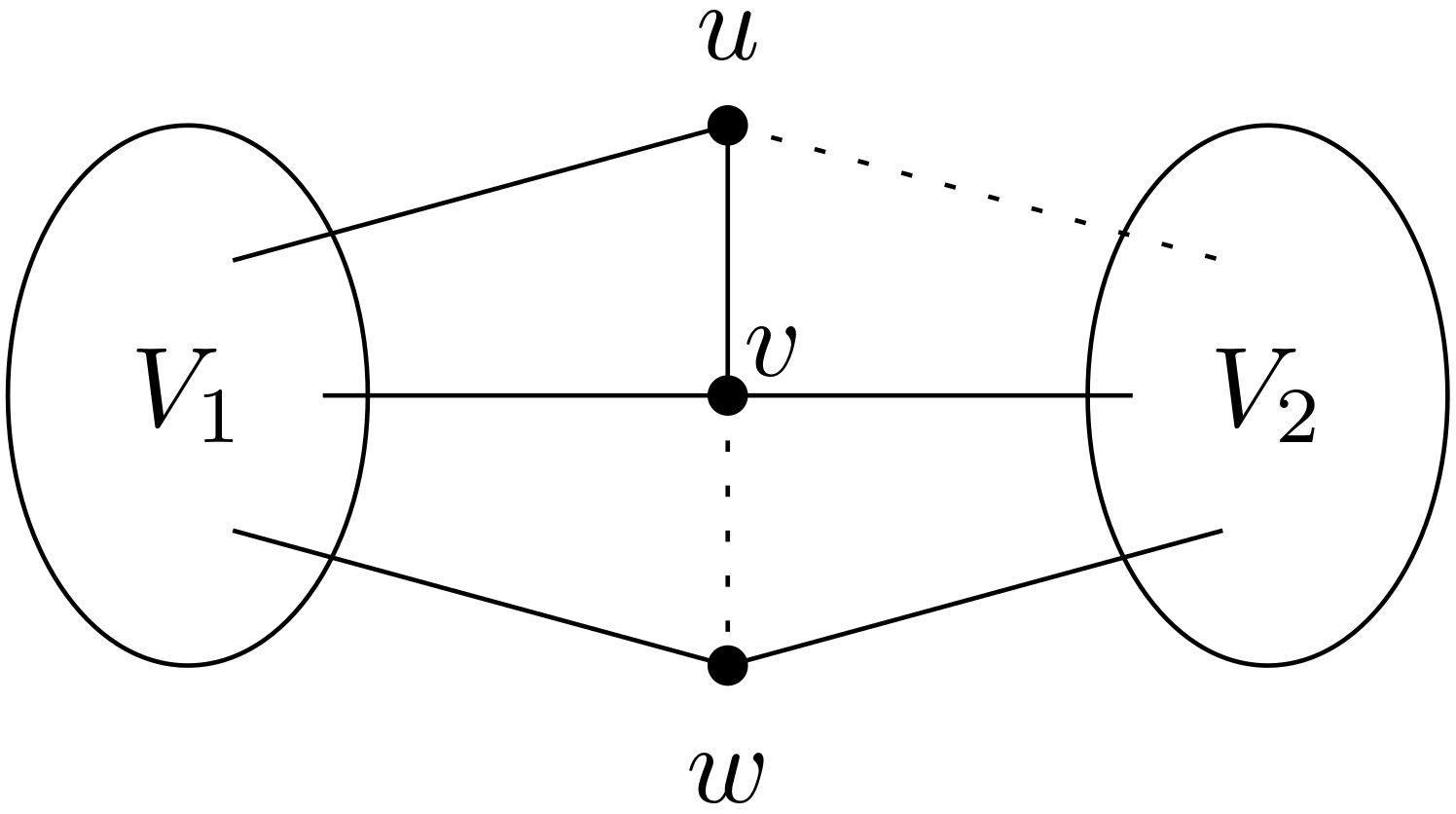}
\end{subfigure}%
\begin{subfigure}{.5\textwidth}
  \centering
  \includegraphics[width=.7\linewidth]{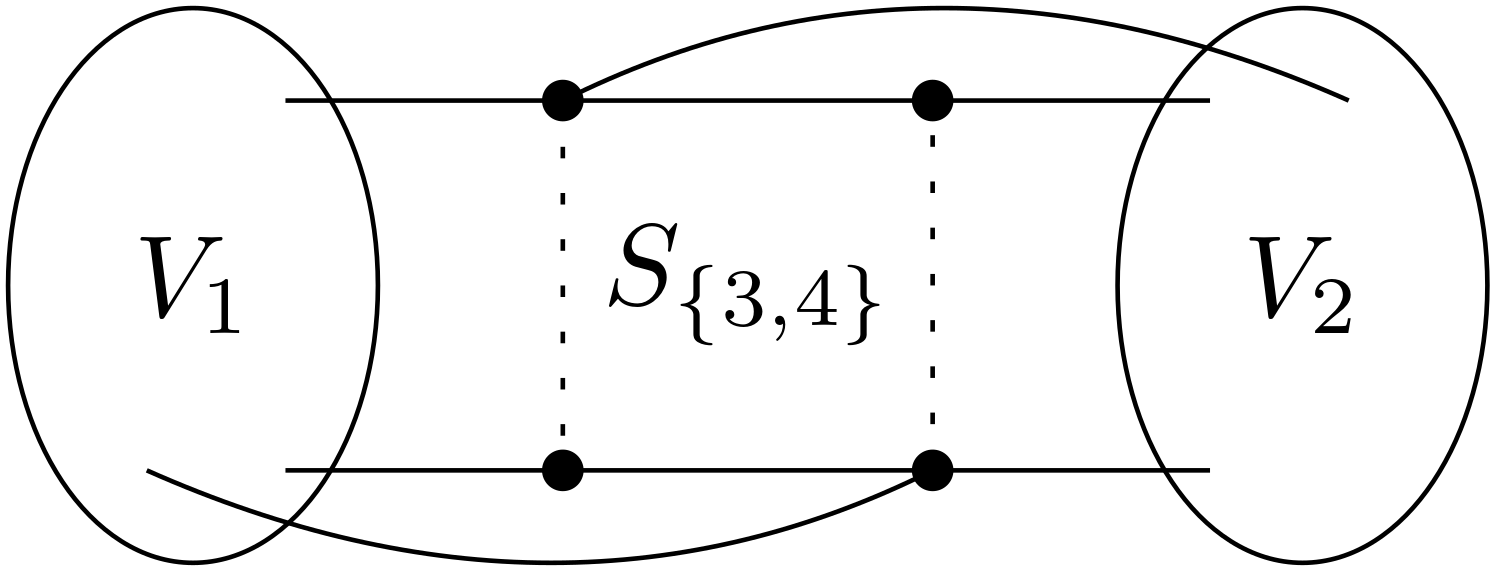}
\end{subfigure}
\caption{Illustration of an $S_2$ and $S_{\{3,4\}}$. Zero-edges are dotted, and unit edges are not dotted.}
\label{fig:S2-S34}
\end{figure}

\begin{definition}[$S_2$]
An $S_2$ is a path $P_3 = u, v, w$ on three vertices, containing exactly one unit-edge $e = \{u, v\}$, such that $G \setminus \{u, v, w \}$ consists of at least two connected components on vertex sets $V_1$ and $V_2$.
Furthermore, we have that $u$ is incident to a zero-edge connecting $u$ to some vertex of $V_2$, and $\opt(G[V(G) \setminus V_2] / \{u,v,w\}) \geq 3$ and $\opt(G[V(G) \setminus V_1] / \{u,v,w\}) \geq 4$.
\end{definition}

\begin{definition}[$S_{\{3, 4\}}$]
Given a simple graph $G$, an $S_{\{3, 4\}}$ of $G$ is a 2-vertex connected subgraph $C$ of $G$ with $|V(C)| \in \{3, 4\}$ that has a spanning cycle of cost two such that $G \setminus V(C)$ has at least 2 connected components on the vertex sets $V_1$ and $V_2$, and the cut $\delta(V(C))$ has no zero-edges; moreover, $\opt(G[V(G) \setminus V_2 / C) \geq 4$ and $\opt(G[V(G) \setminus V_1 / C) \geq 3$. 
\end{definition}

\begin{definition}[$S_3, S_4, S_5, S_6$]
Given a simple graph $G$, an $S_{k}$, $k \in \{3, 4, 5, 6\}$ of $G$ is a 2-vertex connected subgraph $C$ of $G$ with $|V(C)| = k$ that has a spanning cycle of cost 3 such that $G \setminus V(C)$ has at least 3 connected components on vertex sets $V_1$, $V_2$, and $V_3$, and the cut $\delta(V(C))$ has no zero-edges; moreover, $\opt(G[V(G) \setminus V_i / C) \geq 4$, $i \in \{1, 2, 3\}$. 
\end{definition}

\begin{definition}[$S_3', S_4', S_5', S_6'$]
Given a simple graph $G$, an $S_{k}'$, $k \in \{3, 4, 5, 6\}$ of $G$ is a 2-vertex connected subgraph $C$ of $G$ with $|V(C)| = k$ that has a spanning cycle of cost 3 such that $G \setminus V(C)$ has at least 2 connected components, and the cut $\delta(V(C))$ has no zero-edges; moreover, there is a vertex $u_1 \in V(C)$ that is only adjacent to vertices in $V(C)$, and $\opt(G[V(G) \setminus V_i / C) \geq 4$, $i \in \{1, 2\}$. 
\end{definition}

\begin{definition}[contractible subgraphs]
Let $\alpha \geq 1$ and $t \geq 2$ be fixed constants. Given a $2$-edge-connected graph $G$, a collection of vertex-disjoint $2$-edge-connected subgraphs $H_1, H_2, ..., H_k$ of $G$ is called \emph{$(\alpha, t , k)$-contractible} if $2 \leq |V(H_i)| \leq t$ for every $i \in [k]$ and every $\ecss$ of $G$ contains at least $\frac{1}{\alpha} ||\bigcup_{i \in [k]} E(H_i)||$ unit-edges from  $\bigcup_{i \in [k]} E(G[V(H_i)])$.
\end{definition}

\begin{definition}[forbidden configurations, types, structured graphs]
Given a graph, a {\bf forbidden configuration} is a cut vertex, a parallel edge, a contractible subgraph, an $S_0$, an $S_1$, an $S_2$
, an $S_{\{3,4\}}$, an $S_i$, or an $S'_i$ for $i\in\{3,4,5,6\}$.
We refer to cut vertex, parallel edge, contractible subgraph, $S_0$, $S_1$, $S_2$, $S_i$, $S'_i$ for $i\in\{3,4,5,6\}$ as the {\bf types} of forbidden configurations.
A graph is  {\bf structured} if it is a MAP instance with at least 20 vertices and does not contain any forbidden configurations.
\end{definition}

Let $\alg$ be an algorithm that outputs for each structured graph $G$ a $2$-edge-connected spanning subgraph of $G$. Then, our reduction is given by the following Algorithm~\ref{alg:reduce}.
In this algorithm, $\divideT_T$ and $\combine_T$ are subroutines that are defined and proven in the lemmas below.

The main results of this section are the following two lemmas.

\lempolytimesubs*

\lemdivideCombine*

In the next section, we show some useful lemmas that help us prove the above statements.

\subsection{Useful lemmas}

\begin{lemma}\label{lem:polytimecheckT}
Given a MAP instance $G$ and a type $T$, one can check in polynomial time if $G$ contains a forbidden configuration of type $T$.
\end{lemma}

\begin{proof}
Note that each of the forbidden configurations is a subgraph on $O(1)$ vertices.
Hence, except for contractible subgraphs, we can exhaustively check for such configurations in our graph, leading to a polynomial time algorithm.

In order to find small contractible subgraphs in polynomial time, we do the following.
Note that contractible subgraphs are defined to be $(\alpha, 12, 1)$-contractable.
We iterate over all $2$-edge-connected subgraphs $H$ of $G$ such that $|V(H)|\leq 12$ in polynomial time (as $G$ is simple) and check whether $H$ is $(\alpha, 12, 1)$-contractible in polynomial time as follows.

Fix a $2$-edge-connected subgraph $H$. 
Now we iterate over each spanning subgraph $H'$ of $G[V(H)]$ (polynomially many possibilities as $G$ is simple and $|V(H)| \leq 12$) and perform the following check whether: 
\begin{enumerate}
    \item[$(i)$] $(V(G), E(H') \cup (E(G)\setminus E(G[V(H)])))$ is $2$-edge-connected, and
    \item[$(ii)$] $|E(H')| < \frac 1 \alpha |E(H)|$.
\end{enumerate}
If we can find an $H'$ for which the above condition holds, then $H$ is not $(\alpha,12, 1)$-contractible (as $(V(G), E(H') \cup (E(G)\setminus E(G[V(H)])))$ is a witness). 
Otherwise, it is $(\alpha, 12, 1)$-contractible (since there has to be a witness for non-contractibility if $H$ is not contractible).
Clearly, the entire check takes only polynomial time.
\end{proof}

In the following lemmas we define, for each type $T$ of a forbidden configuration, individual routines $\divideT_T$ and $\combine_T$ and show the properties states in Lemma~\ref{lem:divideCombine}.

\begin{lemma}\label{lem:divideCombine:contractibleSubgraphs}
Let $G$ be a graph that has a set of vertex-disjoint 2-edge-connected subgraphs $H_1, ..., H_k$ of $G$
that form a contractible subgraph.
Then there is a MAP instance $G_1$ such that
\begin{itemize}
    \item $G_1$ is 2EC, and $|V(G_1)| < |V(G)|$,
    \item $|E(G_1)| < |E(G)|$, and 
    \item if there is a 2EC subgraph $G_1^* \leq \max( \alpha \cdot \opt(G_1) -2, \opt(G_1))$, then 
      $G_1^* \cup \bigcup_{i \in [k]} E(G[V(H_i)])$ is a feasible solution to the MAP instance on $G$ and 
    $||G_1^* \cup \bigcup_{i \in [k]} E(G[V(H_i)]) || \leq \alpha \cdot \opt(G) -2$.
\end{itemize}
\end{lemma}

\begin{proof}
Let $G_1 = G / \{H_1, H_2, ..., H_k \}$, i.e.\ $G_1$ arises from $G$ by contracting each $H_i$, $i \in [k]$, to a single vertex.
Hence, the first two statements immediately follow.
Furthermore, since $V(G) \geq 20$ and $|V(H_i)| \leq 12$, we have that $G_1^* \leq \alpha \cdot \opt(G_1) -2$.
Furthermore, note that by the definition of contractible subgraphs we have that
\[ \opt(G) \geq \opt(G_1) + \frac{1}{\alpha} \sum_{i \in [k]} ||E(H_i)|| \ . \]
Hence, we have that
\begin{align*}
    ||G_1^* \cup \bigcup_{i \in [k]} E(G[V(H_i)]) || \leq \alpha \cdot \opt(G) -2 \ .
\end{align*}
This finishes the proof.
\end{proof}

\begin{lemma}\label{lem:divideCombine:cutvertices}
Let $G$ be a graph that has a cut vertex $u$. Then there are two MAP instances $H_1$ and $H_2$ such that
\begin{itemize}
    \item $H_i$ is 2EC, and $|V(H_i)| < |V(G)|$ for $i \in \{1, 2\}$,
    \item $|V(H_1)| + |V(H_2)| \leq |V(G)| + 1$, and $|E(H_1)| + |E(H_2)| \leq |E(G)|$, and
    \item if there is a 2EC subgraph $||H_i^*|| \leq \max(\alpha \cdot \opt(H_i) -2, \opt(H_i))$ for $i \in \{1, 2\}$, then $H_1^* \cup H_2^*$ is a feasible solution to the MAP instance on $G$ and  $||H_1^* \cup H_2^* || \leq \max( \alpha \cdot \opt(G) -2, \opt(G))$.
\end{itemize}
\end{lemma}

\begin{proof}
Recap the definition of $V_1$ and $V_2$ from the definition of a cut vertex.
We define $H_1 = G \setminus V_2$ and $H_2 = G \setminus V_1$.
The first two statements are trivial.
For the last statement, observe that $\opt(G) = \opt(H_1) + \opt(H_2)$.
Hence, $||H_1^* \cup H_2^* || \leq \max( \alpha \cdot \opt(G) -2, \opt(G))$.
\end{proof}

\begin{lemma}\label{lem:divideCombine:paralleledges}
Let $G$ be a graph that has no cut vertices and let $e \in E(G)$ be a parallel edge. Then $G - \{e\}$ is 2EC and $\opt(G) = \opt(G-e)$.
\end{lemma}

\begin{proof}
Assume that $G$ has no cut vertices and $\opt(G)$ contains two parallel edges $e, e' \in E(G)$.
Clearly, $\opt(G') - \{e \}$ is not two-edge connected and the only bridge in $\opt(G') - \{e \}$ is $e$, where $e = \{u v\}$.
Since $G$ does not have cut-vertices, there must be an edge $f \in E(G)$ not incident to $u$ such that $\opt(G') - \{e \} + \{ f \}$ is a 2-edge-connected spanning subgraph.
Hence, there is also an optimal solution that does not contain both parallel edges $e$ and $e'$.
\end{proof}

\begin{lemma}\label{lem:divideCombine:S0}
Let $G$ be a graph that has no cut vertices, no parallel edges and that has an $S_0$ $uv$.
Then there are two MAP instances $H_1$ and $H_2$ such that
\begin{itemize}
    \item $H_i$ is 2EC and $|V(H_i)| < |V(G)|$ for $i \in \{1, 2\}$,
    \item $|V(H_1)| + |V(H_2)| \leq |V(G)| + 1$, and $|E(H_1)| + |E(H_2)| \leq |E(G)|$, and
    \item if there is a 2EC subgraph $H_i^*$ of $H_i$ such that $||H^*_i|| \leq \max( \alpha \cdot \opt(H_i) -2, \opt(H_i))$ for $i \in \{1, 2\}$, then there is an edge-set $F$ with $|F| \leq 2$ such that $E(H^*_1) \cup E(H^*_2) \cup F$ is a feasible solution to the MAP instance on $G$ and $||E(H^*_1) \cup E(H^*_2) \cup F|| \leq \max( \alpha \cdot \opt(G) -2, \opt(G))$.
\end{itemize}
\end{lemma}

\begin{proof}
We prove the statements one by one, define the MAP instances $H_1$ and $H_2$ that are later used for Function $\dividea$, and explain how to obtain the edge set $F$ used in Function $\combinea$.
Recall that $|V(G)| > 20$.
We define $H_1 = G[V(G) \setminus V_2] / uv$ and $H_2 = G[V(G) \setminus V_1] / uv$ such that $\opt(H_1) \leq \opt(H_2)$.
Hence, the first two statements clearly follow.

First, observe that $\opt(H_1) \geq 2$ since $\{u, v\}$ is a zero-edge and therefore there are at least 2 unit edges incident to the contracted vertex $\{u, v\}$ in $\opt(H_1)$. 
Second, observe that $\opt(H_2) \geq 4$ since $\opt(H_1) \leq \opt(H_2)$ and $|V(G)| > 20$.
Third, observe that $\opt(G) \geq \opt(H_1) + \opt(H_2)$.

If $(V(G), \hat{E}(H_1^*) \cup \hat{E}(H_2^*) \cup \{u, v\})$ is 2EC, we simply output $(V(G), \hat{E}(H_1^*) \cup \hat{E}(H_2^*) \cup \{uv\}$ and observe that
\begin{align*}
    || \hat{E}(H_1^*) \cup \hat{E}(H_2^*) \cup \{u, v\} || \leq & \ \max(\alpha \cdot \opt(H_1) -2, \opt(H_1)) +  \max(\alpha \cdot \opt(H_2) -2, \opt(H_2))\\
     \leq & \ \alpha \cdot (  \opt(H_1) +  \opt(H_2)) - 2 \leq \alpha \cdot \opt(G) -2 \ ,
\end{align*}
where the second inequality holds since $\opt(H_2) \geq 4$.

Else, if $(V(G), \hat{E}(H_1^*) \cup \hat{E}(H_2^*) \cup \{u, v\})$ is not 2EC, observe that there is at most one cut edge and this edge is $\{u, v \}$.
In this case, both $\hat{E}(H_1^*)$ and $\hat{E}(H_2^*)$ are only connected to $u$ or $v$, respectively. 
Since there are no cut-vertices in $G$, there must be an edge $f \in E(G)$ that is connected to either $u$ or $v$ such that $(V(G), \hat{E}(H_1^*) \cup \hat{E}(H_2^*) \cup \{u, v \} \cup \{f\})$ is 2EC.
Then, we have
\begin{align*}
    || \hat{E}(H_1^*) \cup \hat{E}(H_2^*) \cup \{uv\} \cup \{f\} || \leq &  \max(\alpha \cdot \opt(H_1) -2, \opt(H_1)) \\
    & + \max(\alpha \cdot \opt(H_2) -2, \opt(H_2)) + 1 \\ 
    \leq & \ \max(\alpha \cdot \opt(H_1) -2, \opt(H_1)) +  \alpha \cdot \opt(H_2) - 2 + 1 \\
    \leq & \ \alpha \cdot \opt(H_1) - 1 + \alpha \cdot \opt(H_2) - 2 + 1 \leq \alpha \cdot \opt(G) -2 \ ,
\end{align*}
where the second inequality holds since $\opt(H_2) \geq 4$ and the third inequality holds since $\opt(H_1) \geq 2$.
This proves the lemma.
\end{proof}

\begin{lemma}\label{lem:divideCombine:S1}
Let $G$ be a graph that has no cut-vertices, no parallel edges and that has an $S_1$ $uv$. 
Then there are two MAP instances $H_1$ and $H_2$ such that
\begin{itemize}
    \item $H_i$ is 2EC and $|V(H_i)| < |V(G)|$ for $i \in \{1, 2\}$,
    \item $|V(H_1)| + |V(H_2)| \leq |V(G)| + 1$, and $|E(H_1)| + |E(H_2)| \leq |E(G)| + 1$, and
    \item if there is a 2EC subgraph $H_i^*$ of $H_i$ such that $||H^*_i|| \leq \max( \alpha \cdot \opt(H_i) -2, \opt(H_i))$ for $i \in \{1, 2\}$, then there is an edge-set $F$ with $|F| \leq 2$ such that $E(H^*_1) \cup E(H^*_2) \cup F$ is a feasible solution to the MAP instance on $G$ and $||E(H^*_1) \cup E(H^*_2) \cup F|| \leq \max( \alpha \cdot \opt(G) -2, \opt(G))$.
\end{itemize}
\end{lemma}

\begin{figure}
\centering
\begin{subfigure}{.5\textwidth}
  \centering
  \includegraphics[width=.7\linewidth]{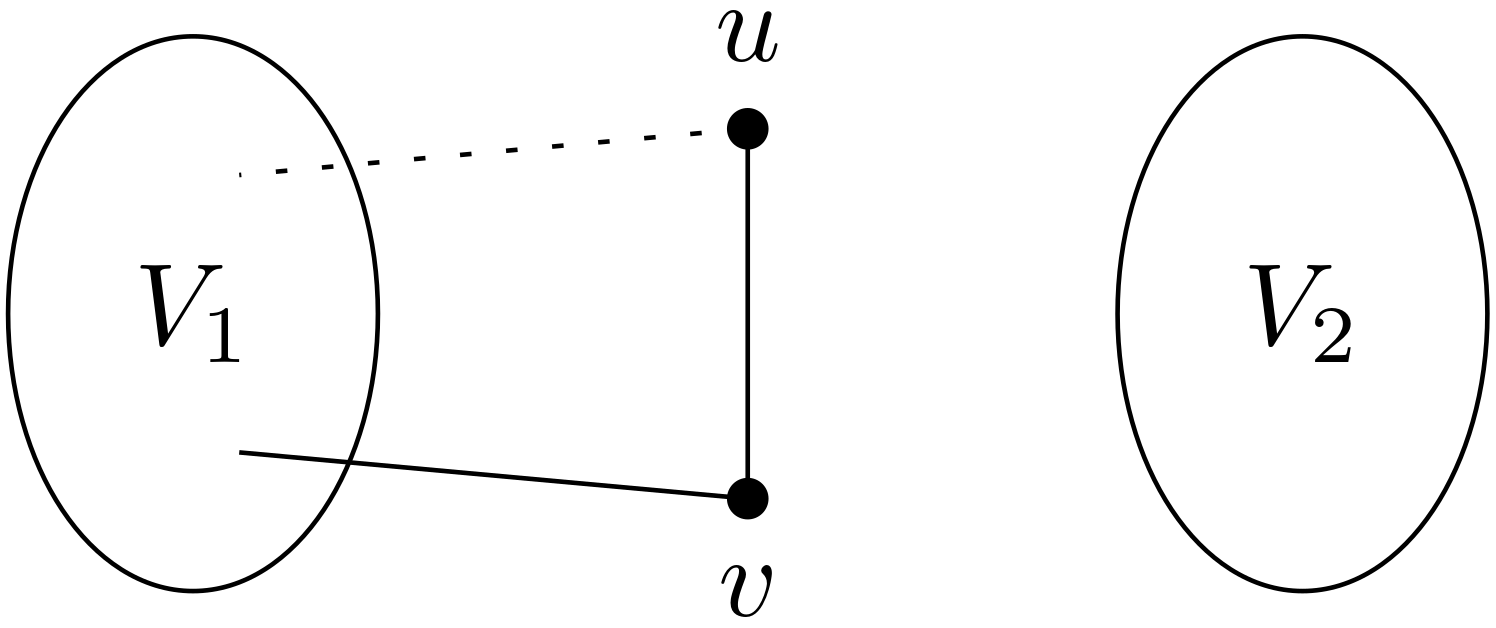}
\end{subfigure}%
\begin{subfigure}{.5\textwidth}
  \centering
  \includegraphics[width=.38\linewidth]{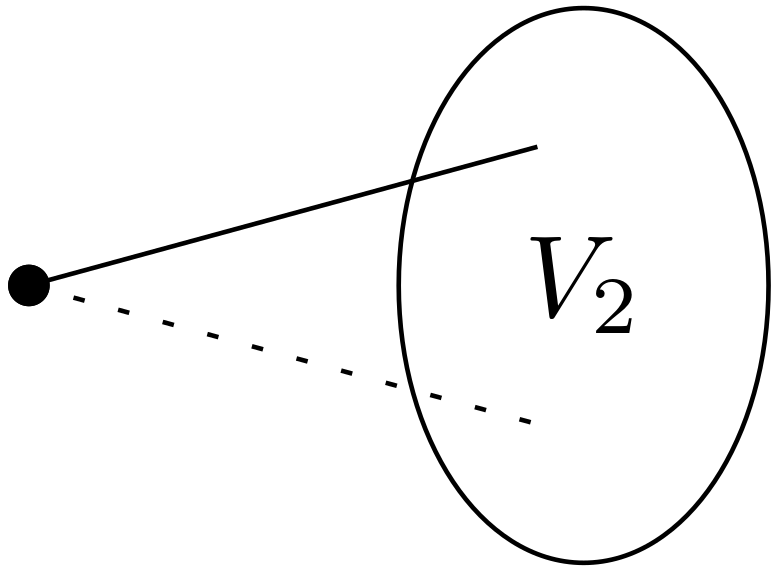}
\end{subfigure}
\caption{Illustration of an $S_1$ and $H_2'$. Zero-edges are dotted, and unit edges are not dotted.} 
\label{fig:S1_1}
\end{figure}

\begin{proof}
We prove the statements one by one, define the MAP instances $H_1$ and $H_2$ that are later used for Function $\divideb$, and explain how to obtain the edge set $F$ used in Function $\combineb$. 

Recall that $|V(G)| > 20$. 
We define $H_1 = G[V(G) \setminus V_2] / uv$ and $H_2' = G[V(G) \setminus V_1] / uv$ such that $\opt(H_1) \leq \opt(H_2)$.
Furthermore, let $H_2'' = G[V(G) \setminus V_1] \cup \{e'\}$, where $e'$ is a unit-cost parallel edge to $\{u, v\}$. 
That is, the difference between $H_2'$ and $H_2''$ is whether the edge $\{u, v\}$ is contracted or not (plus some additional parallel edge $e'$).
However, when handling an $S_1$, depending on $H_1$, we sometimes run our algorithm on $H_2'$ and sometimes on $H_2''$. 
With slight abuse of notation and in order to increase readability, by $H_2^*$ we will refer to the solution computed by the algorithm on $H_2'$ and $H_2''$.
In any case, the first two statements clearly follow.

Recall, by the definition of an $S_1$, we have that $\opt(H_1) \geq 3$ and $\opt(H_2) \geq 4$.
Furthermore, observe that $\opt(G) \geq \opt(H_1) + \opt(H_2')$.

Let us first assume that $||\opt(H_1)|| \geq 4$.
In this case, we run our algorithm on $H_2 = H_2'$.
Then, observe that (similar to the proof of Lemma~\ref{lem:divideCombine:S0}), there can be at most one cut-edge in $(V(G), \hat{E}(H_1^*) \cup \hat{E}(H_2^*) \cup \{u, v\})$ and this edge is $\{u, v \}$.
Then, since there are no cut-vertices, there must be an edge $f \in E(G)$ incident to $u$ or $v$ such that $(V(G), \hat{E}(H_1^*) \cup \hat{E}(H_2^*) \cup \{u, v\} \cup \{f\})$ is 2EC.
Then we have 
\begin{align*}
    || \hat{E}(H_1^*) \cup \hat{E}(H_2^*) \cup \{u, v\} \cup \{f\} || \leq & \ \max(\alpha \cdot \opt(H_1) -2, \opt(H_1)) \\ 
    & \ +  \max(\alpha \cdot \opt(H_2) -2, \opt(H_2)) + 2\\
     \leq & \ \alpha \cdot (  \opt(H_1) +  \opt(H_2)) - 4 + 2 \leq \alpha \cdot \opt(G) -2 \ ,
\end{align*}
where the second inequality follows from the fact that $\opt(H_1) \geq 4$ and $\opt(H_2) \geq 4$.

Now let us assume that $\opt(H_1) = 3$.
First, let us assume there is an optimal solution $\OPT(H_1)$, also called $H_1^*$ for brevity, such that the edge set $\hat{E}(H_1^*)$ is connected to both $u$ and $v$ when uncontracting $\{u, v\}$.
An illustration of this case is given in Figure~\ref{fig:S1_1}.
In this case, we run our algorithm on $H_2'$.
Observe that $(V(G), \hat{E}(H_1^*) \cup \hat{E}(H_2^*) \cup \{uv\})$ is 2EC.
Hence, we have 
\begin{align*}
    || \hat{E}(H_1^*) \cup \hat{E}(H_2^*) \cup \{uv\} || \leq & \ \max(\alpha \cdot \opt(H_1) -2, \opt(H_1))  +  \max(\alpha \cdot \opt(H_2) -2, \opt(H_2)) + 1\\
    \leq & \ \opt(H_1) + \alpha \cdot \opt(H_2) -2 + 1 \leq \alpha \cdot \opt(H_1) - \alpha + \alpha \cdot \opt(H_2) -1 \\
    \leq & \ \alpha \cdot (  \opt(H_1) +  \opt(H_2)) - 2 \leq \alpha \cdot \opt(G) -2  \ ,
\end{align*}
where the second and third inequality follows from the fact that $\opt(H_2) \geq 4$ and $\opt(H_1) = 3$, respectively.

\begin{figure}
\centering
\begin{subfigure}{.5\textwidth}
  \centering
  \includegraphics[width=.7\linewidth]{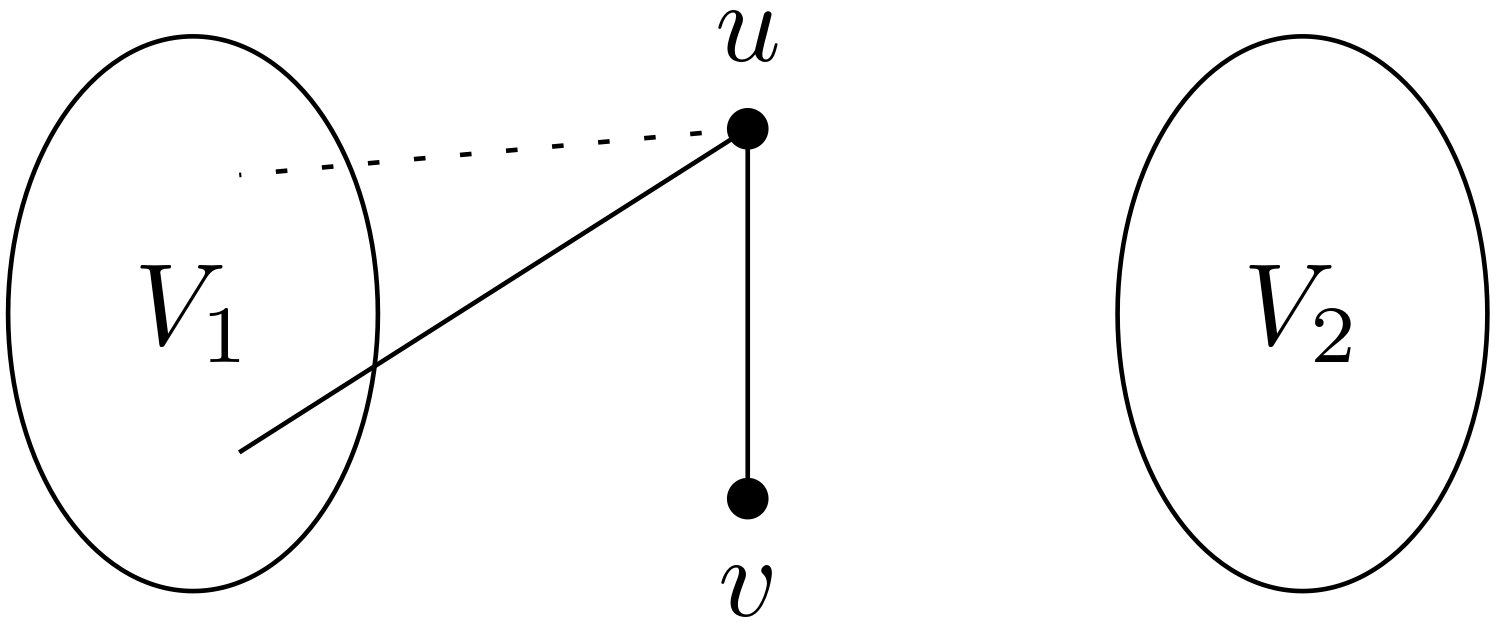}
\end{subfigure}%
\begin{subfigure}{.5\textwidth}
  \centering
  \includegraphics[width=.38\linewidth]{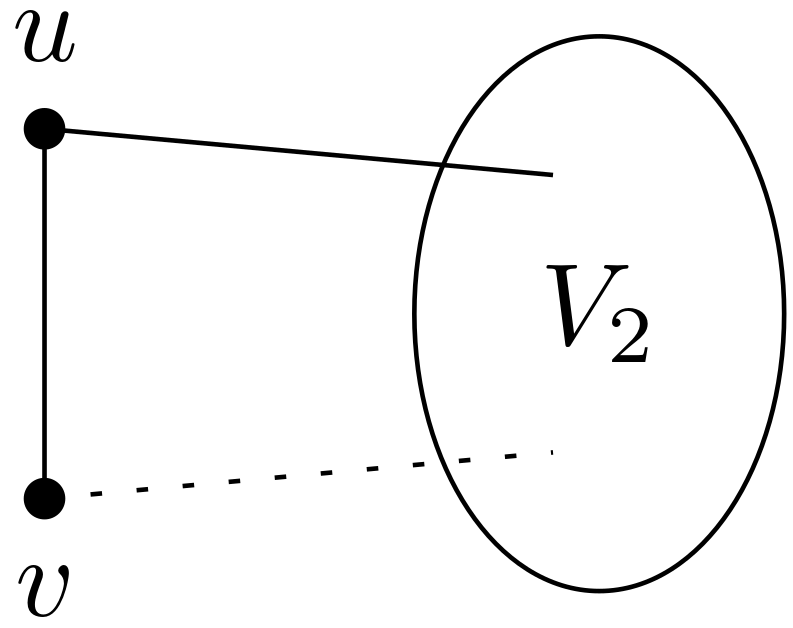}
\end{subfigure}
\caption{Illustration of an $S_1$ and $H_2''$. Zero-edges are dotted, and unit edges are not dotted.}
\label{fig:S1_2}
\end{figure}

Now, let us assume that there is no $\opt(H_1)$ such that the edge set $\hat{E}(H_1^*)$ is connected to both $u$ and $v$ when uncontracting $\{u, v\}$.
In this case, we run our algorithm on $H_2''$. 
An illustration of this case is given in Figure~\ref{fig:S1_2}.
However, for brevity and readability by $H_2^*$ we refer to the solution output by the algorithm on $H_2''$.
First, observe that $(V(G), \hat{E}(H_1^*) \cup \hat{E}(H_2^*))$ is 2EC. 
Furthermore, if in $H_2^*$ the introduced parallel edge $e'$ is chosen, then we can replace this edge with an appropriate edge $f \in E(G)$ such that the resulting graph is 2EC since there are no cut-vertices.
Second, observe that $\opt(H_2') \leq \opt(H_2'') \leq \opt(H_2') + 2$. 
The first inequality follows from the fact that $H_2'$ arises from $H_2'$ by contracting $\{u, v\}$.
The second inequality follows from the fact that any solution to $H_2$ can be transformed to a solution for $H_2''$ by adding $\{u,  v\}$ and $e'$, where $e'$ was the introduced parallel edge to $\{u, v\}$.

We now make a case distinction on whether $\opt(H_2') = \opt(H_2'')$ or $\opt(H_2') + 1 \leq \opt(H_2'')||$.
In the first case, if $\opt(H_2') = \opt(H_2'')$, observe that $\opt(G) \geq \opt(H_1) + \opt(H_2'')$, since $\opt(G) \geq \opt(H_1) + \opt(H_2)$.
Hence, we have 
\begin{align*}
    || \hat{E}(H_1^*) \cup \hat{E}(H_2^*)|| \leq & \ \max(\alpha \cdot \opt(H_1) -2, \opt(H_1)) +  \max(\alpha \cdot \opt(H_2'') -2, \opt(H_2''))\\
    \leq & \ \opt(H_1) + \alpha \cdot \opt(H_2'') -2 \leq \alpha \cdot (  \opt(H_1) +  \opt(H_2'')) - 2\\
    \leq & \ \alpha \cdot \opt(G) -2  \ ,
\end{align*}
where the second inequality follows from the fact that $\opt(H_2') \geq 4$.

In the remaining case, we have that $\opt(H_2') + 1 \leq \opt(H_2'') \leq \opt(H_2') + 2$.
We now show that $\opt(G) \geq \opt(H_1) + \opt(H_2') + 1$ (a slightly stronger statement compared to the one we used before).
To derive a contradiction, assume this is not true and we have $\opt(G) = \opt(H_1) + \opt(H_2')$. 
Consider the vertices $u$ and $v$ in $G$. 
Since $\OPT(G)$ is 2EC, there must be a cycle in $\OPT(G)$ that goes through $u$ and $v$.
Possibly, this cycle uses the edge $\{u, v\} \in E(G)$.
Hence, in $(V(G), \OPT(G) \setminus \{u, v\})$ there must be path $P$ from $u$ to $v$.
If $P$ is also a path in $H_1$, then $\opt(G) = \opt(H_1') + 1 + \opt(H_2)$, since we are in the case that there is no $\OPT(H_1)$ such that the edge set $\hat{E}(H_1^*)$ is connected to both $u$ and $v$ when uncontracting $\{u, v\}$. 
However, this is a contradiction.
Hence, $P$ must be a path in $H_2$.
But then $\OPT(G)[V(G) \setminus V_1]$ induces a solution to $H_2''$ and it follows that $\opt(H_2') = \opt(H_2'')$, again a contradiction.

Using  $\opt(G) \geq \opt(H_1) + \opt(H_2') + 1$ and the fact that $ \opt(H_2'') \leq \opt(H_2') + 2$, we obtain $\opt(G) \geq \opt(H_1) + \opt(H_2'') - 1$.
We now set $H_2 = H_2''$.
Hence, we have 
\begin{align*}
    || \hat{E}(H_1^*) \cup \hat{E}(H_2^*)|| \leq & \ \max(\alpha \cdot \opt(H_1) -2, \opt(H_1)) +  \max(\alpha \cdot \opt(H_2') -2, \opt(H_2'))\\
    \leq & \ \opt(H_1) + \alpha \cdot \opt(H_2') -\alpha -2 \leq \alpha \cdot (  \opt(H_1) +  \opt(H_2')) - \alpha - 2\\
    \leq & \ \alpha \cdot (  \opt(G) + 1 ) - \alpha - 2 = \alpha \cdot \opt(G) -2  \ ,
\end{align*}
where the second and third inequality follows from the fact that $\opt(H_2) \geq 4$ and $\opt(H_1) = 3$, respectively, and the fourth inequality follows from the fact that $\opt(G) \geq \opt(H_1) + \opt(H_2') - 1$.
This proves the lemma.
\end{proof}

\begin{lemma}\label{lem:divideCombine:S2}
Let $G$ be a graph that has no cut-vertices, no parallel edges and that has an $S_2$ $uvw$. 
Then there are two MAP instances $H_1$ and $H_2$ such that
\begin{itemize}
    \item $H_i$ is 2EC and $|V(H_i)| < |V(G)|$ for $i \in \{1, 2\}$,
    \item $|V(H_1)| + |V(H_2)| \leq |V(G)| + 1$, and $|E(H_1)| + |E(H_2)| \leq |E(G)| + 2$, and
    \item if there is a 2EC subgraph $H_i^*$ of $H_i$ such that $||H^*_i|| \leq \max( \alpha \cdot \opt(H_i) -2, \opt(H_i))$ for $i \in \{1, 2\}$, then there is an edge-set $F$ with $|F| \leq 2$ such that $E(H^*_1) \cup E(H^*_2) \cup F$ is a feasible solution to the MAP instance on $G$ and $||E(H^*_1) \cup E(H^*_2) \cup F|| \leq \max( \alpha \cdot \opt(G) -2, \opt(G))$.
\end{itemize}
\end{lemma}

\begin{proof}
We prove the statements one by one, define the MAP instances $H_1$ and $H_2$ that are later used for Function $\dividec$, and explain how to obtain the edge set $F$ used in Function $\combinec$.

Recall that $|V(G)| > 20$. 
We define $H_1 = G[V(G) \setminus V_2] / uvw$ and $H_2' = G[V(G) \setminus V_1] / uvw$ such that $\opt(H_1) \leq \opt(H_2)$.
Furthermore, let $H_2'' = G[V(G) \setminus V_1] / vw \cup \{e''\}$, where $e'' \subseteq \{ u, v \}$ is a parallel edge to $\{u, v\}$.
Additionally, let $H_2''' = G[V(G) \setminus V_1] \cup \{F'''\}$, where $F''' \subseteq \{ \{u, v\}, \{u, w\}, \{v, u\} \}$ are possible parallel edges to $\{u, v\}$, $\{u, w\}$, and $\{v, u\}$. For $x, y \in \{u, v, w \}$, $x \neq y$, the edge $\{x, y\}$ is present in $F'''$ if and only if there is a path from $x$ to $y$ in $G[V(G) \setminus V_2)] \setminus \{ \{u, v\}, \{u, w\}, \{v, u\} \}$.

That is, the difference between $H_2'''$, $H_2''$, and $H_2'$ is whether no edge, the edge $\{u, v \}$, or the edges $\{u, v\}$ and $\{ v, w\}$ are contracted (plus some additional parallel edges), respectively.
However, when handling an $S_2$, depending on $H_1$, we sometimes run our algorithm on $H_2'$, sometimes on $H_2''$, and sometimes on $H_2'''$. 
Similar to previous lemmas, with a slight abuse of notation and in order to increase readability, by $H_2^*$ we will refer to the solution computed by the algorithm on $H_2'$, $H_2''$, and $H_2'''$ .
In any case, the first two statements clearly follow.

Recall, by the definition of an $S_2$, we have that $\opt(H_1) \geq 3$ and $\opt(H_2) \geq 4$ and let $g \in E(G)$ be the zero-edge incident to $u$ and some vertex of $V_2$.
Furthermore, $\{u, v\}$ is a unit-edge, $\{v, w\}$ is a zero-edge, and observe that $\opt(G) \geq \opt(H_1) + \opt(H_2')$.

Let us first assume that $\opt(H_1) \geq 4$.
In this case, we run our algorithm on $H_2 = H_2'$.
Then, observe that (similar to the proof of the Lemma~\ref{lem:divideCombine:S0}), there can be at most two cut-edges in $(V(G), \hat{E}(H_1^*) \cup \hat{E}(H_2^*) \cup \{u, v\} \cup \{v, w \} \cup \{g\})$ and these edges are $\{u, v \}$ and $\{v, w \}$.
Then, since there are no cut-vertices, no $S_0$ and no $S_1$, there must be an edge $f \in E(G)$ incident to $u$, $v$ or $w$ such that $(V(G), \hat{E}(H_1^*) \cup \hat{E}(H_2^*) \cup \{u, v\} \cup \{v, w \} \cup \{g\} \cup \{f\})$ is 2EC.
Then we have 
\begin{align*}
    || \hat{E}(H_2^*) \cup \{uv\} \cup \{v w \} \cup \{g\} \cup \{f\} || \leq & \ \max(\alpha \cdot \opt(H_1) -2, \opt(H_1)) \\ 
    & \ +  \max(\alpha \cdot \opt(H_2) -2, \opt(H_2)) + 2\\
     \leq & \ \alpha \cdot (  \opt(H_1) +  \opt(H_2)) - 4 + 2\\
     \leq & \ \alpha \cdot \opt(G) -2 \ ,
\end{align*}
where the second inequality follows from the fact that $\opt(H_1) \geq 4$ and $\opt(H_2) \geq 4$.

Now let us assume that $\opt(H_1) = 3$.
First, let us assume there is an $\OPT(H_1)$, also called $H_1^*$ for brevity, such that the edge set $\hat{E}(H_1^*)$ is connected to both $u$ and $w$ when uncontracting $\{uvw\}$.
In this case, we run our algorithm on $H_2 = H_2'$.
\begin{figure}
\centering
\begin{subfigure}{.5\textwidth}
  \centering
  \includegraphics[width=.7\linewidth]{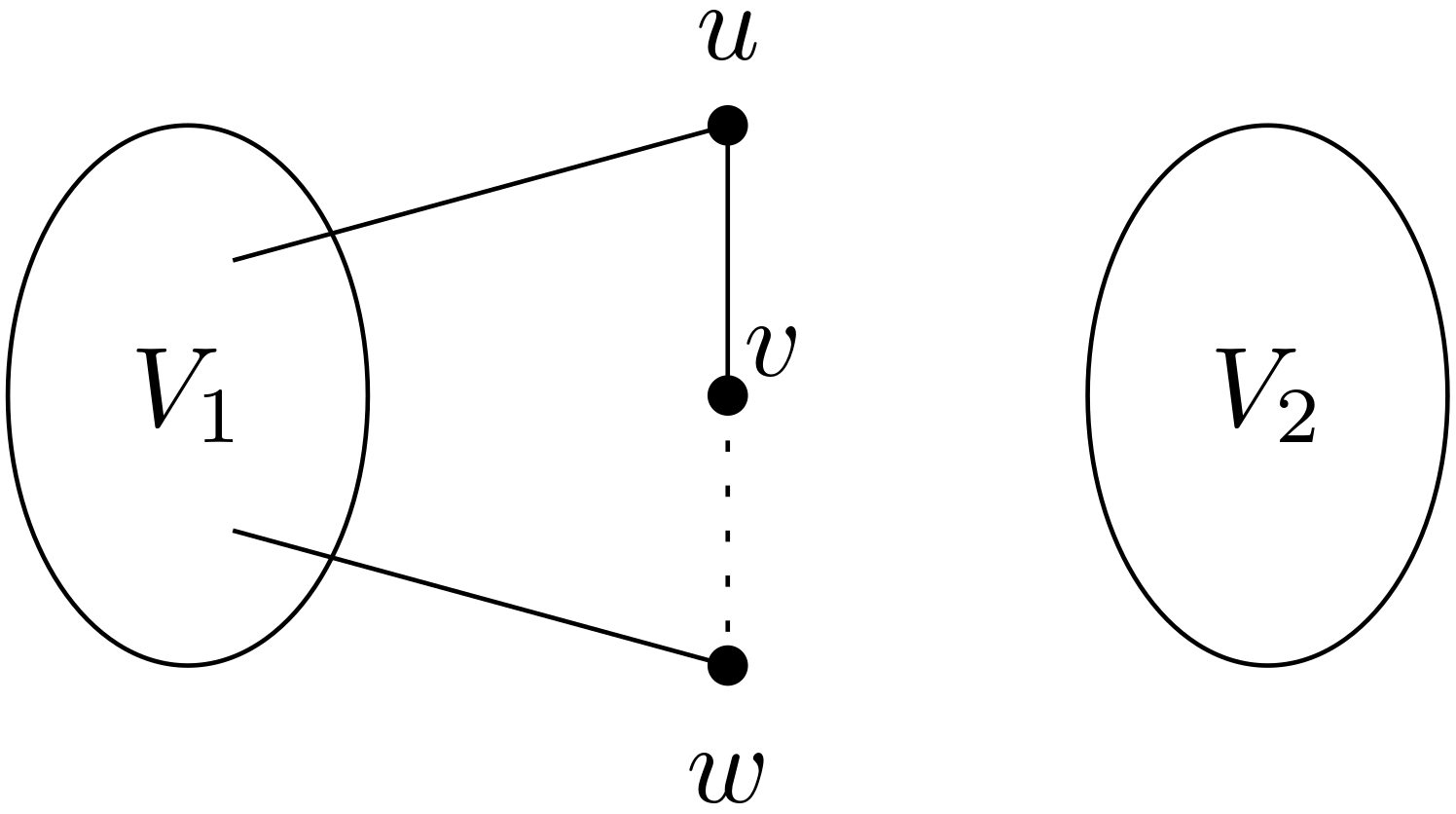}
\end{subfigure}%
\begin{subfigure}{.5\textwidth}
  \centering
  \includegraphics[width=.38\linewidth]{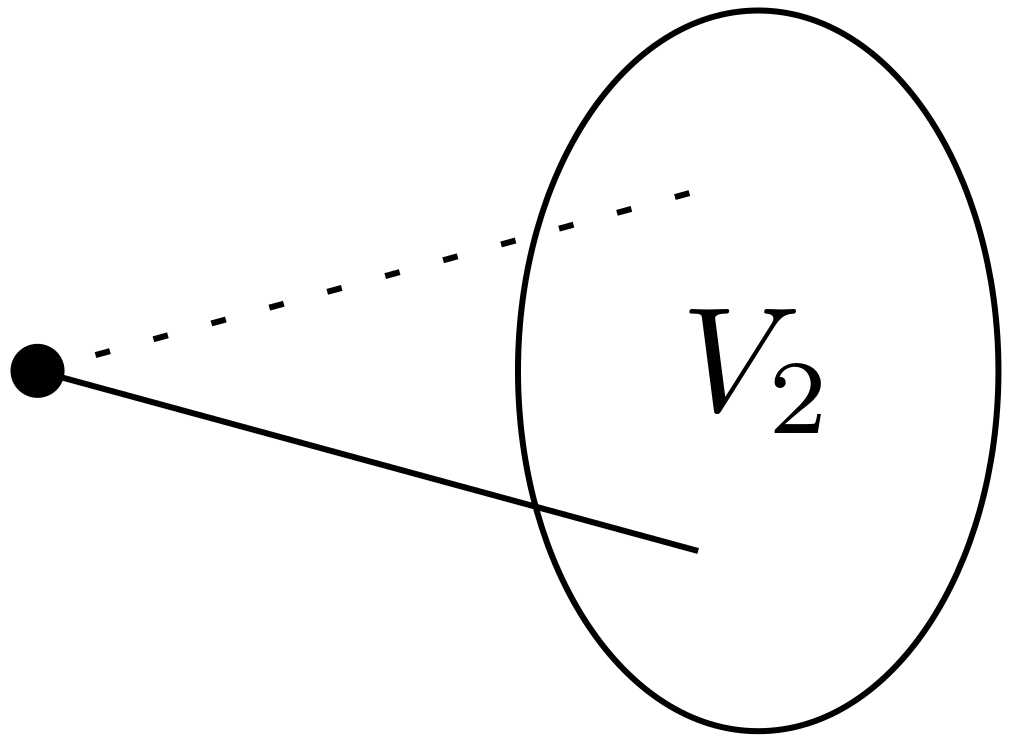}
\end{subfigure}
\caption{Illustration of an $S_2$ and $H_2'$. Zero-edges are dotted, and unit edges are not dotted.}
\label{fig:S2_2}
\end{figure}
An illustration of this case is given in Figure~\ref{fig:S2_2}.
Observe that $(V(G), \hat{E}(H_1^*) \cup \hat{E}(H_2^*) \cup \{u, v\} \cup \{v, w\})$ is 2EC.
Hence, we have 
\begin{align*}
    ||  \hat{E}(H_1^*) \cup \hat{E}(H_2^*) \cup \{u, v\} \cup \{v, w\} || \leq & \ \max(\alpha \cdot \opt(H_1) -2, \opt(H_1)) \\ 
    & \ +  \max(\alpha \cdot \opt(H_2) -2, \opt(H_2)) + 1\\
    \leq & \ \opt(H_1) + \alpha \cdot \opt(H_2) -2 + 1 \\
    \leq & \ \alpha \cdot \opt(H_1) - \alpha + \alpha \cdot \opt(H_2) -2 + 1 \\
    \leq & \ \alpha \cdot \opt(G) -2  \ ,
\end{align*}
where the second and third inequality follows from the fact that $\opt(H_2) \geq 4$ and $\opt(H_1) = 3$, respectively.

Second, let us assume there is an $\OPT(H_1)$, also called $H_1^*$ for brevity, such that the edge set $\hat{E}(H_1^*)$ is connected to both $u$ and $v$ when uncontracting $\{uvw\}$.
In this case, we run our algorithm on $H_2= H_2'$.
Observe that $(V(G), \hat{E}(H_1^*) \cup \hat{E}(H_2^*) \cup \{v, w\}  \cup \{f\} )$ is 2EC, where $f$ is any edge incident to $w$.
Hence, we have 
\begin{align*}
    || \hat{E}(H_1^*) \cup \hat{E}(H_2^*) \cup \{v, w\}  \cup \{f\} || \leq & \ \max(\alpha \cdot \opt(H_1) -2, \opt(H_1))  \\ 
    & \ +  \max(\alpha \cdot \opt(H_2) -2, \opt(H_2)) + 1\\
    \leq & \ \opt(H_1) + \alpha \cdot \opt(H_2) -2 + 1 \\
    \leq & \ \alpha \cdot \opt(H_1) - \alpha + \alpha \cdot \opt(H_2) -1 \\
    \leq & \ \alpha \cdot \opt(G) -2  \ ,
\end{align*}
where the second and third inequality follows from the fact that $\opt(H_2) \geq 4$ and $\opt(H_1) = 3$, respectively.

Third, let us assume there is no $\OPT(H_1)$, such that such that the edge set $\hat{E}(\opt(H_1))$ is connected to both $u$ and $w$ or to both $u$ and $v$, when uncontracting $\{uvw\}$.
Additionally, let us assume that there is some $\OPT(H_1)$, also called $H_1^*$ for brevity, such that the edge set $\hat{E}(H_1^*)$ is connected to both $v$ and $w$ when uncontracting $\{uvw\}$.
In this case, we run our algorithm on $H_2 = H_2''$ (which is formally defined at the beginning of this proof). 
\begin{figure}
\centering
\begin{subfigure}{.5\textwidth}
  \centering
  \includegraphics[width=.7\linewidth]{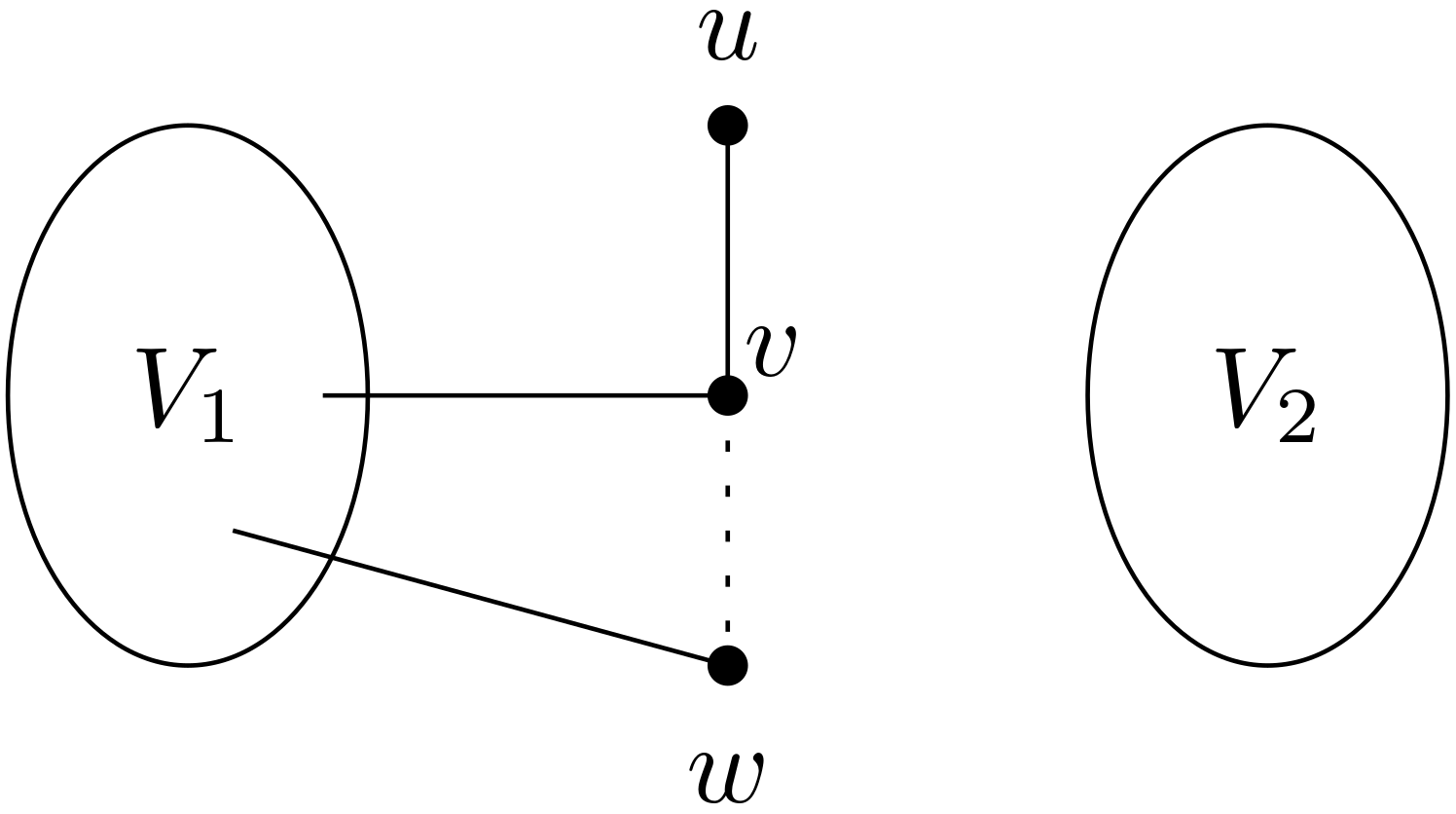}
\end{subfigure}%
\begin{subfigure}{.5\textwidth}
  \centering
  \includegraphics[width=.38\linewidth]{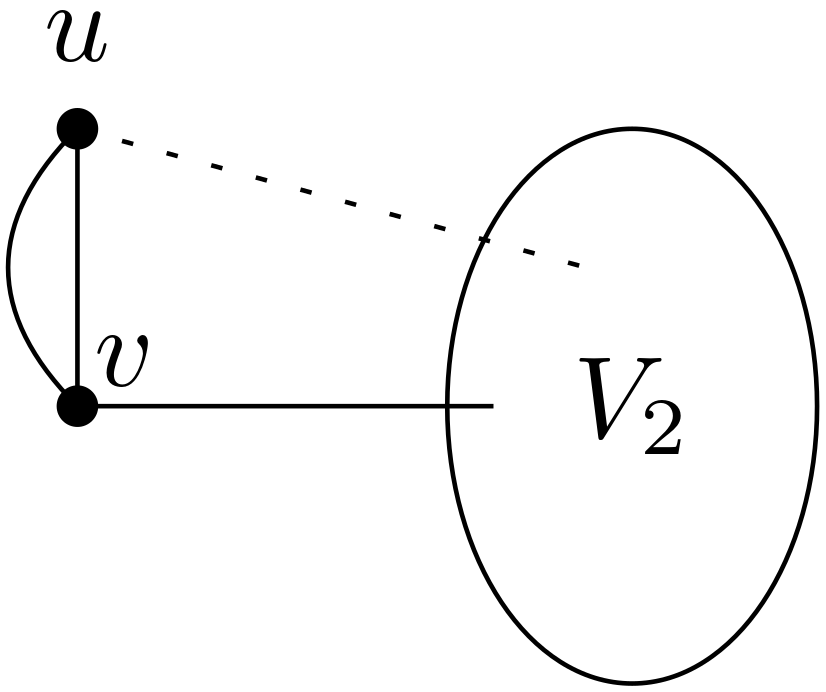}
\end{subfigure}
\caption{Illustration of an $S_2$ and $H_2''$. Zero-edges are dotted, and unit edges are not dotted.}
\label{fig:S2_3}
\end{figure}
An illustration of this case is given in Figure~\ref{fig:S2_3}.
Recall that $H_2''$ is the graph arising from $G[V(G) \setminus V_1] / \{v, w \}$ and adding an additional parallel edge $\{u, v\}$.
However, for brevity and readability by $H_2^*$ we refer to the solution output by the algorithm on $H_2''$.

We proceed similarly to the last case of the proof in Lemma~\ref{lem:divideCombine:S1}.
First, observe that the graph $(V(G), \hat{E}(H_1^*) \cup \hat{E}(H_2^*) \cup \{v, w\})$ is 2EC. 
Furthermore, if in $(V(G), \hat{E}(H_1^*) \cup \hat{E}(H_2^*) \cup \{v, w\})$ the introduced parallel edge $e'$ is chosen, then we can replace this edge by an appropriate edge $f \in E(G)$ such that the resulting graph is 2EC since there are no cut-vertices.
Second, observe that $\opt(H_2') \leq \opt(H_2'') \leq \opt(H_2') + 2$. 

We now make a case distinction on whether $\opt(H_2') = \opt(H_2'')$ or $\opt(H_2') + 1 \leq \opt(H_2'')$.
In the first case, if $\opt(H_2') = \opt(H_2'')$, observe that $\opt(G) \geq \opt(H_1) + \opt(H_2'')$, since $\opt(G) \geq \opt(H_1) + ||\opt(H_2')$.
Hence, we have 
\begin{align*}
    || \hat{E}(H_1^*) \cup \hat{E}(H_2^*)  \cup \{v, w\}|| \leq & \ \max(\alpha \cdot \opt(H_1) -2, \opt(H_1))  +  \max(\alpha \cdot \opt(H_2') -2, \opt(H_2'))\\
    \leq & \ \opt(H_1) + \alpha \cdot \opt(H_2') -2 
    \leq \alpha \cdot \opt(G) -2  \ ,
\end{align*}
where the second inequality follows from the fact that $\opt(H_2) \geq 4$.

In the remaining case, we have that $\opt(H_2') + 1 \leq \opt(H_2'') \leq \opt(H_2') + 2$.
Since there is no $\opt(H_1)$, such that such that the edge set $\hat{E}(\opt(H_1))$ is connected to both $u$ and $w$ or to both $u$ and $v$, when uncontracting $\{uvw\}$, similarly as in the proof of Lemma~\ref{lem:divideCombine:S1} we can show that in this case $\opt(G) \geq \opt(H_1) + \opt(H_2) + 1$.
Using  $\opt(G) \geq \opt(H_1) + \opt(H_2') + 1$ and the fact that $ \opt(H_2'') \leq \opt(H_2') + 2$, we obtain $\opt(G) \geq \opt(H_1) + \opt(H_2'') - 1$.
Hence, we have 
\begin{align*}
    || \hat{E}(H_1^*) \cup \hat{E}(H_2^*)  \cup \{v, w\}|| \leq & \ \max(\alpha \cdot \opt(H_1) -2, \opt(H_1)) +  \max(\alpha \cdot \opt(H_2') -2, \opt(H_2'))\\
    \leq & \ \opt(H_1) + \alpha \cdot \opt(H_2') -\alpha -2
    \leq  \alpha \cdot (  \opt(H_1) +  \opt(H_2')) - 2 - \alpha\\
    \leq & \ \alpha \cdot (  \opt(G) + 1 ) - 2 - \alpha
    \leq \alpha \cdot \opt(G) -2  \ ,
\end{align*}
where the second and third inequality follows from the fact that $\opt(H_2) \geq 4$ and $\opt(H_1) = 3$, respectively, and the fourth inequality follows from the fact that $\opt(G) \geq \opt(H_1) + \opt(H_2') - 1$.

Finally, let us assume that there is no $OPT(H_1)$, such that the edge set $\hat{E}(\opt(H_1))$ is connected to two distinct vertices of the set $\{u, v, w \}$.
In this case, we run our algorithm on $H_2 = H_2'''$ (which is formally defined at the beginning of this proof). 
Recall that $H_2'''$ is the graph arising from $G[V(G) \setminus V_1]$ and adding some potential parallel edges between some pairs of the set $\{u, v, w\}$, depending on the structure of $H_1$.
However, for brevity and readability by $H_2^*$ we refer to the solution output by the algorithm on $H_2'''$.
\begin{figure}
\centering
\begin{subfigure}{.5\textwidth}
  \centering
  \includegraphics[width=.7\linewidth]{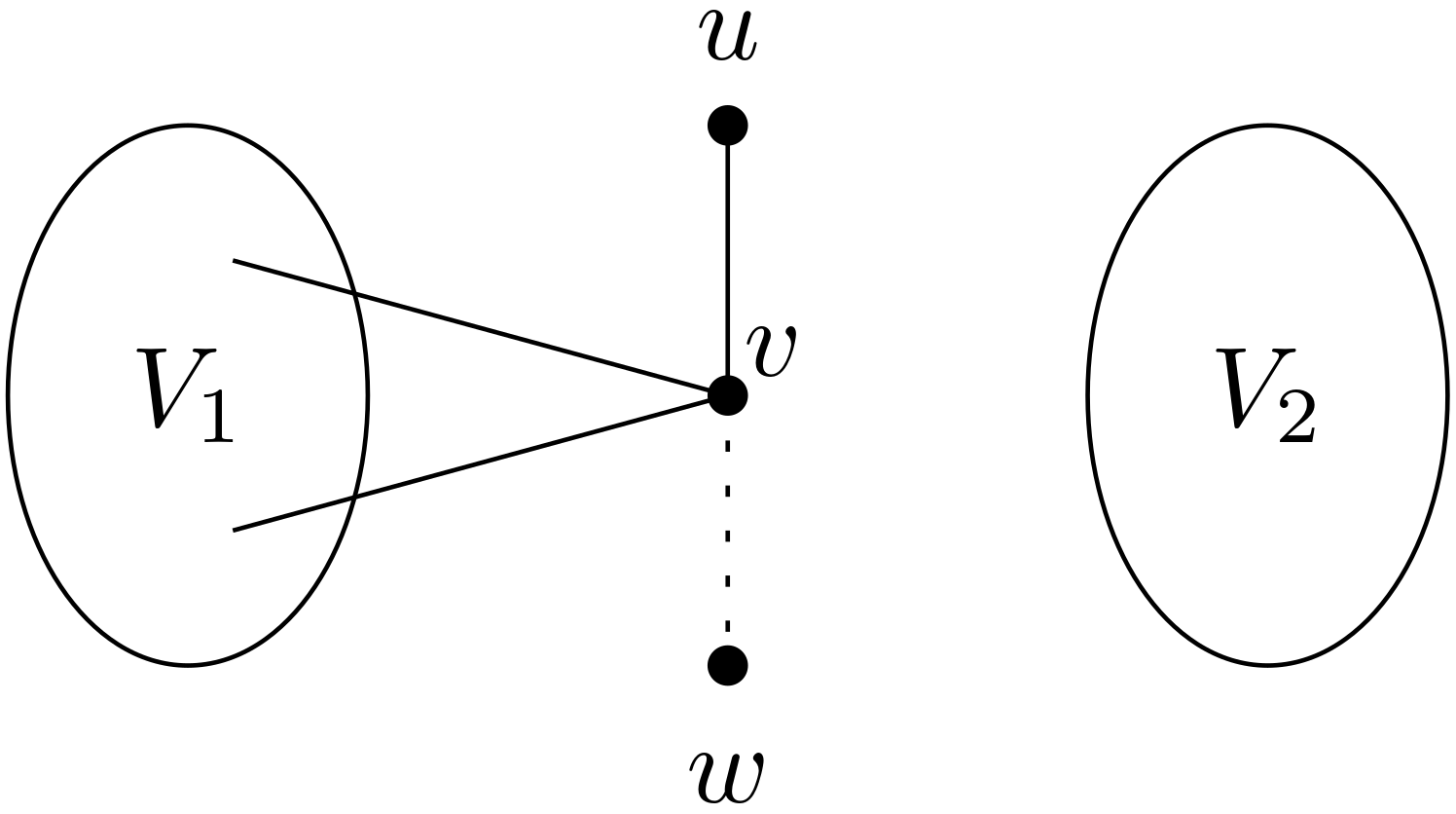}
\end{subfigure}%
\begin{subfigure}{.5\textwidth}
  \centering
  \includegraphics[width=.38\linewidth]{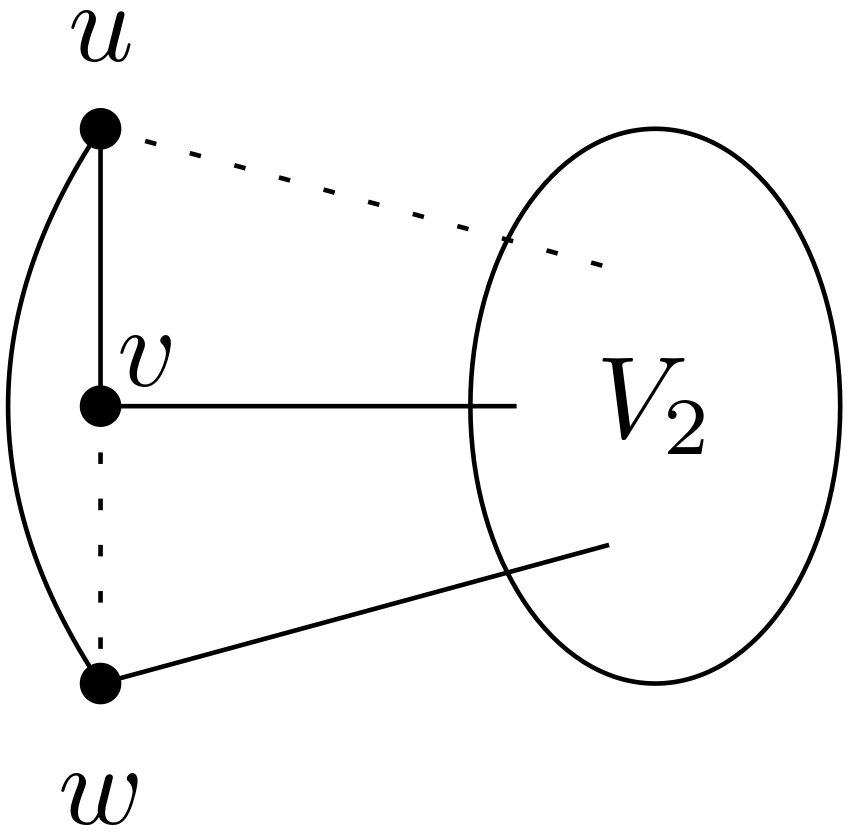}
\end{subfigure}
\caption{Illustration of an $S_2$ and $H_2'''$. Zero-edges are dotted, and unit edges are not dotted.}
\label{fig:S2_4}
\end{figure}
An illustration of this case is given in Figure~\ref{fig:S2_4}.
We proceed similarly to the previous case.
First, observe that $(V(G), \hat{E}(H_1^*) \cup \hat{E}(H_2^*))$ is 2EC. 
Furthermore, if in $(V(G), \hat{E}(H_1^*) \cup \hat{E}(H_2^*))$ any introduced parallel edge $e'$ is chosen, then we can replace this edge with an appropriate edge $f \in E(G)$ such that the resulting graph is 2EC, since there are no cut-vertices.
Second, observe that $\opt(H_2') \leq \opt(H_2''') \leq \opt(H_2') + 2$. 
The second inequality follows from the fact that any solution $\OPT(H_2')$ can be transformed to a solution for $\OPT(H_2''')$ by adding the edges $g$, $\{u, v\}$, $\{v, w\}$, and $\{ w, x \}$ to $\OPT(H_2')$, where $x$ is some vertex in $V_2$. 
Such an edge $\{ w, x \}$ exists since otherwise $uv$ forms an $S_1$, a contradiction. 
Finally, observe that $g$ and $\{u, v\}$ are zero-edges.

We now make a case distinction on whether $\opt(H_2') = \opt(H_2''')$ or $\opt(H_2') + 1 \leq \opt(H_2''')$.
In the first case, if $\opt(H_2') = \opt(H_2''')$, observe that $\opt(G) \geq \opt(H_1) + \opt(H_2''')$, since $\opt(G) \geq \opt(H_1) + \opt(H_2')$.
Hence, we have 
\begin{align*}
    || \hat{E}(H_1^*) \cup \hat{E}(H_2^*)  \cup \{v, w\} || \leq & \ \max(\alpha \cdot \opt(H_1) -2, \opt(H_1))  +  \max(\alpha \cdot \opt(H_2'') -2, \opt(H_2''))\\
    \leq & \ \opt(H_1) + \alpha \cdot \opt(H_2'') -2 
    \leq \alpha \cdot \opt(G) -2  \ ,
\end{align*}
where the second inequality follows from the fact that $\opt(H_2') \geq 4$.

In the remaining case, we have that $\opt(H_2') + 1 \leq \opt(H_2''') \leq \opt(H_2) + 2$.
Similar to before we can show that in this case, $\opt(G) \geq \opt(H_1) + \opt(H_2') + 1$.
Using  $\opt(G) \geq \opt(H_1) + \opt(H_2') + 1$ and the fact that $ \opt(H_2''') \leq \opt(H_2') + 2$, we obtain $\opt(G) \geq \opt(H_1) + \opt(H_2''') - 1$.
Hence, we have 
\begin{align*}
    || \hat{E}(H_1^*) \cup \hat{E}(H_2^*)  \cup \{v, w\}|| \leq & \ \max(\alpha \cdot \opt(H_1) -2, \opt(H_1)) +  \max(\alpha \cdot \opt(H_2'') -2, \opt(H_2''))\\
    \leq & \ \opt(H_1) + \alpha \cdot \opt(H_2'') -\alpha -2 
    \leq \alpha \cdot (  \opt(H_1) +  \opt(H_2'')) - 2 - \alpha\\
    \leq & \ \alpha \cdot (  \opt(G) + 1 ) - 2 - \alpha
    \leq \alpha \cdot \opt(G) -2  \ ,
\end{align*}
where the second and third inequality follows from the fact that $\opt(H_2) \geq 4$ and $\opt(H_1) = 3$, respectively, and the fourth inequality follows from the fact that $\opt(G) \geq \opt(H_1) + \opt(H_2') - 1$.

This proves the lemma.
\end{proof}

\begin{lemma}\label{lem:divideCombine:S34-1}
Let $G$ be a graph that has no cut-vertices, no parallel edges and that has an $S_{\{3, 4\}}$ $u_1, u_2, ..., u_k$, $k \in \{3, 4\}$.
If $\opt((G \setminus V_1) / C) \geq 4$ and $\opt((G \setminus V_1) / C) \geq 4$, then there are two MAP instances $H_1$ and $H_2$ such that
\begin{itemize}
    \item $H_i$ is 2EC and $|V(H_i)| < |V(G)|$ for $i \in \{1, 2\}$,
    \item $|V(H_1)| + |V(H_2)| \leq |V(G)| - 2$, and $|E(H_1)| + |E(H_2)| \leq |E(G)|$, and
    \item if there is a 2EC subgraph $H_i^*$ of $H_i$ such that $||H^*_i|| \leq \max( \alpha \cdot \opt(H_i) -2, \opt(H_i))$ for $i \in \{1, 2\}$, then there is an edge-set $F$ with $|F| \leq 2$ such that $E(H^*_1) \cup E(H^*_2) \cup F$ is a feasible solution to the MAP instance on $G$ and $||E(H^*_1) \cup E(H^*_2) \cup F|| \leq \max( \alpha \cdot \opt(G) -2, \opt(G))$.
\end{itemize}
\end{lemma}

\begin{proof}
We prove the statements one by one, define the MAP instances $H_1$ and $H_2$ that are later used for Function $\divided$, and explain how to obtain the edge set $F$ used in Function $\combined$.

Let $C$ be the cycle of cost 2 on the vertices of the $S_{\{3, 4\}}$.
We define $H_1 = G[V(G) \setminus V_2] / C$ and $H_2 = G[V(G) \setminus V_1] / C$.
\begin{figure}
\centering
\begin{subfigure}{.5\textwidth}
  \centering
  \includegraphics[width=.7\linewidth]{pictures/S4_1.png}
\end{subfigure}%
\begin{subfigure}{.5\textwidth}
  \centering
  \includegraphics[width=.7\linewidth]{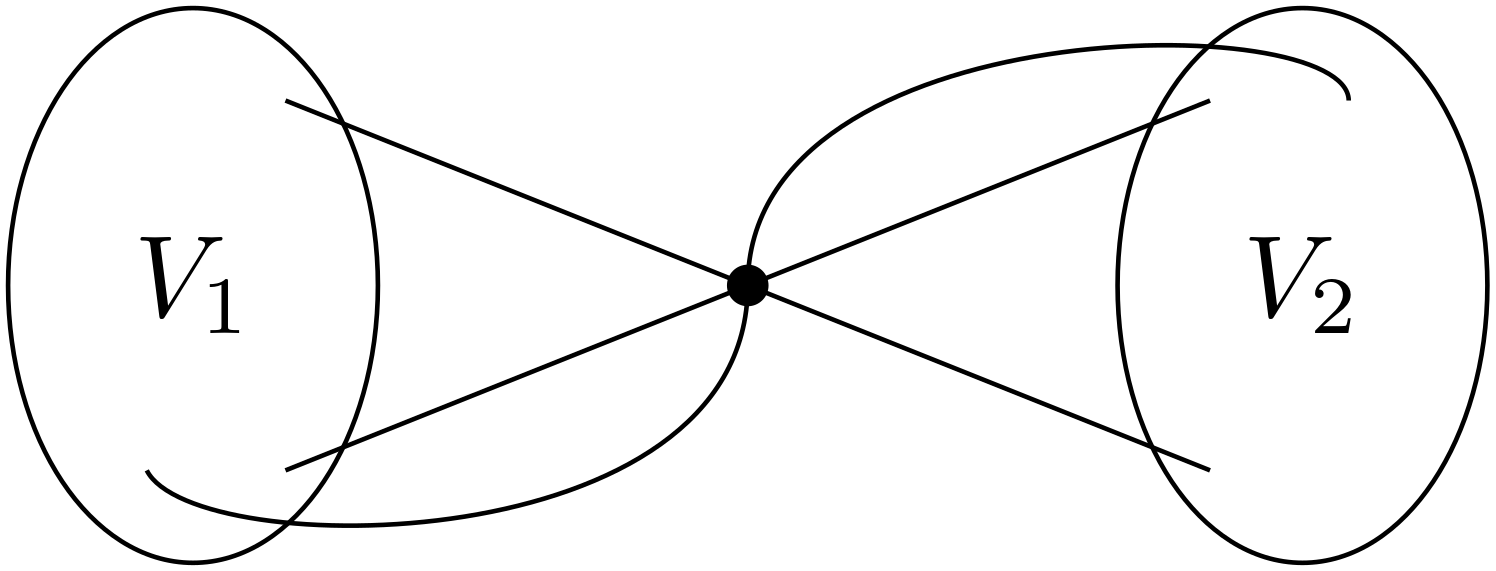}
\end{subfigure}
\caption{Illustration of an $S_{\{3,4\}}$ and $H_2$ for the proof of Lemma~\ref{lem:divideCombine:S34-1}.}
\label{fig:S4_1}
\end{figure}
An illustration is given in Figure~\ref{fig:S4_1}.
By assumption, we have that $\opt(H_1) \geq 4$ and $\opt(H_2) \geq 4$.
Hence, the first two statements directly follow.
Note that $\opt(G) \geq \opt(H_1) + \opt(H_2)$.
Then, for $F = C $, we have that 
\begin{align*}
    || \hat{E}(H_1^*) \cup \hat{E}(H_2^*) \cup F || \leq & \ \alpha \cdot \opt(H_1) -2 + \alpha \cdot \opt(H_2) -2 +2
     \leq  \alpha \cdot \opt(G) +1 \ .
\end{align*}
This proves the lemma.
\end{proof}

\begin{lemma}\label{lem:divideCombine:S34-2}
Let $G$ be a graph that has no cut-vertices, no parallel edges and that has an $S_{\{3, 4\}}$ $u_1, u_2, ..., u_k$, $k \in \{3, 4\}$.
If $\opt((G \setminus V_1) / C) \geq 4$ and $\opt((G \setminus V_1) / C) = 3$, then there are two MAP instances $H_1$ and $H_2$ such that
\begin{itemize}
    \item $H_i$ is 2EC and $|V(H_i)| < |V(G)|$ for $i \in \{1, 2\}$,
    \item $|V(H_1)| + |V(H_2)| \leq |V(G)| - 2$, and $|E(H_1)| + |E(H_2)| \leq |E(G)| + 2$, and
    \item if there is a 2EC subgraph $H_i^*$ of $H_i$ such that $||H^*_i|| \leq \max( \alpha \cdot \opt(H_i) -2, \opt(H_i))$ for $i \in \{1, 2\}$, then there is an edge-set $F$ with $|F| \leq 2$ such that $E(H^*_1) \cup E(H^*_2) \cup F$ is a feasible solution to the MAP instance on $G$ and $||E(H^*_1) \cup E(H^*_2) \cup F|| \leq \max( \alpha \cdot \opt(G) -2, \opt(G))$.
\end{itemize}
\end{lemma}

\begin{proof}
We prove the statements one by one, define the MAP instances $H_1$ and $H_2$ that are later used for Function $\divided$, and explain how to obtain the edge set $F$ used in Function $\combined$.
Let $C$ be the cycle of cost 2 on the vertices of the $S_{\{3, 4\}}$.

First, we prove the statement for the case that the $S_{\{3, 4\}}$ is a triangle, i.e.\ $C$ consists of the vertices $u_1, u_2, u_3$. 
In this case, we define $H_1 = G[V(G) \setminus V_2 ]$ and $H_2 = G[V(G) \setminus V_1 ] / C$.
Hence, the first two statements directly follow.

Fix an optimal solution $\OPT(G)$.
We say that $\OPT(G)$ uses a \emph{link} $\{u, v\}$, $u, v \in V(C)$, if there is a path from $u$ to $v$ in $\opt(G)$ using edges incident to some vertex of $V_2$ only.
We make a case distinction on whether $\OPT(G)$ uses 0, 1, or 2 links.
First, if $\OPT(G)$ uses 0 links, then clearly we have $\opt(G) \geq \opt(H_1) + 3$, since $\opt(H_2) = 3$.
Second, if $\OPT(G)$ uses 1 link, then we have $\opt(G) \geq \opt(H_1) + 2$.
Finally, if $\OPT(G)$ uses 2 links, then we have $\opt(G) \geq \opt(H_1) -2 + 4 = \opt(H_1) +2$.
To see this, observe that on the one hand $\OPT(H_1)$ might have to buy all edges in $C$, which are exactly 2.
On the other hand, in this case, at least 4 edges have to go from $V(H_2)$ to $V(C)$.
Hence, in any case, we have $\opt(G) \geq \opt(H_1) + 2$.

Hence, for $F = \opt(H_2)$ we have 
\begin{align*}
    || \hat{E}(H_1^*) \cup F || \leq & \ \alpha \cdot \opt(H_1) -2 + 3 \leq \alpha \cdot (\opt(G) -2) + 1 \leq \alpha \cdot \opt(G) -2 \ .
\end{align*}

Second, we prove the statement for the case that the $S_{\{3, 4\}}$ is a cycle of length 4, i.e.\ $C$ consists of the vertices $u_1, u_2, u_3, u_4$ (the edges of $C$ are exactly in this order). 
Furthermore, we set $u_5 = u_1$ for notational convenience.
\begin{figure}
\centering
\begin{subfigure}{.5\textwidth}
  \centering
  \includegraphics[width=.7\linewidth]{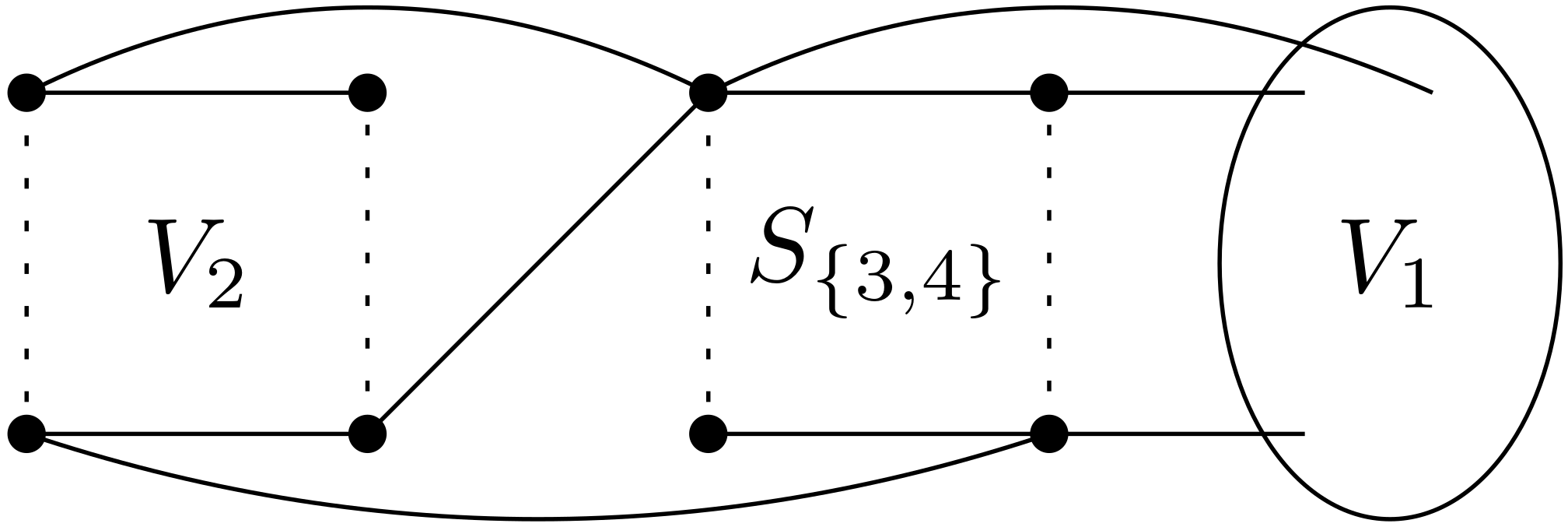}
\end{subfigure}%
\begin{subfigure}{.5\textwidth}
  \centering
  \includegraphics[width=.4\linewidth]{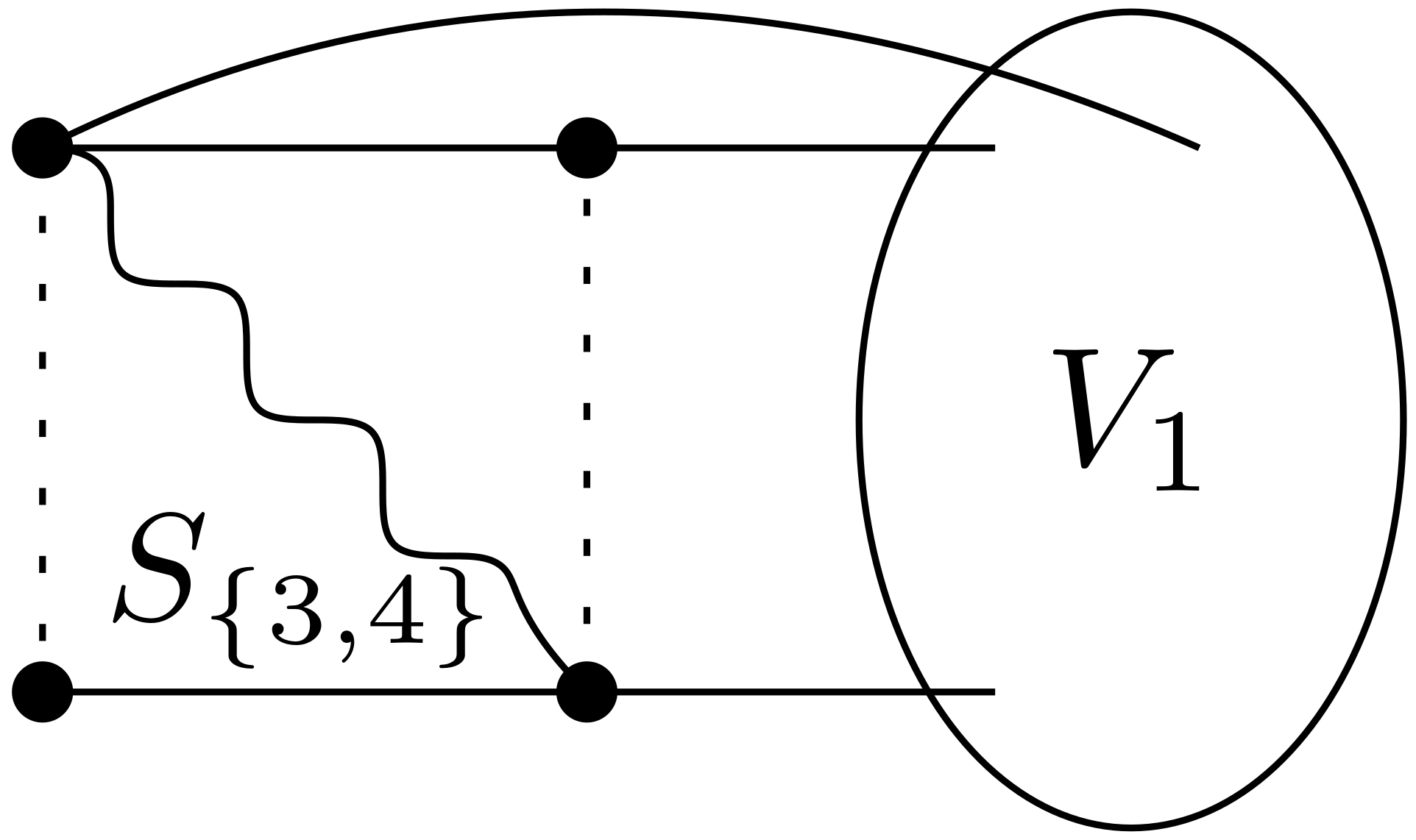}
\end{subfigure}
\caption{Illustration of an $S_{\{3,4\}}$ and $H_2$ for the proof of Lemma~\ref{lem:divideCombine:S34-2}.}
\label{fig:S4-2}
\end{figure}
An illustration is given in Figure~\ref{fig:S4-2}.
Let us fix some optimal solution $\OPT(G)$ and let $H_1' = G[V(G) \setminus V_2 ]$ and $H_2 = G[V(G) \setminus V_1 ] / C$.
Identical to the triangle case, we say that $\OPT(G)$ uses a \emph{link} $\{u, v\}$, $u, v \in V(C)$ if there is a path from $u$ to $v$ in $\OPT(G)$ using edges incident to some vertex of $V_2$ only.

The graph $H_1$ arises from $H_1'$ by adding so-called \emph{pseudo-edges}, which are essentially a subset of the links.
The pseudo-edges are defined as follows.
A pseudo-edge is a 'diagonal' edge $\{u_1, u_3 \}$ or $\{u_2, u_4 \}$ if this edge is not yet present in $H_1$ and if there is a feasible solution $F'$ to $H_2$ such that $|F'| = 3$ and $F'$ is connected to the vertices $u_1$ and $u_3$ (when uncontracted), or to the vertices $u_2$ and $u_4$ (when uncontracted), respectively.
Let $Q$ be the set of pseudo-edges.
Then, $H_1 = H_1' \cup Q$.
Note that the first two statements of the lemma directly follow.

The algorithm proceeds as follows. 
Whenever no pseudo-edge is contained in $E(H_1^*)$, we simply set $\hat{E}(H_1^*)$ as before and set $F = \opt(H_2)$.
If there is exactly one pseudo-edge $e' = \{v, w \}$ contained in $E(H_1^*)$, we simply delete it from $E(H_1^*)$ and set $F = \OPT(H_2)$ such that $F$ connects $v$ and $w$ (such a solution has to exist by the definition of pseudo-edges).
If there are exactly two pseudo-edges contained in $E(H_1^*)$, we delete both pseudo-edges from $E(H_1^*)$, add the edges of $C$ to $E(H_1^*)$, and set $F = \OPT(H_2)$.
Finally, observe that each of these solutions leads to a 2-edge-connected spanning subgraph of $G$ and that $||E(H_1^*) \cup F|| \leq \alpha \cdot \opt(H_1) - 2 + 3 = \alpha \cdot \opt(H_1) + 1$.

Next, we bound $\opt(H_1)$ in terms of $\opt(G)$.
We make a case distinction on how many links are used in $\OPT(G)$.

In the first case, let us assume that $\OPT(G)$ uses no links.
Then we have that $\OPT(G) \setminus V_2$ is 2-edge-connected and get $\opt(G) \geq \opt(H_1) + 3$.

In the second case, let us assume that $\OPT(G)$ uses exactly one link.
If this link connects two vertices $u_i$ and $u_{i+1}$, $1 \leq i \leq 4$, then we have $\opt(G) \geq \opt(H_1) - 1 + 3 = \opt(H_1) + 2$, where the "$-1$" comes from the fact that an optimal solution to $H_1$ might have to buy the edge $\{u_i, u_{i+1} \}$.
Else, if this link connects two vertices that are connected by a pseudo-edge in $H_1$, then we have $\opt(G) \geq \opt(H_1) - 1 + 3 = \opt(H_1) + 2$, by a similar reason as in the previous case.
Else, the link connects either $u_1$ and $u_3$, or $u_2$ and $u_4$ and there does not exists a pseudo-edge between these vertices.
Without loss of generality, assume that $u_1$ and $u_3$ are connected by the link.
Then, by the definition of a pseudo-edge, there is no solution to $H_1$ that has a cost of 3 and additionally connects $u_1$ and $u_3$ (otherwise there is a pseudo-edge between these vertices).
Hence, in this case, $\opt(G) \geq \opt(H_1) - 2 + 4 = \opt(H_1) + 2$, where the "$-2$" comes from the fact that an optimal solution to $H_1$ might have to buy the edges of $C$ in order to connect $u_1$ and $u_3$, instead of the link between $u_1$ and $u_3$.

Finally, assume that $\OPT(G)$ uses at least two links.
Note that in this case, there has to be at least 4 unit edges incident to vertices of $V_2$.
Hence, in this case, $\opt(G) \geq \opt(H_1) - 2 + 4 = \opt(H_1) + 2$, where the "$-2$" again comes from the fact that an optimal solution to $H_1$ might have to buy the edges of $C$ in order to connect the vertices of $C$ that were previously connected via links.

Hence, in any of the three cases, we obtain $\opt(G) \geq \opt(H_1) + 2$.
By the previous discussion, we then obtain 
\begin{align*}
    || E(H_1^*) \cup F || \leq & \ \alpha \cdot \opt(H_1) + 1 \leq \alpha \cdot (\opt(G) -2) + 1 \leq \alpha \cdot \opt(G) -2 \ .
\end{align*}
This proves the lemma.
\end{proof}

\begin{lemma}\label{lem:divideCombine:Sk}
Let $G$ be a graph that has no cut-vertices, no parallel edges and that has an $S_k$ $u_1, u_2, ..., u_k$, $3 \leq k \leq 6$.
Then there are three MAP instances $H_1$, $H_2$, and $H_3$ such that
\begin{itemize}
    \item $H_i$ is 2EC and $|V(H_i)| < |V(G)|$ for $i \in \{1, 2, 3\}$,
    \item $|V(H_1)| + |V(H_2)| + |V(H_3)| \leq |V(G)| + 2$, and $|E(H_1)| + |E(H_2)| \leq |E(G)|$, and
    \item if there is a 2EC subgraph $H_i^*$ of $H_i$ such that $||H^*_i|| \leq \max( \alpha \cdot \opt(H_i) -2, \opt(H_i))$ for $i \in \{1, 2, 3\}$, then there is an edge-set $F$ with $|F| \leq 3$ such that $E(H^*_1) \cup E(H^*_2) \cup F$ is a feasible solution to the MAP instance on $G$ and $||E(H^*_1) \cup E(H^*_2) \cup \cup E(H^*_2) \cup F|| \leq \max( \alpha \cdot \opt(G) -2, \opt(G))$.
\end{itemize}
\end{lemma}

\begin{proof}
We prove the statements one by one, define the MAP instances $H_1$, $H_2$, and $H_3$ that are later used for Function $\dividee$ and explain how to obtain the edge set $F$ used in Function $\combinee$.
Let $C$ be the cycle of cost 3 on the vertices of the $S_k$.
We define $H_1 = G[V(G) \setminus (V_2 \cup V_3)] / \{u_1, u_2, ..., u_k \}$, $H_2 = G[V(G) \setminus (V_1 \cup V_3)] / \{u_1, u_2, ..., u_k \}$, and  $H_3 = G[V(G) \setminus (V_1 \cup V_2)] / \{u_1, u_2, ..., u_k \}$.
Hence, the first two statements clearly follow.

First, note that $\opt(G) \geq \opt(H_1) + \opt(H_2) + \opt(H_3)$.
Here, we also used the fact that $\delta(V(C))$ has no zero-edges.
Second, 
recall that $||E(H_i^*)|| \leq \alpha \cdot \opt(H_i) -2$, $i =1, 2, 3$ (by definition of $S_k$).

Then, for $F = C $, we have that 
\begin{align*}
    || \hat{E}(H_1^*) \cup \hat{E}(H_2^*) \cup \hat{E}(H_3^*) \cup F || \leq & \ \alpha \cdot \opt(H_1) -2 + \alpha \cdot \opt(H_2) -2 + \alpha \cdot \opt(H_3) -2 + 3\\
     \leq & \ \alpha \cdot \opt(G) - 3
     < \alpha \cdot \opt(G) -2 \ .
\end{align*}
This proves the lemma.
\end{proof}

\begin{lemma}\label{lem:divideCombine:Sk'}
Let $G$ be a graph that has no cut-vertices, no parallel edges and that has an $S_k'$ $u_1, u_2, ..., u_k$, $3 \leq k \leq 6$.
Then there are three MAP instances $H_1$, $H_2$, and $H_3$ such that
\begin{itemize}
    \item $H_i$ is 2EC and $|V(H_i)| < |V(G)|$ for $i \in \{1, 2, 3\}$,
    \item $|V(H_1)| + |V(H_2)| + |V(H_3)| \leq |V(G)| + 2$, and $|E(H_1)| + |E(H_2)| \leq |E(G)|$, and
    \item if there is a 2EC subgraph $H_i^*$ of $H_i$ such that $||H^*_i|| \leq \max( \alpha \cdot \opt(H_i) -2, \opt(H_i))$ for $i \in \{1, 2, 3\}$, then there is an edge-set $F$ with $|F| \leq 3$ such that $E(H^*_1) \cup E(H^*_2) \cup F$ is a feasible solution to the MAP instance on $G$ and $||E(H^*_1) \cup E(H^*_2) \cup \cup E(H^*_2) \cup F|| \leq \max( \alpha \cdot \opt(G) -2, \opt(G))$.
\end{itemize}
\end{lemma}

\begin{proof}
We prove the statements one by one, define the MAP instances $H_1$, $H_2$, and $H_3$ that are later used for Function $\dividef$ and explain how to obtain the edge set $F$ used in Function $\combinef$.
Let $C$ be the cycle of cost 3 on the vertices of the $S_k'$.
We define $H_1 = G[V(G) \setminus (V_2 \cup V_3)] / \{u_1, u_2, ..., u_k \}$ and $H_2 = G[V(G) \setminus (V_1 \cup V_3)] / \{u_1, u_2, ..., u_k \}$.
Hence, the first two statements clearly follow.

First, note that since $N(u_1) \subseteq V(S_k')$ (by the definition of an $S_k'$), we have that $\opt(G) \geq \opt(H_1) + \opt(H_2) + 1$.
Here, we also used the fact that $\delta(V(C))$ has no zero-edges.
Second, 
recall that $||E(H_i^*)|| \leq \alpha \cdot \opt(H_i) -2$, $i =1, 2$, since $\opt(H_1) \geq 4$ and $\opt(H_2) \geq 4$, by the definition of an $S_k'$.

Then, for $F = C $, we have that
\begin{align*}
    || \hat{E}(H_1^*) \cup \hat{E}(H_2^*) \cup F || \leq & \ \alpha \cdot \opt(H_1) -2 + \alpha \cdot \opt(H_2) -2 +3 
     \leq \alpha \cdot (\opt(H_1) + \opt(H_2)) - 1 \\
     \leq & \ \alpha \cdot (\opt(G) - 1) - 1 < \alpha \cdot \opt(G) -2 \ .
\end{align*}
This proves the lemma.
\end{proof}

\subsection{Proof of Lemmas~\ref{lem:polytimesubs} and~\ref{lem:divideCombine}}

We are now ready to prove the Lemmas~\ref{lem:polytimesubs} and~\ref{lem:divideCombine}.

\begin{proof}[Proof of Lemma~\ref{lem:polytimesubs}]
The second part of the lemma is simply Lemma~\ref{lem:polytimecheckT}.

To see the first part of the statement, recap that the definitions of $\divideT_T$ and $\combine_T$ are given in the Lemmas~\ref{lem:divideCombine:cutvertices}-\ref{lem:divideCombine:Sk'}.
Therein, it is easy to see that their construction works in polynomial time.
Furthermore, the definition of $\divideT_T$ and $\combine_T$ when $T$ is a contractible subgraph is given in Lemma~\ref{lem:divideCombine:contractibleSubgraphs}.
Therein, it is easy to see that their construction works in polynomial time.
This finishes the proof of the lemma.
\end{proof}

\begin{proof}[Proof of Lemma~\ref{lem:divideCombine}]
The statements (i) and (iii) directly follow from Lemmas~\ref{lem:divideCombine:contractibleSubgraphs}-\ref{lem:divideCombine:Sk'}.
It remains to prove (ii):
For any forbidden configuration $F$ that appears in $G$ of type $T$ from the list $L=$(cut vertex, parallel edge, contractible subgraph, $S_0$, $S_1$, $S_2$, $S_k$, $S_k'$, $k \in \{3, 4, 5, 6\}$) such that $G$ does not contain any forbidden configuration of a type that precedes $T$ in the list $L$ and  $\divideT_T(G,F)=(H_1,H_2, H_3)$, then $s(H_1) + s(H_2) + s(H_3) < s(G)$.

Recap that $|V(G)| \geq 20$. 
Let $V_i$ and $E_i$ be the vertex- and edge-sets of $H_i$, respectively, $i \in \{ 1, 2, 3 \}$.
Then we have by Lemmas~\ref{lem:divideCombine:contractibleSubgraphs}-\ref{lem:divideCombine:Sk'}
\begin{align*}
    s(H_1) + s(H_2) + s(H_3) & = 10 |V_1|^2 + 10 |V_2|^2 + 10 |V_3|^2 + |E_1| + |E_2| + |E_3| \\
    & \leq 10 (|V(G)|-1)^2 + 10 \cdot 3^2 + |E(G)| + 2 \\
    & \leq 10 |V(G)|^2 - 20 |V(G)| + 102 + |E(G)| \\
    & = s(G) - 20 |V(G)| + 102 < s(G) \ ,
\end{align*}
since $|V(G)| \geq 20$.
This proves the lemma.
\end{proof}

\newpage

\section{Bridge Covering}\label{sec:bridgeCovering}
In this section, we show how to obtain an \emph{economical} bridgeless 2-edge-cover $H$ of some structured graph $G$ in polynomial time.

\begin{definition}
An economical bridgeless 2-edge-cover $H$ is a 2-edge-cover of some graph $G$ such that $H$ is bridgeless and $||H|| \leq 1.5 ||D_2|| - 2 n_\ell - 1.5 n_m - n_s$, where $n_\ell, n_m$, and $n_s$ is the number of large, medium, and small components of $H$, respectively.
\end{definition}

First, we need to define some notation.

\begin{definition}[bridgeless, blocks, complex, pendant, black vertices]
Let $H$ be a $2$-edge-cover of some graph.
$H$ is {\bf bridgeless} if it contains no bridges, i.e., all the components of $H$ are $2$-edge-connected. A component of $H$ is {\bf complex} if it contains a bridge.  Maximal\footnote{w.r.t. the subgraph relationship} $2$-edge-connected subgraphs of $H$ are {\bf blocks}. A block $B$ of some component $C$ of $H$ is called \textbf{pendant} if $C \setminus V(B)$ is connected. Vertices of $H$ that are not part of any block are {\bf black}. 
\end{definition}

Observe the following fact.
\begin{fact}
Let $H$ be a $2$-edge-cover of some graph. Then, each component of $H$ is either complex or a block.  Furthermore, a complex component  consists of blocks, black vertices, and bridges. 

If $C$ is a complex component with blocks $B_1,\cdots,B_t$, $s$ black vertices, and $r$ bridges, then $C/\{V(B_1),\cdots,V(B_t)
\}$ forms a tree with exactly $t+s$ vertices and $r$ edges, where all the leaves are contracted blocks (some contracted blocks may be internal) and the edges are in one-to-one correspondence with the bridges of $C$. In particular, $r=t+s-1$.
\end{fact}

Our main theorem of this section is the following.

\thmbridgecover*

We prove the theorem in three steps. 
First, recall that a $D_2$ can be found in polynomial time via an extension of Edmond's matching algorithm.
\begin{fact} \label{fact:D2}
A $D_2$ can be computed in polynomial time.
\end{fact}

Second, we prove that we can transform any $D_2$ into a \emph{canonical} $D_2$.

\begin{definition}
A canonical $D_2$ is a $D_2$ satisfying the following properties:
\begin{itemize}
    \item all zero-edges are contained in the $D_2$,
    \item each small pendant block contains exactly 4 vertices, and
    \item each medium pendant block is incident to a bridge that is a unit-edge.
\end{itemize}
\end{definition}

We then prove the following lemma.

\begin{lemma} \label{lem:canonical-D2}
There is a polynomial-time algorithm that, given a $D_2$ of a structured graph $G$, computes a canonical $D_2$.
\end{lemma}
We prove this lemma in Section~\ref{sec:bridge-covering:canonical-D2}.

Third, we prove that any canonical $D_2$ can be transformed into an economical bridgeless 2-edge-cover.

\begin{lemma} \label{lem:nice-bridgeless-2-edge-cover}
There is a polynomial-time algorithm that, given a canonical $D_2$ of a structured graph~$G$, computes an economical bridgeless 2-edge-cover of $G$.
\end{lemma}
We prove this lemma in Section~\ref{sec:bridge-covering:nice-bridgeless-D2}.

It is easy to see that Theorem~\ref{thm:bridgeCover} follows from Fact~\ref{fact:D2} and Lemmas~\ref{lem:canonical-D2} and~\ref{lem:nice-bridgeless-2-edge-cover}. 

\subsection{Transforming a $D_2$ into a canonical $D_2$}
\label{sec:bridge-covering:canonical-D2}
In this section, we prove Lemma~\ref{lem:canonical-D2}. 
Throughout this section, we assume that $G$ is a structured graph.
Let $H$ be a $D_2$. 
Without loss of generality, we can assume that $H$ contains all zero-edges, since adding these edges does not increase the cost.
In order to transform $H$ into a canonical $D_2$ we make use of the following potential function.

\begin{definition}
Let $H$ be a 2-edge-cover of some graph $G$. Then the potential function $\rho$ is defined by
$$ \rho(H) = n^2 \cdot n_c + n \cdot n_s + n_m \ ,$$
where $n = |V(G)|$, $n_c$ is the number of connected components of $H$, $n_s$ is the number of small blocks of $H$, and $n_m$ is the number of medium blocks of $H$ that is not incident to a unit-edge bridge, respectively.
\end{definition}

Our goal is to modify $H$ by deleting unit edges of $H$ and adding unit edges to $H$ such that the total number of unit edges stays the same while decreasing the potential function~$\rho$ of~$H$.
The iterative application of the following lemma (and the algorithm described in its proof) implies Lemma~\ref{lem:canonical-D2}.

\begin{lemma} \label{lem:D2-reduction-of-small-medium-blocks}
Let $H$ be a $D_2$.  
If $H$ contains a block $B$ that is  small pendant and has less than 4 vertices, then $H$ can be transformed into a graph $H'$ such that $H'$ is a $D_2$ and $\rho (H') < \rho (H)$.
Furthermore, if $H$ contains a block $B$ that is medium with no unit-edge bridge incident to it, then $H$ can be transformed into a graph $H'$ such that $H'$ is a $D_2$ and $\rho (H') < \rho (H)$.
\end{lemma}

\begin{proof}
We prove the statements one by one. So let us first assume that $H$ contains a small block $B$ of some component $C$ that is pendant and has less than 4 vertices.
Note that $B$ has at least 3 vertices since $G$ is structured and therefore does not have parallel edges.
Hence, $B$ is a cycle of size 3 with exactly one zero-edge and two unit edges.
Let $u, v, w$ be the vertices of $B$.

Let us assume $B$ is a pendant small block and assume $u$ is the vertex incident to the unique bridge.
Note that if $v$ and $w$ both have degree 2 in $G$, then $u$ is a cut vertex, a contradiction to the fact that $G$ is structured.
Without loss of generality let $v$ be the vertex that has at least one edge $e$ incident to it in $E(G) \setminus E(H)$.
Then, if $\{ u, v \}$ is a unit-edge, we remove $\{ u, v \}$ from $H$ and add $e$ to $H$ in order to obtain $H'$ and observe that either the number of components has decreased or the newly created block that contains $u, v, w$ has strictly more vertices. 
Hence, $\rho (H') < \rho (H)$.
Else, if $\{ u, v \}$ is a zero-edge, we have that $w$ is incident to 2 unit-edges in $H$.
If $w$ has degree 2 in $G$, then $B$ is a contractable subgraph, a contradiction to the fact that $G$ is structured.
Hence, $w$ has some edge $e$ incident to it in $E(G) - E(H)$.
In this case $H' = H \cup \{e\} \setminus \{ u, w \}$ and observe that $H'$ is a $D_2$.
If $e$ is incident to a component different from $C$, then $\rho (H') < \rho (H)$.
Otherwise, let $B'$ be the block of $H'$ that contains $u, v, w$.
Then, $B'$ contains more vertices than $B$ and we have $\rho (H') < \rho (H)$.

Next, we prove the second statement. 
Let $B$ be a block of some component $C$ that is medium with no unit-edge bridge incident to it. 
If $B$ has 6 vertices, then $B$ is a cycle with 3 unit edges and 3 zero-edges. Hence all bridges incident to $B$ are unit edges.
Thus, we can assume that the number of vertices of $B$ is at most 5.
Furthermore, the number of vertices of $B$ is at least 3, since $G$ is structured and $H$ is a $D_2$.

Next, observe that if $B$ is incident to 2 or more (zero-edge) bridges, then $B$ must contain at least 4 vertices, as otherwise $H$ contains a unit-edge that can be deleted (a contradiction to the fact that $H$ is a $D_2$).
But then $B$ contains at least 4 unit-edges and $B$ is not medium (but large).

Hence, $B$ is incident to precisely one bridge and this is a zero-edge, by assumption.
Let $u$ be the vertex incident to the bridge and $v, w$ be the neighbors of $u$ in $B$. 
Observe that $\{u, v\}$ and $\{u, w \}$ are unit-edges.
If $\{ v, w \} \in E(G)$, then $H' = H \cup \{ v, w \} \setminus \{ u, v \}$ is a $D_2$, but the newly created block $B'$ containing $w$ now contains 4 vertices.
Hence, $\rho (H') < \rho (H)$.

Therefore, we assume that $\{ v, w \} \notin E(G)$.
Similar to before, if $v$ and $w$ have no edges to vertices of $V(G) \setminus V(B)$ in $G$, then $B$ is a $\frac{3}{2}$-contractable subgraph, since the number of unit-edges in $B$ is three and $v$ and $w$ must be incident to distinct unit-edges in $B$ in any optimal solution.
Hence, without loss of generality let $v$ have some edge $e$ incident to a vertex of $V(G) \setminus V(B)$.
Let $H' = H \cup \{e \} \setminus \{u, v\}$.
Clearly, $H'$ is a $D_2$.
If $e$ is incident to a component different from $C$, then the number of components is reduced, and hence $\rho (H') < \rho (H)$.
Else, observe that there is a block $B'$ in $H'$ that contains all vertices of $B$.
Since $B'$ has more vertices than $B$, either $B'$ is large or $B'$ is incident to a bridge that is a unit-edge.
Hence, in either case, $\rho (H') < \rho (H)$.
\end{proof}

\subsection{Transforming a canonical $D_2$ into an economical bridgeless 2-edge-cover}
\label{sec:bridge-covering:nice-bridgeless-D2}
In this section, we prove Lemma~\ref{lem:nice-bridgeless-2-edge-cover}. 
Throughout this section, we assume that $G$ is a structured graph.
Let $H$ be a canonical $D_2$ of $G$.
We transform $H$ into an economical bridgeless 2-edge-cover by adding (not too many) edges.

\subsubsection{Credit Invariant}

We use the following credit scheme. 
Each unit edge of $H$ is assigned a credit of $\frac{13}{8}$ and keeps a credit of $1$ in order to pay for the edge itself. 
Hence, each unit-edge has a left-over credit of $\frac{5}{8}$ and we wish to buy additional edges from this left-over and distribute the remaining credit such that the resulting graph does not contain bridges and the condition of an economical bridgeless 2-edge-cover is met.

We reassign the $\frac{5}{8}$ credit of each unit-edge to blocks, components, and vertices and wish to maintain this assignment throughout the transformation until we end up with an economical bridgeless 2-edge-cover.
We call a vertex \emph{black} if it is not contained in any block and {\em white} otherwise. 

We use the following {\em credit invariant}:
\begin{itemize}
    \item each component receives a component credit of at least 1, called c-credit,
    \item each component that is a small block receives a block credit of $\frac{1}{4}$, called b-credit,
    \item each component that is a medium block receives a block credit of $\frac{7}{8}$,
    \item every other block receives at least one block credit, and
    \item each black vertex $v$ receives $\frac{5}{16} \Deg_H^{(1)}(v)$ credit, called n-credit, where $\Deg_H^{(1)}(v)$ denotes the number of unit-edges incident to $v$ in $G$.
\end{itemize}

Next, we show that at the beginning of the algorithm the graph $H$, the canonical $D_2$ satisfies the following credit invariant.

\begin{lemma}\label{lem:credit-invariant-initialization}
Let $H$ be a canonical $D_2$. Then $H$ satisfies the credit invariant.
\end{lemma}

\begin{proof}
First, note that each unit-edge has a credit of $\frac{5}{8}$ and we distribute the credit of a unit-edge such that each of its end vertices receives a credit of $\frac{5}{16}$. This directly satisfies the credits for black vertices.

Next, consider components that consist of single blocks only. 
Observe that each small, medium, or large block has 2, 3, or at least 4 unit-edges, and hence each component can receive a c-credit of one and a block credit of $\frac{1}{4}$, $\frac{7}{8}$ or $1$, respectively. 
Therefore, the credit invariant holds for these components.

Next, consider a complex component that contains a block that is large. 
This large block has a credit of at least 2 and we can use one credit for the component credit of this component and keep~1 credit for the block credit of this large block. 
Furthermore, observe that the credit of the vertices of each other block can be used to satisfy the block credit of its block.
Hence, these components and blocks satisfy the credit invariant.

Finally, consider a complex component $C$ that contains no large block.
Then $C$ has two pendant blocks $B_1$ and $B_2$ that are medium or small.
By the definition of a canonical $D_2$, $B_1$ and $B_2$ each must contain a unit-edge incident to it.
Hence, the total credit of the vertices of $B_1$ and $B_2$ is at least $6 \cdot \frac{5}{8} \geq 3$.
Thus, we can use this credit of at least 3 for the component credit of 1 and the two block credits of 1 each. 

Hence, the credit invariant is satisfied.
\end{proof}

\subsubsection{Pseudo-ear augmentation}

\begin{figure}
\centering
\begin{subfigure}{.5\textwidth}
  \centering
  \includegraphics[width=.7\linewidth]{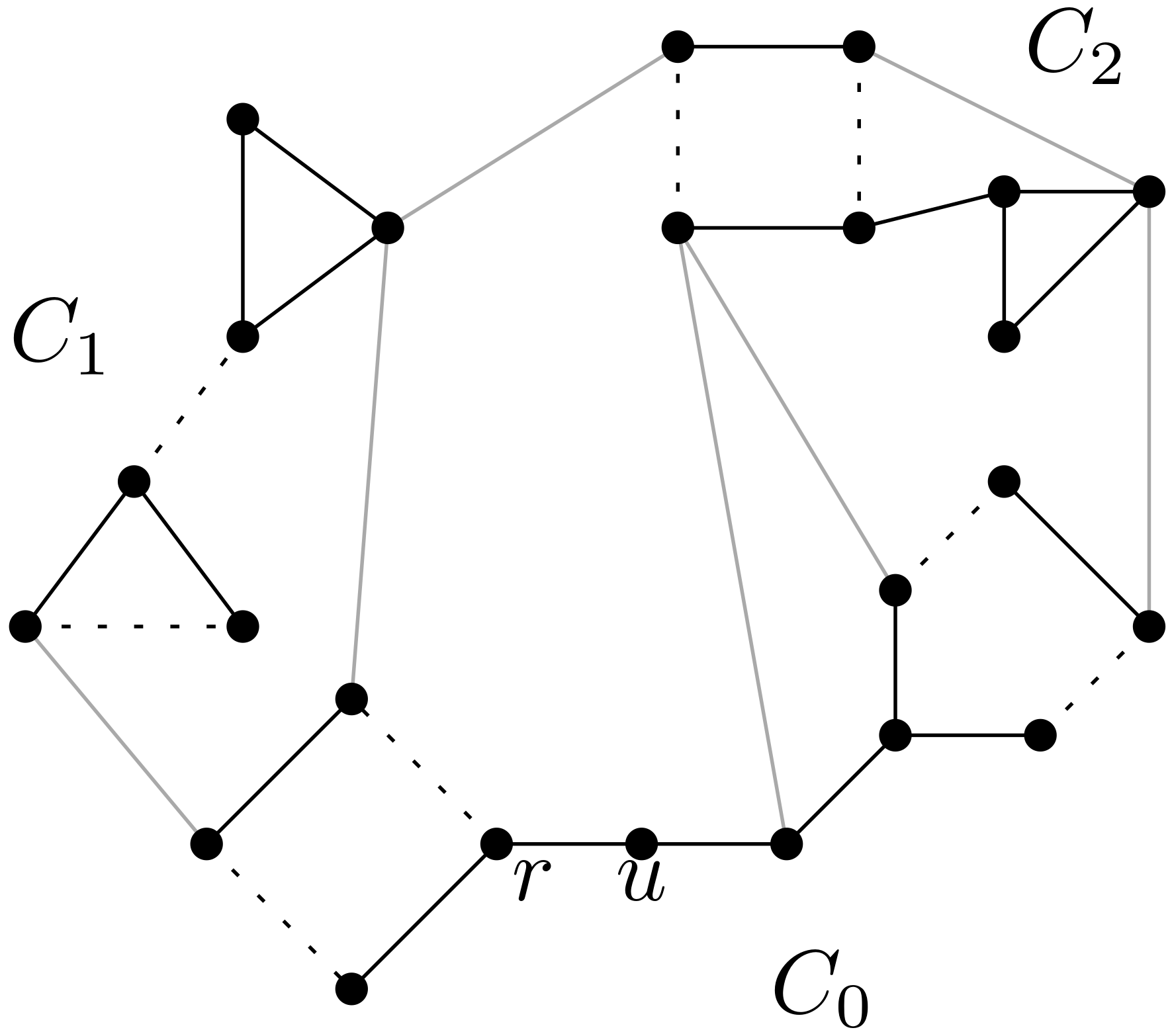}
\end{subfigure}%
\begin{subfigure}{.5\textwidth}
  \centering
  \includegraphics[width=.7\linewidth]{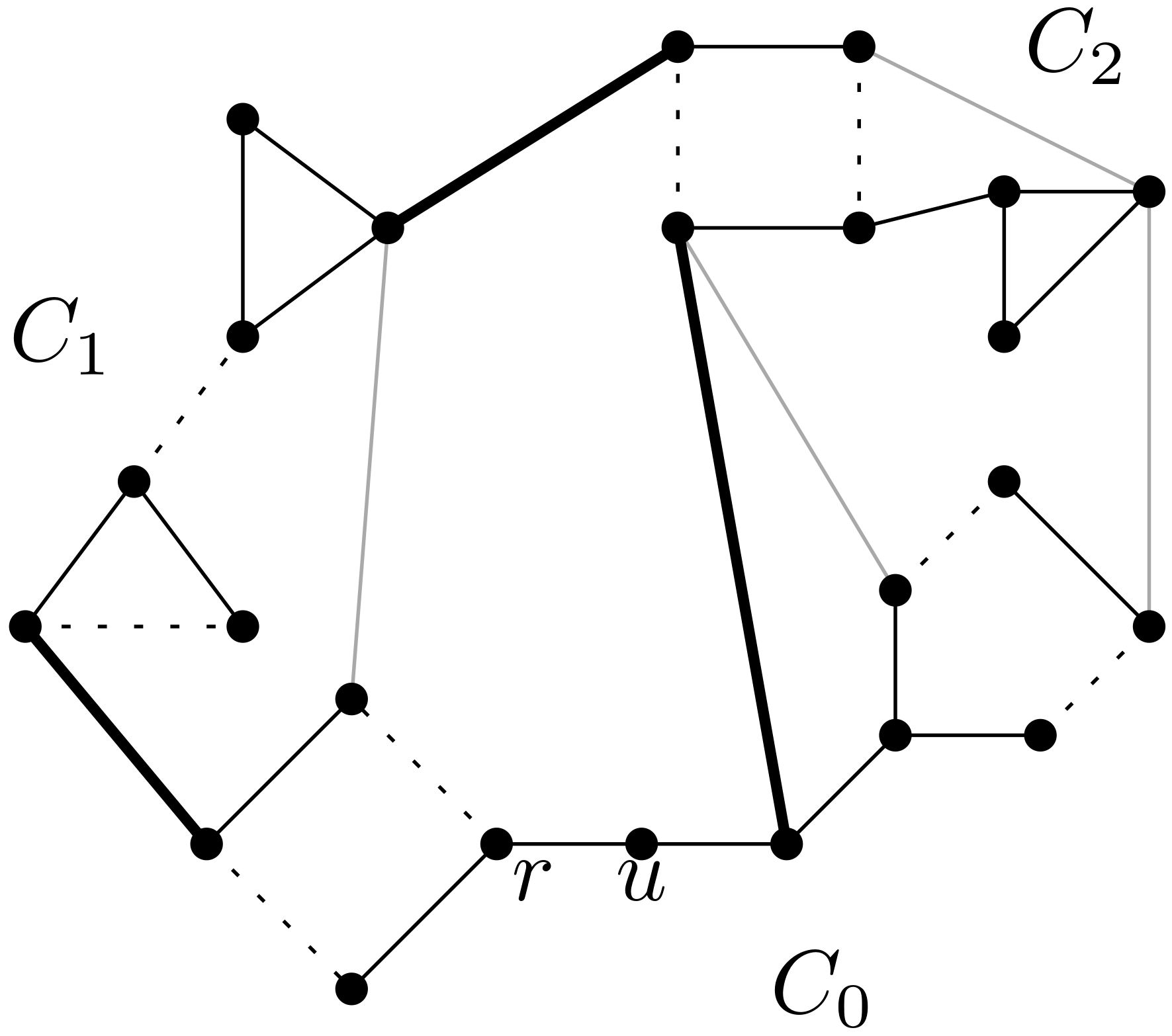}
\end{subfigure}
\caption{Illustration of a bridge-covering step. Zero-edges are dotted, black edges are in $H$ and gray edges are in $G \setminus H$. Bold edges are added in the current step of the pseudo-ear augmentation.}
\label{fig:bridge-covering}
\end{figure}

In order to transform $H$ into an economical bridgeless 2-edge-cover, we proceed as follows.
In each iteration, we pick a component $C$ that has a bridge incident to a pendant block $B$ and add an appropriate set of edges $F \subseteq E(G) \setminus E(H)$ to $H$ such that the unique bridge incident to $B$ is "covered" (afterward, the edge is not a bridge anymore).
We would like to "pay" for this edge by "finding" enough credits, for example by covering enough bridges, which then release their credit since they are contained in a newly formed block. 

Let $\Tilde{H}$ be the graph that arises from $H$ by contracting each block of $H$ to a single vertex.
Similarly, let $\Tilde{G}$ be the graph that arises from $G$ by contracting the blocks of $H$ to single vertices.
In the following, we stick to the notation of~\cite{CheriyanCDZ21} and use the following definition.
Consider some component $C_0$ of $H$ that has a bridge, let $B$ be a pendant 2EC block of $C_0$ and let $\{r, u \}$ be the unique bridge incident to $B$, where $r \in V(B)$.
\begin{definition}
A pseudo-ear of 2-edge-cover $H$ with respect to $C_0$ starting at $B$ is a sequence $B, f_1, C_1, f_2, C_2, ..., f_{k-1}, C_{k-1}, f_k$, where $C_1, C_2, ..., C_{k-1}$ are distinct connected components of $H$, for each $i \in [k-1]$ we have $f_i \in E(G) - E(H)$, $f_i$ has one end vertex in $C_{i-1}$ and one end vertex in $C_i$, and $f_1$ has an end vertex in $B$ and $f_k$ has one end vertex in $C_0 - V(B)$.
The end vertex of $f_k$ in $C_0 - V(B)$ is called the head vertex of the pseudo-ear.

Any shortest path (w.r.t. the number of edges) of $C_0$ between $r$ and the head node of the pseudo-ear is called the witness path of the pseudo-ear.
\end{definition}

First, note that a pseudo-ear is a simple cycle of length $k$ in $\Tilde{H}$, when contracting each connected component to a single vertex.
Second, observe that the block created by the pseudo-ear contains $B$, such that the block credit of $B$ can be transferred to the newly created block.
Third, note that after adding a pseudo-ear to $H$, all components $C_0, C_1, ..., C_{k-1}$ are now in a single component. 
Hence, the component credit of each component visited by the pseudo-ear can be used to pay for the edges of the pseudo-ear. 
However, we still need to "find" one credit in order to obtain the component credit of one for the newly created component.
To do so, observe that all credits of the witness path $Q - r$ are released and can be used for this component credit.
This is exactly our goal in the preceding part of this section.

In each iteration of the pseudo-ear augmentation, we compute a satisfying pseudo-ear covering some bridge using a polynomial time algorithm, which is presented in Lemma~\ref{lem:pseudo-ear:algo-find-credits}.
By \emph{satisfying} we mean that the witness path of the computed pseudo-ear releases a credit of at least 1.
The next lemma suggests which witness paths and therefore pseudo-ears are satisfying, which is an extension of Lemma~17 in~\cite{CheriyanCDZ21}.

\begin{lemma}\label{lem:pseudo-ear:enough-credit}
Let $Q$ be a pseudo-ear of $H$ with respect to $C_0$ starting at $B$, let $P$ be a witness path of $Q$, and let $\{ r, u \}$ be the unique bridge of $C_0$ incident to $B$.
Suppose that $P$ satisfies one of the following:
\begin{itemize}
    \item[a)] $P$ contains a white vertex distinct from $r$, or
    \item[b)] $P$ contains exactly one white vertex and at least 3 unit-edges, or
    \item[c)] $P$ contains exactly one white vertex, 2 unit edges and the head vertex $w$ of $P$ is incident to a unit edge that is not in $P$, or
    \item[d)] $P$ contains exactly one white vertex, 2 unit edges, and the edge $ \{ r, u \}$ is a zero-edge. 
\end{itemize}
Then $P-r$ has at least one credit, and that credit is not needed for the credit invariant of the graph resulting from the pseudo-ear augmentation that adds $Q$ to $H$.
\end{lemma}

\begin{proof}
First, suppose that $Q$ contains a white vertex $u$ distinct from $r$.
Then, the block $B_w$ that contains $w$ is distinct from $B$, and after adding the pseudo-ear $Q$ to $H$, the newly created block $B'$ of $H'$ that contains the vertices of $B$ also contains the vertices of $B_w$.
Recall that $B_u$ is a block of a complex component and hence has a block-credit of 1 and this credit can be released.

Next, consider the case where $Q$ contains exactly one white vertex ($r$) and at least 3 unit-edge bridges. 
In this case, all vertices of $P$ except $r$ can release their credit and this credit is at least 1.

Next, consider the case where $Q$ contains exactly one white vertex ($r$), exactly 2 unit-edge bridges and $u$ is incident to a unit-edge bridge that is not part of $P$.
In this case, note that the 2 unit-edge bridges release a credit of at least $\frac{3}{4}$ since at most $\frac{1}{4}$ of its credit was distributed to $r$ (and therefore is not released).
However, since the head vertex $w$ of $P$ is also incident to a unit-edge bridge $e$ that does not belong to $P$, $u$ receives an additional credit of $\frac{1}{4}$ by $e$.
Hence, in total, we obtain a credit of at least 1.

Finally, if $P$ contains exactly one white vertex, 2 unit-edges, and the edge $ \{ r u \}$ is a zero-edge, then the credit of the 2 unit-edges of $P$ is released. Hence, we obtain a credit of at least 1.
\end{proof}

Finally, similar to~\cite{CheriyanCDZ21}, we show that finding such pseudo-ears can be done in polynomial time.

\begin{lemma} \label{lem:pseudo-ear:algo-find-credits}
There is a polynomial-time algorithm that finds a pseudo-ear (of $H$ with respect to $C_0$ and $B$) such that any witness path $P$ of the pseudo-ear satisfies one of the conditions of Lemma~\ref{lem:pseudo-ear:enough-credit}.
\end{lemma}

\begin{proof}
We use simple case analysis to construct a set $Z$ of vertices of $C_0 \setminus V(B)$ with $|Z| \leq 3$ such that $G[V(G) \setminus Z]$ is connected and $C_0 \setminus (V(B) \cup Z)$ is connected.
Then, there exists a pseudo-ear $Q$ from $B$ to $C_0 \setminus (V(B) \cup Z)$ and it can be found by computing a shortest path in $G[V(G) \setminus Z] \setminus E(C_0)$ between $B$ and $C_0 \setminus (V(B) \cup Z)$.
The construction ensures that any witness path $P$ satisfies one of the conditions of Lemma~\ref{lem:pseudo-ear:enough-credit}.
Let $P'$ be the shortest path from $r$ to any other block in $H$ that uses the edge $ \{r u\}$.
We make a case analysis on the length of $P'$ with respect to the number of vertices. 
Let $P' = r u_1 (= u) u_2 ... u_k$.

If $u_1 = u$ is a white vertex, we set $Z \coloneqq \emptyset$. Then any pseudo-ear and its witness path $P$ satisfies condition a) of Lemma~\ref{lem:pseudo-ear:enough-credit}.

If $u_2$ is a white vertex, then we set $Z \coloneqq \{u_1 \}$. 
Then, since $G$ is structured and therefore has no cut-vertices, there must be a pseudo-ear $Q$ from $B$ to $C_0 \setminus (V(B) \cup Z)$ and its witness path $P$ satisfies condition a) of Lemma~\ref{lem:pseudo-ear:enough-credit}.

If $u_3$ is a white vertex and all edges of $P'$ are unit-edges, we set $Z \coloneqq \{u_1 \}$. 
Then, since $G$ is structured and therefore has no cut-vertices, there must be a pseudo-ear $Q$ from $B$ to $C_0 \setminus (V(B) \cup Z)$ and its witness path $P$ satisfies condition c) of Lemma~\ref{lem:pseudo-ear:enough-credit}.
Else, if $P'$ contains at most 2 unit-edges, we set $Z \coloneqq \{ u_1, u_2 \}$.
Now observe that since $G$ is structured and therefore can not contain cut-vertices, an $S_0$, or an $S_1$, there must be a pseudo-ear $Q$ from $B$ to $C_0 \setminus (V(B) \cup Z)$ and its witness path $P$ satisfies condition a) of Lemma~\ref{lem:pseudo-ear:enough-credit}. 
To see this, assume this is not true. 
Observe that $P'$ can have at most 2 zero-edges and if it contains exactly 2 zero-edges, then the unit-edge $\{ u_2, u_3 \}$ of $P'$ defines an $S_1$.
Else, if $P'$ contains exactly one zero-edge, then the edge $\{ u_2, u_3 \}$ defines an $S_0$ or an $S_1$, depending on whether $\{ u_2, u_3 \}$ is a zero-edge or a unit-edge. But this is a contradiction.

Now, let us assume that $u_4$ is a white vertex.
If $P'$ has exactly 2 unit-edges, we set $Z \coloneqq \{u_1, u_2, u_3 \}$. 
Then, since $G$ is structured and therefore has no cut-vertices, no $S_0$, no $S_1$, and no $S_2$ there must be a pseudo-ear $Q$ from $B$ to $C_0 \setminus (V(B) \cup Z)$ and its witness path $P$ satisfies condition a) of Lemma~\ref{lem:pseudo-ear:enough-credit}.
To see this, assume this is not true. 
Note that if $\{ r, u_1 \}$ is a zero-edge, then by the definition of a canonical $D_2$ the block $B$ must be large.
Hence, either $\{ r, u_1 \}$ is a unit-edge or the block $B$ is large.
In any case, $u_1 u_2 u_3$ defines an $S_2$, a contradiction.

If $P'$ has exactly 3 unit-edges, we set $Z \coloneqq \{u_1, u_2 \}$. 
Then, since $G$ is structured and therefore has no cut-vertices, no $S_0$, no $S_1$, and no $S_2$ there must be a pseudo-ear $Q$ from $B$ to $C_0 \setminus (V(B) \cup Z)$ and its witness path $P$ satisfies condition b), c), or d) of Lemma~\ref{lem:pseudo-ear:enough-credit}.
To see this, assume this is not true. 
Then, depending on whether $\{u_1, u_2 \}$ is a zero-edge or a unit-edge, $u_1 u_2$ defines an $S_0$ or an $S_1$, a contradiction.

If $P'$ has exactly 4 unit-edges, we set $Z \coloneqq \{u_1 \}$. 
Then, since $G$ is structured and therefore has no cut-vertices there must be a pseudo-ear $Q$ from $B$ to $C_0 \setminus (V(B) \cup Z)$ and its witness path $P$ satisfies condition b), c), or d) of Lemma~\ref{lem:pseudo-ear:enough-credit}.

Finally, let us assume that $u_k$ is a white vertex for some $k \geq 5$.
We make the same case distinction as in the case that $u_4$ is a white vertex and use the same arguments.
Let $P'' = r u_1 u_2 u_3 u_4 u_5$. If $P''$ has exactly 2 unit-edges, then we set $Z \coloneqq \{u_1, u_2, u_3 \}$. 
Then, since $G$ is structured and therefore has no cut-vertices, no $S_0$, no $S_1$, and no $S_2$ there must be a pseudo-ear $Q$ from $B$ to $C_0 \setminus (V(B) \cup Z)$ and its witness path $P$ satisfies condition b), c) or d) of Lemma~\ref{lem:pseudo-ear:enough-credit}.
If $P''$ has exactly 3 unit-edges, we set $Z \coloneqq \{u_1, u_2 \}$. 
Then, since $G$ is structured and therefore has no cut-vertices, no $S_0$, no $S_1$, and no $S_2$ there must be a pseudo-ear $Q$ from $B$ to $C_0 \setminus (V(B) \cup Z)$ and its witness path $P$ satisfies condition b) or c) of Lemma~\ref{lem:pseudo-ear:enough-credit}.
If $P'"$ has exactly 4 unit-edges, we set $Z \coloneqq \{u_1 \}$. 
Then, since $G$ is structured and therefore has no cut-vertices there must be a pseudo-ear $Q$ from $B$ to $C_0 \setminus (V(B) \cup Z)$ and its witness path $P$ satisfies condition b), c), or d) of Lemma~\ref{lem:pseudo-ear:enough-credit}.
\end{proof}

The next lemma shows that using Lemmas~\ref{lem:pseudo-ear:enough-credit} and~\ref{lem:pseudo-ear:algo-find-credits} we can cover the  bridges of $H$.
The iterative application of this lemma together with the Lemma~\ref{lem:credit-invariant-initialization} implies Lemma~\ref{lem:nice-bridgeless-2-edge-cover}.

\begin{lemma}\label{lem:pseudo-ear-augmentation-satisfied}
Suppose that $H$ satisfies the credit invariant.
If a pseudo-ear augmentation is applied to $H$, then the resulting graph $H'$ satisfies the credit invariant.
\end{lemma}

\begin{proof}
Let $C_0$ be a component having a bridge $\{r, u \}$, and let $B$ be the block such that $r \in V(B)$ and $u \notin V(B)$.
Let $B, f_1, C_1, f_2, C_2, ..., f_k$ be a pseudo-ear $Q$ found by Lemma~\ref{lem:pseudo-ear:algo-find-credits} with head vertex $v$ and let $P$ be its witness path.
Let $B'$ be the block that is created by adding $Q$ to $H$ such that $V(B) \subset V(B')$ and let $C'$ be the component that is created by adding $Q$ to $H$, such that $V(C_i) \subseteq V(C')$ for all $i \in [k]$.
The block credit of $B$ is simply used for the block credit of $B'$.
The component credit of each $C_i$ is used to pay for the edge $f_i$.
The additional credit of 1 that is needed for the component credit of $C'$ is obtained since we apply the algorithm of Lemma~\ref{lem:pseudo-ear:algo-find-credits}.
All other credits stay the same and hence the graph $H'$ arising from adding $Q$ to $H$ satisfies the credit invariant.
\end{proof}

\newpage
\section{Computing a special configuration}\label{sec:ComputingSpecialConfiguration} 

In this section, we are given as input an economical bridgeless 2-edge-cover $H$ of a structured graph~$G$.
The goal in this section is to compute a \emph{special configuration} of $G$, which is an economical bridgeless 2-edge-cover with some additional properties.
In order to define a special configuration, we first need the following definitions.

\begin{definition}[notation $G / H$, shortcuts, good cycle, savings]
Fix a structured graph $G$ and a bridgeless $2$-edge-cover $H$ of $G$ with components $C_1, C_2,\cdots, C_t$. We use $G / H$ to denote the graph $G / \{C_1, C_2, \cdots, C_t\}$. 
We refer to the vertices of $G/H$ as {\bf nodes} which correspond to the components of $H$.

Let $ C^\circ$ be a node of $G/H$ which corresponds to a medium or small component $C$ of~$H$. 
Let $S$ be a  $2$-edge-connected subgraph of $G/H$ (not necessarily spanning) such that it contains the node $C^\circ$. 
We say $S$ {\bf shortcuts} $C^\circ$, if there exists a simple spanning path $C'$ of $G[V(C)]$ such that $||C|| - ||C'|| = 1$ and the graph obtained from $S$ by expanding its node $C^\circ$ into $C'$, i.e., $((V(S)\setminus \{\hat C\})\cup V(C), E(S)\cup E(C')))$, is $2$-edge-connected.
When $S$ shortcuts a node $C^\circ$, it is possible to expand the node $C^\circ$ with a different set of edges than those in $C$ such that the resulting graph remains $2$-edge-connected but contains $1$ less unit-edge.

We define $\savings (H, S) = \sum_{C^\circ \in S} ||C|| - ||C'||$. 
Informally speaking, $savings(H, S)$ is the number of edges inside the components incident to $S$ that can be saved.
We sometimes omit $H$ in the definition and write $\savings(S)$ if $H$ is clear from the context.


A (simple) cycle $S$ in $G/H$ is {\bf good} if any of the following holds.
\begin{itemize}
\item[(i)] $S$ contains at least two nodes corresponding to large components.
\item[(ii)] $S$ shortcuts at least two nodes.
\item[(iii)] $S$ contains at least one node corresponding to a large component and shortcuts at least one node.
\end{itemize}
\end{definition}

\begin{figure}
\centering
\begin{subfigure}{.5\textwidth}
  \centering
  \includegraphics[width=.7\linewidth]{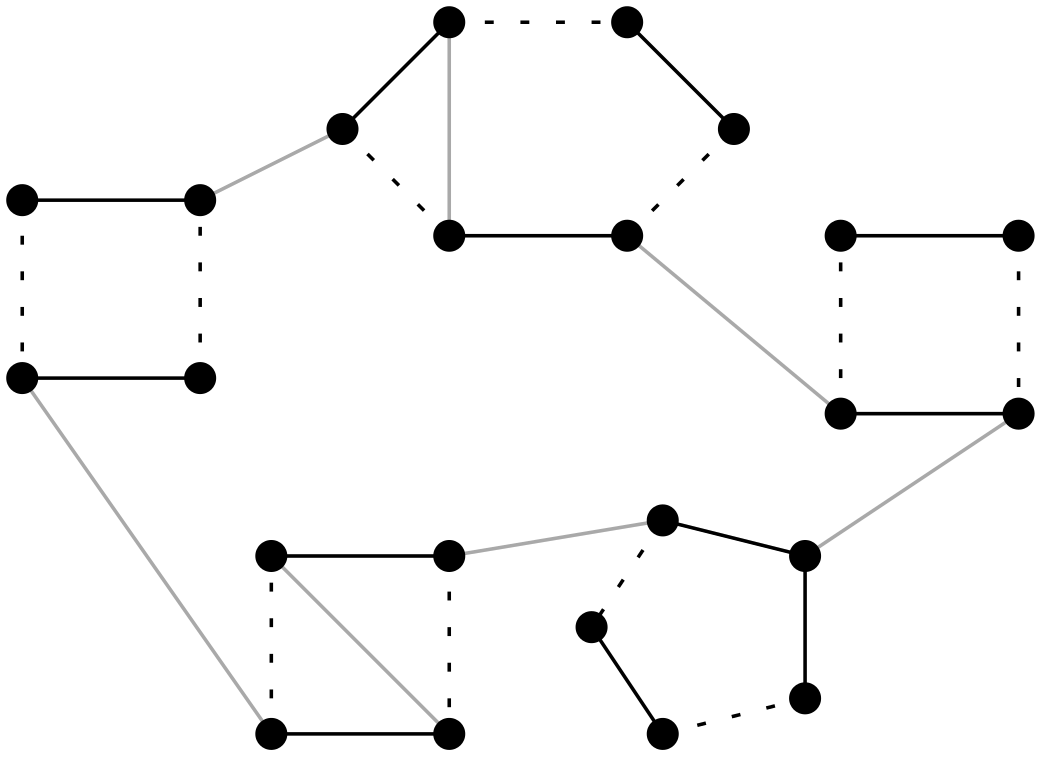}
  \caption{A good cycle.}
  \label{fig:good-cycle-1}
\end{subfigure}%
\begin{subfigure}{.5\textwidth}
  \centering
  \includegraphics[width=.7\linewidth]{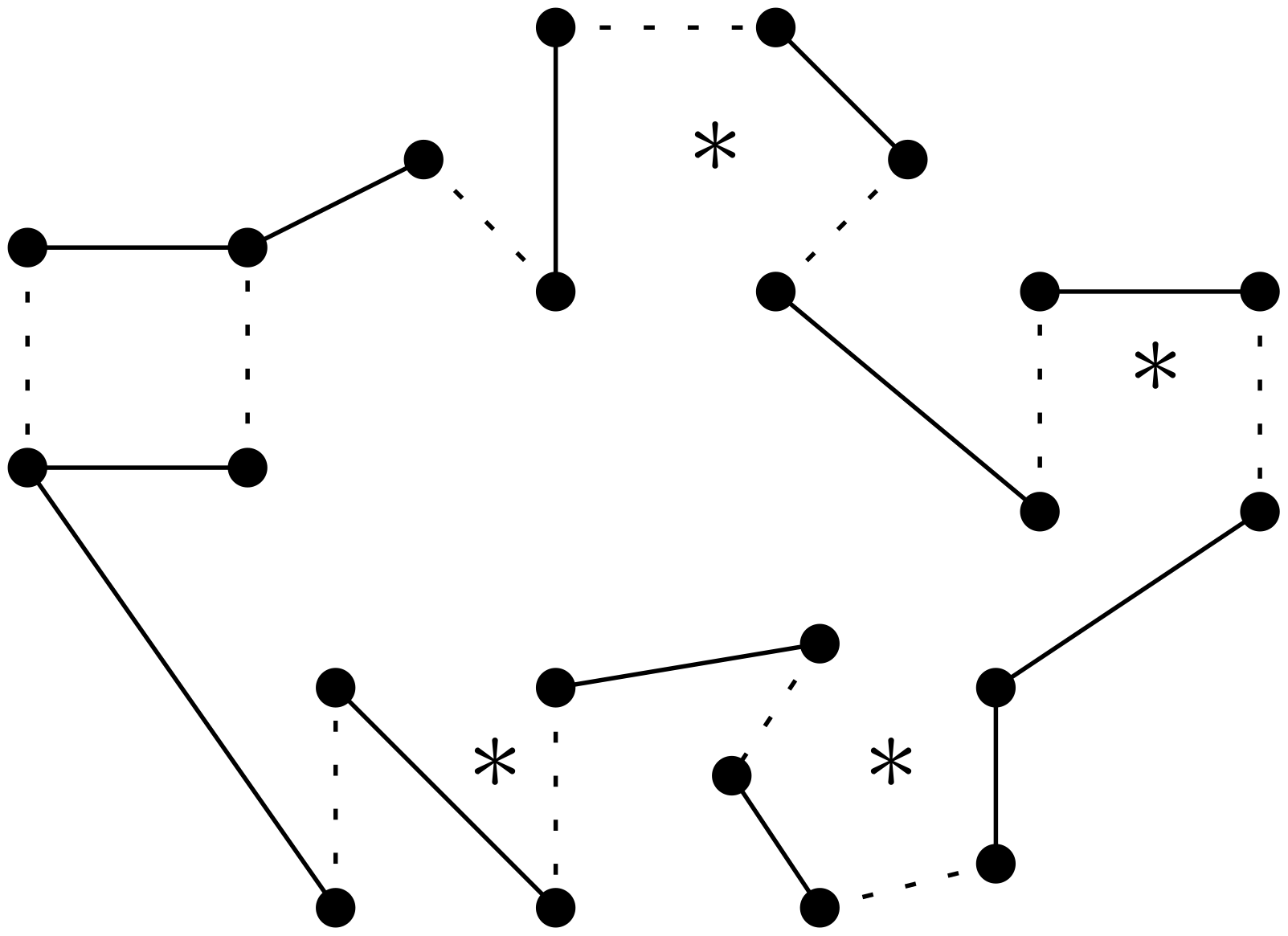}
   \caption{A good cycle with shortcuts.}
  \label{fig:good-cycle-2}
\end{subfigure}
\caption{Example of a good cycle. The unit edges are dotted. In~\ref{fig:good-cycle-1}, the dark edges are contained in the components, the gray edges between the components form the good cycle in $G / H$. Figure~\ref{fig:good-cycle-2} shows how the cycle can be expanded to $G$. Components marked with a $\ast$ are shortcut.}
\label{fig:good-cycle}
\end{figure}

Note, in $S$ if we expand a node $C^\circ$ to its corresponding component $C$, the resulting graph remains $2$-edge-connected. 
An illustration of a good cycle together with shortcuts is given in Figure~\ref{fig:good-cycle}.

\begin{definition}[augmenting paths]
Let $H$ be a bridgeless 2-edge-cover of some structured graph $G$ and let $P$ be a (not necessarily simple) path in $G/H$ of length $k$ such that each interior node corresponds to a small component of $H$. 
We say that $P$ is a \textbf{$k$-augmenting path} w.r.t\ $H$ if for each interior node $C^\circ$ of $P$ there exists a path $P_C$ in $G[V(C)]$, where $C$ is the corresponding component to $C^\circ$, such that
\begin{itemize}
\item[(i)] $P_C$ spans all vertices of $C$, 
\item[(ii)] $P_C$ contains exactly one unit-edge, and
\item[(iii)] $P \cup \left( \bigcup_{C^\circ \text{interior node in } P} P_C \right)$ is a path in $G$.
\end{itemize} 
We say that $P$ is an \textbf{open} $k$-augmenting path if $P$ is additionally simple, i.e.\ each node appears at most once in $P$.
We say that $P$ is a \textbf{closed} $k$-augmenting path if the head and tail of $P$ are the same nodes and each other node appears exactly once.
Note that a closed $k$-augmenting path is a cycle in $G / H$; it is important to specify the ordering of the nodes in $P$ since the conditions (i)-(iii) only hold for the interior nodes.
\end{definition}

\begin{figure}
\centering
\begin{subfigure}{.5\textwidth}
  \centering
  \includegraphics[width=.7\linewidth]{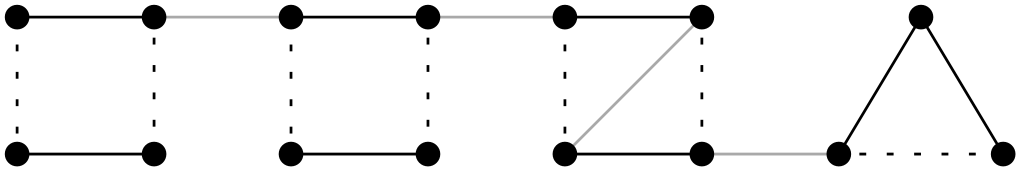}
\end{subfigure}%
\begin{subfigure}{.5\textwidth}
  \centering
  \includegraphics[width=.7\linewidth]{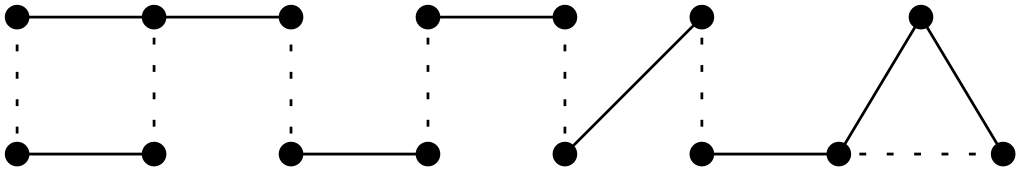}
\end{subfigure}
\caption{Example of an open 3-augmenting path.}
\label{fig:augmenting-path}
\end{figure}

In this section, we only use open $3$-augmenting paths.
The other definitions will be useful in later sections.
For an illustration of an open 3-augmenting path, we refer to Figure~\ref{fig:augmenting-path}.

\begin{definition}[Small to medium merge, small to large merge]
Let $H$ be a bridgeless 2-edge-cover of some structured graph $G$.
We say that $H$ contains a \textbf{small to medium merge} (\textbf{small to large merge}, resp.) if $H$ contains 3 small components $S_1, S_2, S_3$ such that $G[V(S_1) \cup V(S_2) \cup V(S_3)]$ contains two disjoint simple cycles $C_1$ and $C_2$ (one cycle $C_1$, resp.) that span $G[V(S_1) \cup V(S_2) \cup V(S_3)]$ such that $||C_1|| = ||C_2|| = 3$ (such that $||C_1|| = 6$, resp.).
In other words, the three small components can be `merged' into two medium (one large, resp.) components with the same number of unit edges.
\end{definition}

\begin{figure}
\centering
\begin{subfigure}{.5\textwidth}
  \centering
  \includegraphics[width=.7\linewidth]{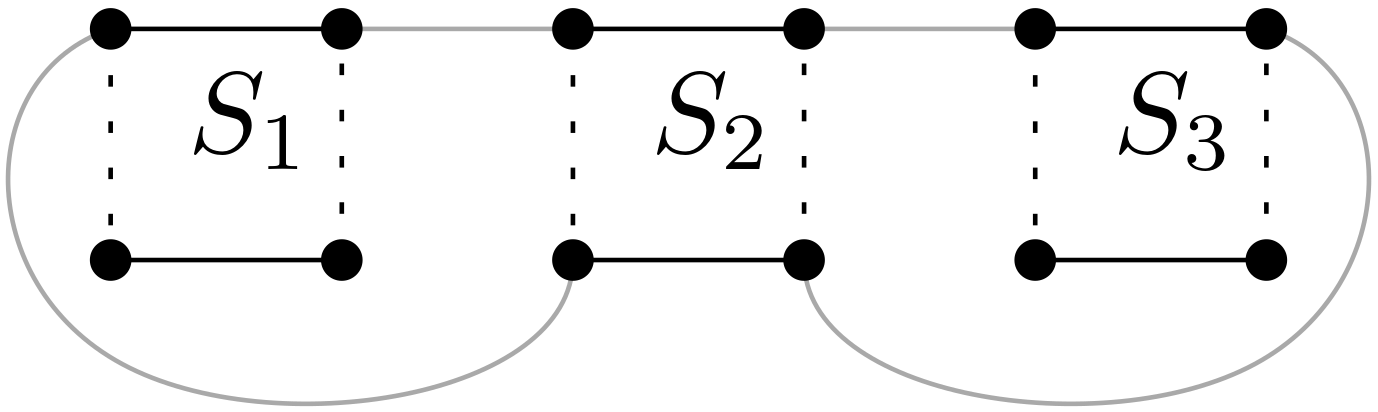}
\end{subfigure}%
\begin{subfigure}{.5\textwidth}
  \centering
  \includegraphics[width=.7\linewidth]{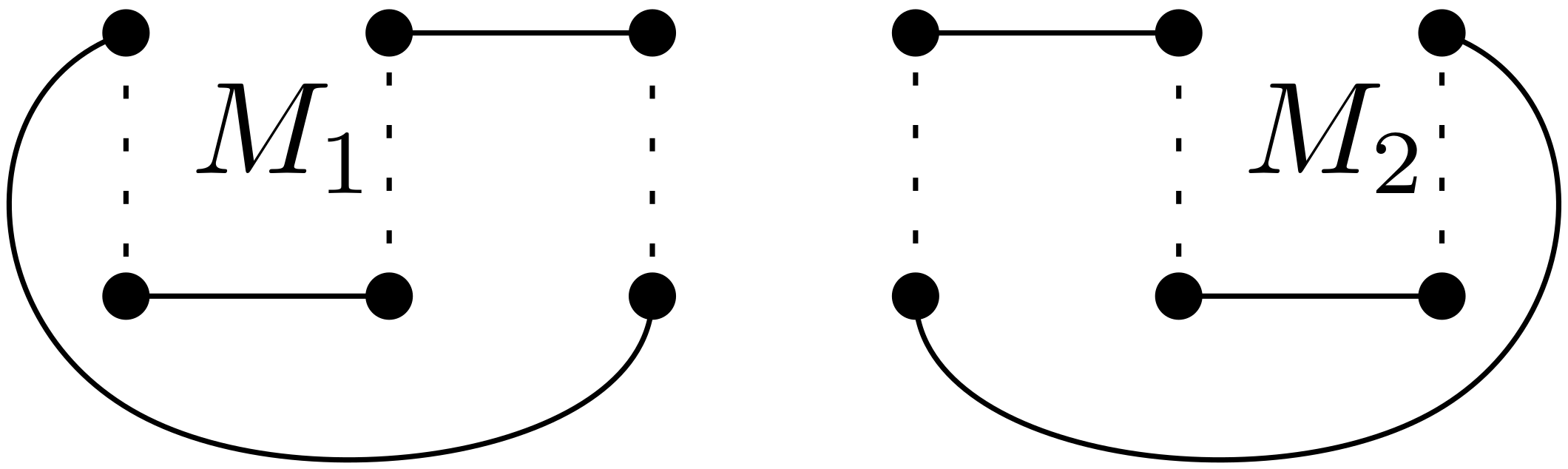}
\end{subfigure}
\caption{Example of a small to medium merge.}
\label{fig:small-to-medium}
\end{figure}

For an illustration of a small to medium merge, we refer to Figure~\ref{fig:small-to-medium}.

\begin{definition}[Special configuration]
Given a structured graph $G$, we say $H$ is a special configuration of $G$ if 
\begin{itemize}
    \item[(i)] $H$ is an economical bridgeless $2$-edge-cover of $G$,
    \item[(ii)] $H$ does not contain medium components,
    \item[(iii)] $G/H$ does not contains good cycles,
    \item[(iv)] $G/H$ does not contain open $3$-augmenting paths, and
    \item[(v)] $H$ does not contain a small to medium merge or small to large merge.
\end{itemize}
\end{definition}

We are now ready to state our main result in this section.

\thmspecialconfig*

To prove the above theorem, we first need to show that we can detect such forbidden configurations.
In order to do so, we will make use of the following key fact, which is proved in~\cite{kawarabayashi2008improved}.

\begin{proposition}
\label{prop:cycle-finding-edges}
Given a graph $G$ and an edge set $F \subseteq E(G)$ such that the number of edges in $F$ is bounded by some constant $k$, we can find a cycle in $G$ that contains all edges of $F$ in polynomial time if one exists.
\end{proposition}

Given this fact, we can check in polynomial time if our economical bridgeless 2-edge-cover $H$ contains any of the four obstructions (ii)-(v).

\begin{lemma}
\label{lem:find-special-obstruction-poly}
Let $H$ be an economical bridgeless 2-edge-cover of some structured graph $G$. 
Then we can find in $H$ in polynomial time any of the four obstructions (ii)-(v) of a special configuration if one exists.
\end{lemma}

\begin{proof}
All obstructions, except for good cycles, only have a constant number of edges and hence we can simply enumerate.
It remains to prove the statement for good cycles.
Consider for example the case in which we are looking for a cycle that shortcuts one small component $S$ and contains one large component $L$ (we enumerate over all such pairs).
We enumerate over all pairs of edges outgoing from $S$ potentially creating a shortcut for $S$ and all pairs of edges outgoing from $L$.
Then, using Proposition~\ref{prop:cycle-finding-edges}, we can check in polynomial time whether there is a simple cycle going through these four edges in $G / H$. 
If there is such a cycle, then we have found a good cycle.
If we did not find a cycle after enumerating over all such components and pairs of edges, then there does not exist a good cycle.
This finishes the proof.
\end{proof}

In our algorithm, we repeatedly search for any of the obstructions (ii)-(v) in $H$ and apply some `merge' step, such that afterward the obtained graph $H'$ is still an economical bridgeless 2-edge-cover, where the number of components of $H'$ is strictly less than the number of components in $H$.
This immediately implies that the running time of the algorithm is polynomial.

In order to do so, similar to prior work, we establish a \emph{credit invariant} as follows.

\begin{itemize}
    \item[(i)] Each large component receives a credit of at least 2,
    \item [(ii)] each medium component receives a credit of $\frac{15}{8}$, and
    \item[(iii)] each small component receives a credit of $\frac{5}{4}$.
\end{itemize}

For a component $C$ of $H$, we write $\credit(C)$ to denote its credit.
It is straightforward to verify that a bridgeless 2-edge-cover $H$ is economical if and only if it satisfies the above credit invariant.
Note that each component of $H$ corresponds to a node in $G / H$ and hence we sometimes say that a node of $G / H$ has certain credit.

The credit invariant is used as follows.
Consider some 2-edge-connected subgraph $S$ of $G / H$.
Note that when adding the edge-set of $S$ to $H$ results in a graph such that all nodes contained in $S$ form a new large component $C_S$ when suitably expanded.
Hence, the number of components in the resulting graph $H \cup S$ is strictly less.
Furthermore, note that the credits of all nodes in $S$ are released and can be used as credit for the newly created large component and it can be used to `pay' for the edges in $S$.
Hence, by the definition of $\savings(H, S)$, the overall credit that can be transferred to the newly created large component $C_S$ is at least 
$$\credit(C_S) = \sum_{C^\circ \in S} \credit(C) + \sum_{C^\circ \in S} \savings(C) - ||S|| \ . $$

In order to satisfy the credit invariant and since the newly created component $C_S$ is necessarily large, we have to argue that $\credit(C_S) \geq 2$.

\subsection{Handling good cycles}
In this subsection, we show how to turn an economical bridgeless 2-edge-cover $H$ of some structured graph $G$ that admits a good cycle $S$ of $G/H$ into an economical bridgeless 2-edge-cover $H'$ with strictly fewer components.
The main statement is the following.

\begin{lemma}
\label{lem:handling-good-cycles}
Let $H$ be an economical bridgeless 2-edge-cover of some structured graph $G$ that admits a good cycle $S$ of $G / H$.
Then there is an economical bridgeless 2-edge-cover $H'$ that has strictly fewer components than $H$.
\end{lemma}

\begin{proof}
Let us assume that $S$ contains some components $C_1, C_2, ..., C_k$ that are all shortcut.
Then, by the definition of a good cycle, we can replace $H[V(C_i)]$ by some spanning subgraph $X_i$ of $C_i$ such that the graph 
\[H' \coloneqq (H \setminus \bigcup_{i \in [k]} H[V(C_i)]) \cup S \cup \bigcup_{i \in [k]} X_i \] 
is a bridgeless 2-edge-cover and all nodes of $S$ are now contained in a single large component.
Hence, it remains to prove that $H'$ is economical.

First, observe that all components not contained in $S$ simply maintain their credits, and hence their credits are satisfied.
Therefore, it remains to check if the credits of the newly created component $C_S$ are satisfied.
Note that each component has a credit of at least 1, and by the definition of a good cycle, it can be easily checked that $S$ satisfies
\[
\savings(H, S) + \sum_{C^\circ \in S} \credit(C) \geq |V(S)| + 2 \ .
\]
Hence, since $S$ is a simple cycle in $G / H$ we have that the newly created large component $C_S$ in $H'$ satisfies
\begin{align*}
\credit(C_S) & = \sum_{C^\circ \in S} \credit(C) + \savings(H, S) - ||S|| \\
& \geq  |V(S)| + 2 - |V(S)| \geq 2 \ .
\end{align*}
This proves the lemma.
\end{proof}

\subsection{Handling a small to medium merge and a small to large merge}
In this subsection, we consider the case that $H$ contains a small to medium merge or a small to large merge and show how to transform it into an economical bridgeless 2-edge-cover with strictly fewer components.
\begin{lemma}
\label{lem:handling-small-to-medium-merge}
Let $H$ be an economical bridgeless 2-edge-cover of some structured graph $G$ that admits a small to medium merge or a small to large merge.
Then there is an economical bridgeless 2-edge-cover $H'$ that has strictly fewer components than $H$.
\end{lemma}
\begin{proof}
If $H$ has a small to medium merge with small components $S_1, S_2, S_3$ and the two cycles $C_1$ and $C_2$, we set
\[ H' \coloneqq (H \setminus \bigcup_{i \in [3]} H[V(S_i)]) \cup C_1 \cup C_2 \ . \]
Note that $H'$ is a bridgeless 2-edge-cover and $H'$ has exactly one less component than $H$.
To see that $H'$ is economical, observe that the number of unit-edges in $H$ and $H'$ is the same but the number of small components has been reduced by 3 and the number of medium components has increased by 2.
Since $3 \cdot \frac{5}{4} = 2 \cdot \frac{15}{8}$, the credit invariant is satisfied.

Analogously it is straightforward to verify that a similar statement holds for small to large merge, and hence the credit invariant holds, since $3 \cdot \frac{5}{4} \geq 2$.
This proves the lemma.
\end{proof}

\subsection{Handling open $3$-augmenting paths}
In this subsection, we consider the case that $H$ contains an open $3$-augmenting path and show how to transform it into an economical bridgeless 2-edge-cover with strictly fewer components when $G / H$ does not contain good cycles.

\begin{lemma}
\label{lem:handling-open}
Let $H$ be an economical bridgeless 2-edge-cover of some structured graph $G$ that admits an open $3$-augmenting path and $G /H$ does not contain good cycles.
Then there is an economical bridgeless 2-edge-cover $H'$ that has strictly fewer components than $H$.
\end{lemma}
\begin{proof}
Let $P= x_1, e_1, x_2, e_2, x_3, e_3, x_4$ be an open $3$-augmenting path in $G / H$ containing the nodes $x_1, x_2, x_3, x_4$ and edges $e_1, e_2, e_3$.
Let $X_i$ be the component in $H$ corresponding to $x_i$, $i \in [4]$.
We want to show that there exists paths $P_1$, $P_2$ in $G/H$ from $x_1$ to $x_3$ and $x_4$ to $x_2$, respectively, such that $P_1$ and $P_2$ are disjoint.
Then we will show how to merge $P$ together with $P_1$ and $P_2$.

First, assume that there is a simple path $P'$ in $G / H$ from $x_1$ to $x_4$ such that $P' \cap \{x_2, x_3 \} = \emptyset$.
Then, by the definition of an open augmenting path, $P \cup P'$ forms a good cycle (since we can shortcut $x_2$ and $x_3$), a contradiction.

We now claim that in $G / H$ there must be a path $P_1$ from $x_1$ to $x_3$ such that $x_2$ does not belong to $P_1$.
Assume this is not true. 
Then $X_2$ is a separator in $G$, which contradicts the fact that $G$ is structured and does not contain an $S_{\{3, 4\}}$ separator.
Hence, such a path $P_1$ exists. 
Note that we can find such a path in polynomial time by finding a path in $G/H \setminus \{x_2 \}$ from $x_1$ to $x_3$.
Furthermore, by the above discussion, we know that $x_4$ does not belong to $P_1$. 
By symmetry, there also has to be a path $P_2$ from $x_4$ to $x_2$ in $G / H$ such that $x_3$ does not belong to $P_2$ and $x_1$ does not belong to $P_2$.
Also, such a path can be found in polynomial time.
Furthermore, observe that $P_1 \cap P_2 = \emptyset$, otherwise we have a good cycle, which can be found by starting at $x_1$, following $P_1$ until we reach a node $P_1 \cap P_2$, and then follow $P_2$ until we reach $x_4$, and then following $P$.

Let $S \coloneqq (V(P_1) \cup V(P_2), E(P) \cup E(P_1) \cup E(P_2))$ be a subgraph of $G / H$.
Let $C_S$ be the set of components incident to $P_1$ and $P_2$ in $G$ and note that $X_i \subseteq C_S, i \in [4]$, by the definitions of $P_1$ and $P_2$, respectively.
Observe that $S$ is a 2-edge-connected subgraph in $G / H$ and hence we can merge the components of $C_S$ to a single 2-edge-connected component $C'_S$.
Hence, the number of connected components has decreased compared to $H$.

It remains to show that the credit invariant is satisfied. 
For a node $C^\circ \in S$ let $C$ be its corresponding component in $H$.
Let $n_1$ and $n_2$ be the number of internal nodes in $P_1$ and $P_2$, respectively.
By the definition of an augmenting path, we have that $\savings(S) \geq 2$.
Then we have that
\begin{align*}
\credit(C_S') & = \sum_{C^\circ \in S} \credit(C) + \savings(H, S) - ||S|| \\
& \geq  \frac{5}{4} \cdot (4 + n_1 + n_2) + 2 - (3 + n_1 + 1 + n_2 + 1) \geq 2 \ ,
\end{align*}
where in the second inequality $4 + n_1 + n_2$ is the number of nodes in $S$, and $3 + n_1 + 1 + n_2 + 1$ is the number of unit-edges in $S$ (3 comes from $P$, $n_1 + 1$ comes from $P_1$, and  $n_2 + 1$ comes from $P_2$).
\end{proof}

\begin{remark}
\label{remark:bottleneck}
Note that when $P_1$ and $P_2$ have lengths exactly 1, then $\credit(C_S') =2$.
Hence, if we work with a lower approximation ratio, we will not be able to merge such an object in this way.
\end{remark}

\subsection{Handling medium size components}
In this subsection we assume that $H$ does not contain good cycles.
Our main lemma of this subsection is the following.
\begin{lemma}
\label{lem:handling-medium-components}
Let $H$ be an economical bridgeless 2-edge-cover of some structured graph $G$ that contains a medium size component $C$ and $G/H$ does not contain good cycles.
Then there is an economical bridgeless 2-edge-cover $H'$ that has strictly fewer components than $H$.
\end{lemma}

In order to prove the statement, we first show the following lemma.

\begin{lemma}
\label{lem:medium-comp-contain-shortcut}
Let $K$ be a cycle of some structured graph $G$ with $||K|| = 3$.
Then there is at least one unit-edge $e$ in $K$ such that each endpoint of $e$ has an edge going to $V(G) \setminus V(K)$. 
\end{lemma}

\begin{proof}
Assume the statement is not true and for each unit-edge $e$ of $K$ there is at most one endpoint of $e$ that has an edge (which we call outgoing edge) going to $V(G) \setminus V(K)$.
We will show that if the statement is not true, then $\OPT(G)$ will include at least 2 unit-edges inside $G[V(K)]$, implying that $K$ is a contractable subgraph, a contradiction to the fact that $G$ is structured.

We make a case distinction on whether $K$ has 3, 4, 5, or 6 vertices.
If $K$ has exactly 6 vertices, then observe that there must be at least 3 vertices of $C$ that do not have outgoing edges.
Hence, any optimal solution has to buy at least 2 unit-edges inside $G[V(K)]$, a contradiction.

Next, assume that $K$ has exactly 5 vertices. 
Let $u$ be the vertex incident to 2 unit-edges of $K$.
If $u$ does not have an outgoing edge, again any optimal solution has to buy at least 2 unit-edges inside $G[V(K)]$, a contradiction.
Else, if $u$ has an outgoing edge, then observe that there are 3 vertices in $V(K)$ that do not have an outgoing edge. 
This again implies that any optimal solution buys at least 2 unit-edges inside $G[V(K)]$, a contradiction.

Next, assume that $K$ has exactly 3 or 4 vertices.
Observe that in this case there is a vertex not incident to a zero-edge of $K$ that has no outgoing edges.
Hence, any optimal solution has to buy at least 2 unit-edges inside $G[V(K)]$, a contradiction.
This proves the lemma.
\end{proof}

We are now ready to prove the main result of this subsection.

\begin{proof}[Proof of Lemma~\ref{lem:handling-medium-components}]
Let $C$ be a medium component and let $C^\circ$ be the node corresponding to $C$ in $G / H$.
Since $C$ is a medium component, $C$ is a cycle containing exactly 3 unit edges.

We make a case distinction on whether $C$ is a separator or not, that is, whether $G \setminus C$ is connected or not.
First, let us assume that $C$ is not a separator.

By Lemma~\ref{lem:medium-comp-contain-shortcut} we know that there is an edge of $C$, say $e = \{u_1, u_2\}$, such that $u_1$ and $u_2$ both have an edge to $V(G) \setminus V(C)$.
Let these two edges be $e_1$ and $e_2$ and let $C_1$ and $C_2$ be the components incident to $e_1$ and $e_2$ that are not $C$ itself. 
Note that it could be that $C_1 = C_2$.
Let $C_1^\circ$ and $C_2^\circ$ be the nodes corresponding to $C_1$ and $C_2$ in $G / H$.
Let $P$ be a path from $C_1^\circ$ to $C_2^\circ$ in $G / H$ such that $C^\circ$ does not belong to $P$.
Note that such a path must exist since $C$ is not a separator.

Let $S$ be the subgraph of $G/H$ consisting of $P$ and $e_1$ and $e_2$, together with all nodes incident to these edges.
Observe that $S$ is a 2-edge-connected subgraph of $G / H$ containing at least 2 nodes.
Furthermore, by the definition of the edges $e_1$ and $e_2$, we have that $\savings(H, S) \geq 1$, since the edge $e$ is not needed for 2-edge connectivity when expanding $S$.
Hence, by adding $E(S)$ to $H$ and deleting $e$ we obtain a 2-edge-cover $H'$ that has strictly fewer components than $H$.

It remains to show that the credit invariant is satisfied. 
Let $C_S'$ be the newly created connected component containing all nodes in $P$ and $C^\circ$.
For a node $C^\circ \in S$ let $C$ be its corresponding component in $H$.
Let $n_P$ be the number of vertices in $P$.
Then we have that
\begin{align*}
\credit(C_S') & \geq \sum_{C^\circ \in S} \credit(C) + \savings(H, S) - ||S|| \\
& \geq \frac{15}{8} + \frac{5}{4} \cdot n_P + 1 - (1 + n_P) \geq 2 \ ,
\end{align*}
which proves the lemma in the case that $C$ is not a separator.

Now, let us assume that $C$ is a separator, i.e. $G - C$ has at least 2 distinct connected components.
If $C$ is a separator in $G$, then $C^\circ$ is a cut-vertex in $G/H$.
We consider three cases: one of the parts is medium, one of the parts is small, and all parts are large.

First, assume that one of the connected components of $G / H \setminus \{ C^\circ \}$ corresponds to a small or medium size component $C'$ of $H$.
In this case we show that there is an edge-set $F = \{f_1, f_2 \} \subseteq E(G)$ of cardinality 2 between $C$ and $C'$ such that $C$ and $C'$ can be merged to a new 2-edge-connected component $C_F$ using $F$ such that the credit invariant is satisfied.

First, assume that $C'$ is medium. 
Then, by Lemma~\ref{lem:medium-comp-contain-shortcut}, $E[C']$ contains two edges $f_1, f_2$ incident to the vertices of a unit-edge $e'$ of $C'$ that go to the vertices of $V(G) \setminus V(C')$.
Since $C$ is a separator, both these edges must go to $C$.
Then adding $f_1$ and $f_2$ to $H$ and deleting $e'$ yields a bridgeless 2-edge-cover $H'$ that has exactly one component less than $H$.
Furthermore, $\savings(H, F) \geq 1$, by definition of $f_1$ and $f_2$.
Hence,
\begin{align*}
\credit(C_F) & \geq 2 \cdot \frac{15}{8} + 1 - 2  \geq 2 \ .
\end{align*}

Next, assume that $C'$ is small.
Since $G$ is structured and hence there is no contractable subgraph, there must be a path $P$ in $G$ with $||P|| \leq 3$, starting and ending in $C$ and visiting all vertices of $C'$ (otherwise, $\OPT(G)$ includes at least 2 unit-edges of $G[V(C')]$, a contradiction, since then $C'$ is a contractable subgraph).
Actually, $||P||=3$ such that $P$ contains exactly one unit-edge inside $G[V(C')]$.
Then adding $P$ to $H$ and deleting all edges of $H[V(C')]$ that do not intersect with $P$ yield a bridgeless 2-edge-cover $H'$ that has strictly fewer components than $H$.
Furthermore, $\savings(H, F) \geq 1$, by definition of $P$.
Hence,
\begin{align*}
\credit(C_F) & \geq \frac{15}{8} + \frac{5}{4} + 1 - 2 \geq 2 \ .
\end{align*}

Now consider the last case: no connected component of $G / H \setminus \{ C^\circ \}$ consists of a single node that corresponds to a small or medium component of $H$.
Therefore, for each connected component $C'$ of $G \setminus C$ we have that $|| H[V(C')] || \geq 4$.
Note that since $G$ is a structured graph, $G$ does not contain a forbidden structure $S_k$, $k = 3, 4, 5, 6$.
Hence, by the above observations, we know that the number of connected components of $G \setminus C$ is exactly 2, and let these two connected components of $G \setminus C$ be $C_1$ and $C_2$.
Furthermore, $G$ also does not contain an $S_k'$, $k = 3, 4, 5, 6$, which implies that each vertex of $C$ has an edge (called outgoing edges) going to the set $V(G) \setminus V(C)$.

Assume first that there is a unit-edge $e$ of $C$ such that both endpoints of $e$ have an outgoing edge to the same connected component.
In this case, we do exactly the same as in the case in which $C$ is not a separator and we are done.

Hence, we can assume that for each unit-edge $e$ of $C$ both vertices incident to $e$ have outgoing edges only to different connected components.
Let us assume that $C$ is a cycle with exactly 3 unit edges and exactly 3 zero-edges (i.e. $C$ is a cycle of length 6). 
The cases in which the number of zero-edges is less are analogous.
Let $C = u_1, v_1, u_2, v_2, u_3, v_3$ be this cycle and let $e_1 = \{u_1, v_1 \}$, $e_2 = \{u_2, v_2 \}$, and $e_3 = \{u_3, v_3 \}$ be the unit-edges of $C$ (see the graph in Figure~\ref{fig:medium-comp-merge:nice-case}).

\begin{figure}
\centering
\begin{subfigure}{.5\textwidth}
  \centering
  \includegraphics[width=.7\linewidth]{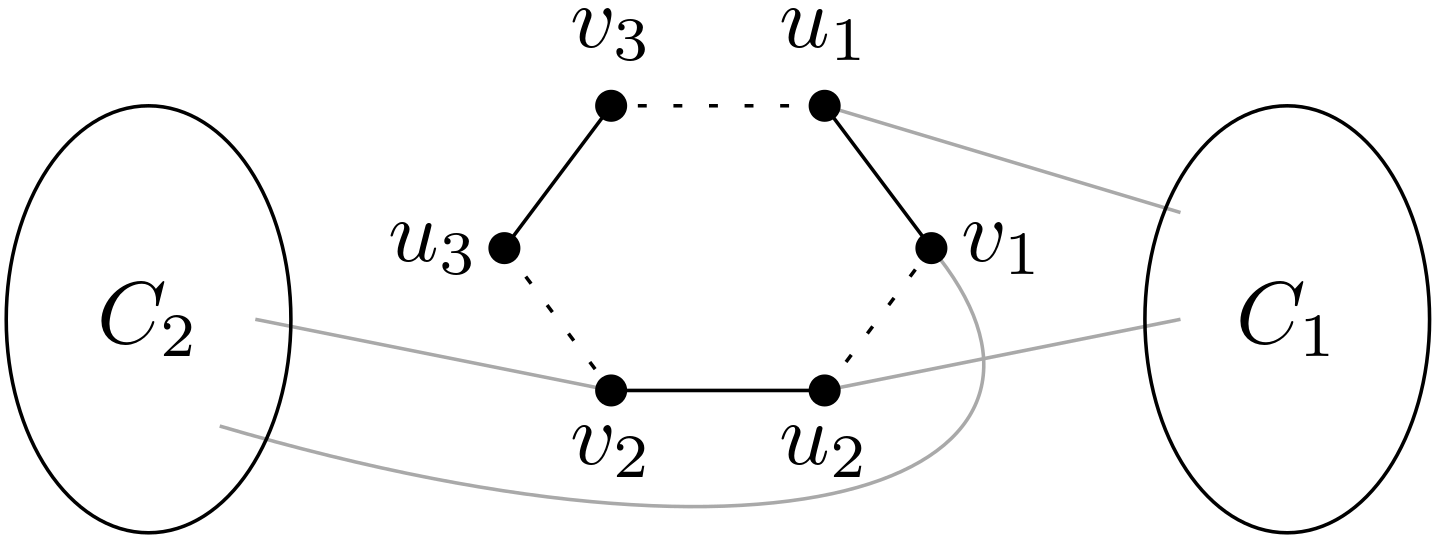}
\end{subfigure}%
\begin{subfigure}{.5\textwidth}
  \centering
  \includegraphics[width=.7\linewidth]{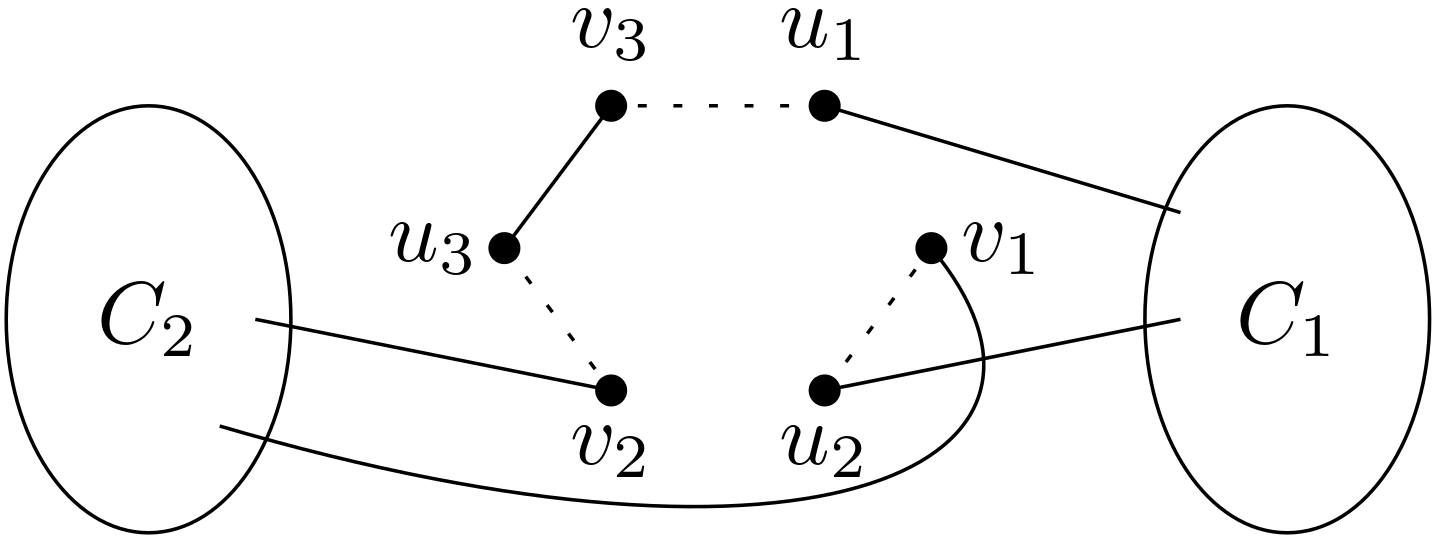}
\end{subfigure}
\caption{Illustration of the `nice case' in the proof of Lemma~\ref{lem:handling-medium-components}.}
\label{fig:medium-comp-merge:nice-case}
\end{figure}

Without loss of generality, we assume that there is an outgoing edge $f_1$ from $u_1$ to $C_1$ and that there is an outgoing edge $g_1$ from $v_1$ to $C_2$.
Next, assume that $u_2$ has an outgoing edge $f_2$ to $C_1$ and $v_2$ has an outgoing edge $g_2$ to $C_2$.
Then, let $F = \{ f_1, f_2, g_1, g_2 \}$.
We refer to this case as the nice case.
For an illustration of this 'nice case', we refer to Figure~\ref{fig:medium-comp-merge:nice-case}.
Furthermore, let $P_1$ and $P_2$ be paths in $G /H$ inside $C_1$ and $C_2$, respectively, that connect the components incident to the edges $f_1$ and $f_2$, and $g_1$ and $g_2$, respectively.
Now observe that $S \coloneqq F \cup P_1 \cup P_2$ is a 2-edge-connected subgraph in $G / H$.
Furthermore, observe that adding $S$ to $H$ and deleting $e_1$ and $e_2$ from $H$ results in a 2-edge-connected subgraph $H'$ that has exactly two fewer components compared to $H$ (see Figure~\ref{fig:medium-comp-merge:nice-case} for an illustration).
It remains to show that the credit invariant is satisfied.

Let $C_S'$ be the newly created connected component containing all nodes in $P_1$, $P_2$, and $C$.
Let $n_1$ be the number of vertices in $P_1$ and $n_2$ be the number of vertices in $P_2$.
Then we have that
\begin{align*}
\credit(C_S') & \geq \frac{15}{8} + \frac{5}{4} \cdot (n_1 + n_2) + 2 - (4 + n_1 - 1 + n_2 - 1) \geq 2 \ .
\end{align*}

Hence, we can assume that the outgoing edge $f_2$ from $u_2$ goes to $C_2$ and that the outgoing edge $g_2$ from $v_2$ goes to $C_1$, i.e., the other way around.

Next, assume that $u_3$ has an outgoing edge $f_3$ to $C_1$ and that $v_3$ has an outgoing edge to $C_2$.
In this case, set $F = \{ f_1, f_3, g_1, g_3 \}$ and observe that this case is analogous to the above where the roles of $e_1$ and $e_2$ are now played by $e_1$ and $e_3$ (which is the nice case from above).

Finally, assume that $u_3$ has an outgoing edge $f_3$ to $C_2$ and that $v_3$ has an outgoing edge to $C_1$.
In this case, set $F = \{ f_2, f_3, g_2, g_3 \}$ and observe that this case is analogous to the above where the roles of $e_1$ and $e_2$ are now played by $e_2$ and $e_3$ (which is the nice case from above).
This covers all relevant cases and hence we have proven the lemma.
\end{proof}

\subsection{Proof of Theorem~\ref{thm:specialConfig}}

\begin{proof}[Proof of Theorem~\ref{thm:specialConfig}]
Let $G$ be a structured graph and $H$ be an economical bridgeless $2$-edge-cover of $G$. 
Our algorithm works as follows. We repeatedly find one of the obstructions according to Lemma~\ref{lem:find-special-obstruction-poly}.
Once we have detected one of the obstructions, we apply one of the Lemmas~\ref{lem:handling-good-cycles} to~\ref{lem:handling-medium-components} and obtain an economical bridgeless 2-edge-cover of $G$, where the number of components has decreased.
Hence, the number of iterations is bounded by $|V(G)|$.
\end{proof}

\newpage
\section{Lower bound (proof of Lemma~\ref{lem:lowerbound})}
\label{sec:lower-bound}

The main goal of this section is to prove the following lemma.
\lemlowerbound*

Throughout this section, we fix a structured graph $G$ together with a special configuration $S$ and a 2-edge-connected spanning subgraph $R$ as specified in the above lemma.

We assume that all small components $H_1, ..., H_k$ of $S$ are cycles of length 4 with precisely 2 unit-edges.
If one of the small components is a cycle of length 3, then we simply replace the vertex $v$ that is incident to 2 unit-edges in the 3-cycle by a zero-edge $\{v_1, v_2 \}$. 
Furthermore, the edges incident to $v$ in $R$ are added to either $v_1$ or $v_2$ as follows:
Since $R$ is $\ecss$, $v$ is contained in some cycle $C$.
We expand $v$ into the edge $\{v_1, v_2 \}$ so that it is part of the expanded cycle $C \cup \{v_1, v_2 \}$.
All the other edges incident to $v$ are now incident to $v_1$.
Observe that the resulting graph remains 2-edge-connected and has precisely the same number of unit edges as $R$.
Also, $||R[V(H_i)]||$ and $||e_R(V(H_i), V(H_j))||$ stay unchanged.
Additionally, note that afterward the instance is still a MAP instance since all zero-edges are contained in the components of the special configuration.

Hence, we can assume that all small components are cycles of length 4 with exactly 2 unit edges.
For a component $H_i$, $i \in [k]$, we let this cycle be $u_1^i, u_2^i, u_3^i, u_4^i$, where the edges $\{u_1^i, u_2^i \}$ and $\{ u_3^i, u_4^i \}$ are the unit-edges of $H_i$. 
Furthermore, we call the edges $\{u^i_1, u^i_3 \}$ and $\{u_2^i, u_4^i \}$ the \emph{diagonal} edges, if present.
We say that $R$ \emph{uses} a diagonal edge if it is contained in $R$.
Without loss of generality, we can assume that if $R$ contains two unit edges inside of the vertices of $H_i$, then these edges are the two unit edges of the cycle $H_i$.
Otherwise, we can simply replace the two unit edges used by $R$ with the unit edges of $H_i$.
Furthermore, for a component $H_i$ let $x_i$ be the corresponding node $x_i$ in $R / H$.
We say that $R$ shortcuts an edge $e$ if $e \in E(H_i) \setminus E(R)$ for some $i \in [k]$.

Before we start with the proof of Lemma~\ref{lem:lowerbound}, we first establish some facts that we obtain from the requirements above.

\begin{fact}
\label{fact:lower-bound}
The graph $G$, the special configuration $S$, and $R$ have the following properties.
\begin{enumerate}
    \item Each small component of $S$ is a cycle of length 4.
    \item If $R$ contains two unit-edges inside of $H_i$, then these edges are the unit-edges of $H_i$.
    \item Let $u \in V(H_i)$ has an outgoing edge $e$ in $G$ to some large component $L$. Let us say $u = u_1^i$. Then, $u_2^i$ does not have any outgoing edge in $G$ to any other component of $S$. Furthermore, if in $G$ there is a diagonal edge $\{ u^i_2, u_4^i \}$ then $u_3^i$ also does not have any outgoing edge in $G$ to any component of $H$. Similar statements can be formulated for $u = u^i_2, u^i_3, u_4^i$.
\end{enumerate}
\end{fact}

\begin{proof}
The proof of the first two statements has been discussed in the discussion preceding this fact.
Hence, it remains to prove the last statement.
To see this, assume this is not true and there is an edge $e' = \{u_2^i, w \}$ such that $w$ is a vertex of a different component $C \neq H_i$.
Then, observe that there is a (empty if $C = L$) path $P$ in $G / H$ from $C$ to $L$ such that $P$ is disjoint from $H_i$, as otherwise $H_i$ is a $S_{\{3, 4\}}$, contradicting that $G$ is structured.
Now, $P$ together with the edges $e$ and $e'$ forms a good cycle, a contradiction to the fact that $S$ is special.
The second part of the last statement follows in the same way by finding a cycle through the diagonal edge shortcutting~$H_i$.
\end{proof}

In particular, the last statement of Fact~\ref{fact:lower-bound} implies that all edges 'used' to shortcut some edge inside some small component $H_i$ are edges between the small components.

Furthermore, we show the following lemma. 
Informally, in presence of open 2-augmenting paths in $R/S$ and specific shortcuts, we show that $R$ must contain certain edges between small components.
Please revisit the definition of augmenting paths and good cycles from Section~\ref{sec:ComputingSpecialConfiguration}.

\begin{lemma}
\label{lem:lower-bound-help-lemma}
Let $P = x_1, x_2, x_3$ be an open 2-augmenting path in $R / S$ with components $H_1, H_2, H_3$, let $P$ start in $u_1^1$ and end in $u_1^3$ when expanded, and $\{u_1^1, u_2^1 \}$ and $\{u_1^3, u_2^3 \}$ are shortcut by $R$.
Then, $R$ must include 2 unit-edges: one edge from either $u_2^1$ or $u_3^1$ to $V(H_2)$, and one edge from either $u_2^3$ or $u_3^3$ to~$V(H_2)$.
\end{lemma}

\begin{proof}
Since the edge $\{u_1^1, u_2^1 \}$ is shortcut in $R$,
there is some unit-edge incident to $u_2^1$ (as $R$ is $\ecss$), which either goes outside to a different component (necessarily small from Fact~\ref{fact:lower-bound}) or is a diagonal edge inside $H_1$. 
In the latter case, observe that there is also one unit-edge incident to $u_3^1$ and by Fact~\ref{fact:lower-bound} this is an outgoing edge to a small component.
Similarly, we can argue that there is an outgoing edge in $R$ from either $u_2^3$ or $u_3^3$ to a small component.
Without loss of generality let us assume that these edges are incident to $u_2^1$ and $u_2^3$, respectively.
If the outgoing edge from $u_2^1$ is incident to some component distinct from $H_2$ and $H_3$, then this edge together with $P$ forms an open 3-augmenting path, a contradiction to the fact that $S$ is special.
If the outgoing edge from $u_2^1$ is incident to $H_3$, then this edge together with $P$ forms a good cycle, again a contradiction to the fact that $S$ is special.
Hence, this edge must go to $H_2$. 
Similarly, the same can be shown for the edge outgoing from $u_2^3$.
\end{proof}

\subsection{Definition of different types of nodes}

In order to prove our main result, we define different types of nodes of $G/S$, where $x_i$ corresponds to the small component $H_i$ of $S$.
Let $A \coloneqq \{x_1, ..., x_k \}$ be the set of all nodes of $G/S$.
We define the following subsets of $A$.
\begin{itemize}
    \item $N(x_i) \coloneqq \{ x_j \mid x_j \text{ is a adjacent in } R / S \text{ to } x_i  \} $ (called neighbors of $x_i$)
    \item $B_2 \coloneqq \{ x_i \mid R \text{ contains no unit-edges from } E(H_i)\}$ (called \emph{double-shortcut} nodes)
    \item $U_2 \coloneqq \{ x_i \in B_2 \mid x_i \text{ is an interior node of an open $2$-augmenting path in } R / S \} $
    \item $B_1 \coloneqq \{ x_i \mid R \text{ contains $1$ unit-edge from } E(H_i)\}$ (called \emph{single-shortcut} nodes)
    \item $U_1 \coloneqq \{ x_i \in B_1 \mid x_i \text{ is an interior node of an open $2$-augmenting path in } R / S \} $
    \item $N_2 \coloneqq \{ x_i  \mid x_i \text{ is a adjacent in } R / S \text{ to a node of } U_2  \} \setminus (U_1 \cup U_2)$
\end{itemize}

Note that $R$ might contain a diagonal edge whose endpoints are in $H_i$ corresponding to an $x_i \in B_2$.
Furthermore, by Fact~\ref{fact:lower-bound}, $R$ cannot contain a diagonal edge with endpoints in $H_i$ corresponding to an $x_i \in B_1$.
We assume without loss of generality that for each $x_i \in B_1$, $\{u_1^i, u_2^i \}$ is the unit-edge of $H_i$ that is contained in $R$.
\begin{itemize}
    \item $N^\top_1(x_i) \coloneqq \{ x_j \mid \exists \{v, w\} \in E(R) \text{ such that } v \in V(H_j) \text{ and } w \in \{ u_1^i, u_2^i \} \}  \setminus \{x_i \}$
    \item $N^\bot_1(x_i) \coloneqq \{ x_j \mid \exists \{v, w\} \in E(R) \text{ such that } v \in V(H_j) \text{ and } w \in \{ u_3^i, u_4^i \} \} \setminus \{x_i \}$
    \item $N^\top_1 \coloneqq  \bigcup_{x_i \in B_1} N^\top_1(x_i) \setminus (U_1 \cup U_2) $
    \item $N^\bot_1 \coloneqq  \bigcup_{x_i \in B_1} N^\bot_1(x_i) \setminus (U_1 \cup U_2) $
    \item $Z \coloneqq A \setminus (B_1 \cup B_2 \cup N_2 \cup N^\top_1 \cup N^\bot_1)$
    \item $I = A \setminus ( U_1 \cup U_2 \cup N_2 \cup N^\top_1 \cup N^\bot_1 \cup Z)$
\end{itemize}

Observe that any node $x_i \in Z$ has the property that its corresponding component $H_i$ satisfies $|| E(H_i) \setminus E(R)|| = 0$.
Furthermore, any $I$-node is a node in $B_1$ or $B_2$ but not a $U_1$ or $U_2$ node or a neighbor of one of them.
This finishes the definition of different types of nodes. We have the following useful propositions.

\begin{proposition}\label{prop:setRelation1}
$A\setminus Z = I \cup U_1 \cup U_2 \cup N_1^\top \cup N_1^\bot \cup N_2.$
\end{proposition}

\begin{proof}
Notice from the definition of $Z$ above, $ Z$ and  $(B_1 \cup B_2 \cup N_2 \cup N_1^\top \cup N_1^\bot)$ are disjoint. Now since we have $U_1\subseteq B_1$ and $U_2 \subseteq B_2$, we have $Z$ is disjoint from $(U_1 \cup U_2 \cup N_2 \cup N_1^\top \cup N_1^\bot)$.

Now from the definition of $I$ above, $A$ is partitioned by the sets $I$ and $U_1\cup U_2\cup N_2 \cup N_1^\top \cup N_1^\bot \cup Z$. Since we noticed  that $Z$ is disjoint from $(U_1 \cup U_2 \cup N_2 \cup N_1^\top \cup N_1^\bot)$, we have $A$ is partitioned by the sets $I$, $Z$ and $U_1\cup U_2\cup N_2 \cup N_1^\top \cup N_1^\bot$. Thus, $A\setminus Z = I \cup U_1 \cup U_2 \cup N_1^\top \cup N_1^\bot \cup N_2.$
\end{proof}

\begin{proposition}\label{prop:setRelation2}
The sets $I$, $U_1$, $U_2$, $N_1^\top\cup N_1^\bot \cup N_2$ are pairwise disjoint.
\end{proposition}

\begin{proof}
$I$ is disjoint from the rest of the sets from the definition of $I$.

Now we have to show $U_1$, $U_2$, and $N_1^\top\cup N_1^\bot \cup N_2$ are pairwise disjoint. 
Since each of $N_1^\top$, $N_1^\bot$ and $N_2$ in their definition explicitly exclude $U_1\cup U_2$, we have that $N_1^\top\cup N_1^\bot \cup N_2$ does not intersect with $U_1$, and also does not intersect with $U_2$.

Now $U_1\subseteq B_1$ and $U_2\subseteq B_2$ are clearly disjoint since one consists of double shortcut nodes, whereas the other consists of single shortcut nodes.
\end{proof}

From the above two propositions, the following proposition is immediate.

\begin{proposition}\label{prop:setRelation3}
The sets $I$, $U_1$, $U_2$, $N_1^\top$, $N_2\setminus N_1^T$, $N_1^\bot \setminus (N_1^\top \cup N_2)$ partition $A\setminus Z$.

\end{proposition}

\subsection{Proof of Lemma~\ref{lem:lowerbound}}
 
 Let us order the set of small components $H_1, ..., H_k$ such that the nodes in $Z$ appear at the end.
 That is, the nodes $x_1, ..., x_\ell$ corresponding to $H_1, ..., H_\ell$ are all contained in $A \setminus Z$ and the nodes $x_{\ell +1}, ..., x_k$ corresponding to $H_{\ell +1}, ..., H_k$ are all contained in $Z$.
 Note that $R$ does not perform any shortcut in nodes $x_i$, for $i  > \ell$.
 We now give the following stronger statement compared to Lemma~\ref{lem:lowerbound}.
 
 \begin{lemma}
\label{lem:lower-bound-stronger}
 $\sum_{i<j\leq k}e_R(V(H_i),V(H_j)) + \sum_{i\in[\ell]} || R[V(H_i)] || \geq (1 + \frac {1}{12}) 2 \ell$.
\end{lemma}

We first show that Lemma~\ref{lem:lower-bound-stronger} implies Lemma~\ref{lem:lowerbound}.
\begin{proof}[Proof of Lemma~\ref{lem:lowerbound}]

For ease of readability, let us define 
\begin{itemize}
    \item $r = \sum_{i<j\leq k}e_R(V(H_i),V(H_j))$
    \item $s = \sum_{i\in [k]} || R[V(H_i)] ||$
     \item $m =\sum_{i\in [\ell]} || R[V(H_i)] ||$ \ .
\end{itemize}
Using these definitions,
we are given (*) $t=2k - s$ and need to show $r \geq (1 + \frac 1{12})t$.

From Lemma~\ref{lem:lower-bound-stronger}, we have $r +  m \geq (1+\frac 1{12})2\ell$.
Since $R$ does not perform any shortcuts in the nodes $x_i$ for $i>\ell$, we have $||R[V(H_i)]|| =2 $ for all $i>\ell$.
Thus observe $s= m + 2(k-\ell)$. Combined with (*), we have, $s= 2k - t = m + 2k - 2\ell \implies 2\ell = m + t $ (**).

From Lemma~\ref{lem:lower-bound-stronger}, we have $r +  m \geq (1+\frac 1{12})2\ell$. Thus,
$$r \geq  (1+\frac 1{12})2\ell - m \overset{(**)}{=} (1+\frac 1{12})2\ell - (2\ell - t) = \frac 1{12}\cdot 2\ell + t \geq (1+\frac 1{12})t,$$
where the last inequality follows from $2\ell \overset{(**)}{=} m+t \geq t$.
\end{proof}

\subsection{Proof of Lemma~\ref{lem:lower-bound-stronger}}
Thus, it remains to prove Lemma~\ref{lem:lower-bound-stronger}.

We first need the following definitions of the set $T$ and a valid assignment.
\begin{definition}[the set T]

$$T \coloneqq \bigcup_{i<j\leq k}E_R(V(H_i),V(H_j)) \cup \bigcup_{i<j\leq \ell} E(R[V(H_i)]),$$ where $E_R(V(H_i),V(H_j))$ denotes the set of edges between the vertices of $V(H_i)$ and $V(H_j)$ in $R$.

\end{definition}

\begin{definition}[valid assignment]
We say that an assignment $\load: \{x_1,\cdots,x_\ell\}\rightarrow \mathbb{R}$ is valid if the following holds.
\begin{enumerate}
    \item $\load(x_i)\geq 0$ for all $i\in[\ell]$.\label{cond1}
    \item $\sum_{i\in [\ell]}\load(x_i) = ||T||$.\label{cond2}
    \item For each node $x_i \in A \setminus (Z \cup I)$ we have $\load(x_i) \geq 2$.\label{lb1}
    \item For each node $x_i \in N_1^\top$ we have $\load(x_i) \geq \frac{5}{2}$.\label{lb2}
    \item For each node $x_i \in U_1$ we have $\load(x_i) \geq \frac{5}{2} + \frac{1}{2} \cdot \max(| N^\bot_1(x_i)| - 2, 0)$.\label{lb3}
    \item For each node $x_i \in U_2$ we have $\load(x_i) \geq \frac{5}{2} + \frac{1}{2} \cdot \max(|N(x_i)| - 2, 0)$.\label{lb4}
    \item For each node $x_i \in I$ we have $\load(x_i) \geq \frac{7}{3}$.\label{lb5}
\end{enumerate}
\end{definition}

We have the following lemma.

\begin{lemma}\label{lem:validAssignment}
There exists a valid assignment.
\end{lemma}

We prove Lemma~\ref{lem:lower-bound-stronger} assuming Lemma~\ref{lem:validAssignment}, and then proof Lemma~\ref{lem:validAssignment} later.
We will use the following simple fact.

\begin{fact}\label{fact:bipartite}
Given a bipartite graph $B=(L,R,F)$ such that every vertex in $R$ has at least one edge incident on it, we have
$$ \frac 1{|L|} \sum_{u\in L} \max(\deg(u)-2, 0)\geq \max\left(\frac{|R|}{|L|}-2,0 \right).$$
\end{fact}
\begin{proof}
Since every vertex in $R$ has at least one edge incident on it, we have
$|R|\leq \sum_{u\in L} \deg(u)$.
Now we have two cases.
Case 1. $\sum_{u\in L} \deg(u) \leq 2|L|$. Here, $\frac {|R|}{|L|} - 2  \leq \frac {\sum_{u\in L} \deg(u)}{|L|} -2 \leq 2 - 2 = 0$. 
Thus, $\max\left(\frac{|R|}{|L|}-2,0 \right)=0$, and we are done since the left-hand side is non-negative.

Case 2.  $\sum_{u\in L} \deg(u) = 2|L| + t$ for some $t\geq 1$. In this case, $\frac{|R|}{|L|} -2 \leq \frac t{|L|}$. Thus, $\max\left(\frac{|R|}{|L|}-2,0 \right) \leq \frac t{|L|}$, since $\frac t{|L|} > 0$. Now, the left hand side is $\frac 1{|L|} \sum_{u\in L} \max(\deg(u)-2, 0) \geq \frac 1{|L|} \sum_{u\in L} \deg(u) -2|L|\geq\frac t{|L|}$ .

\end{proof}

\begin{proof}[Proof of Lemma~\ref{lem:lower-bound-stronger}]
We fix a valid assignment $\load$.
Thus, $\sum_{i\in [\ell]} \load(x_i) = ||T||$. 
We need to show $||T||\geq (1 +\frac{1}{12}) 2\ell$. 
In other words,
we have to show that the average load of a node in $ \{x_1,\cdots, x_\ell\}= A \setminus Z $ is lower bounded by $2(1 + \frac{1}{12})$.

Now from Proposition~\ref{prop:setRelation3}, we have that the sets $I$, $U_1$, $U_2$, $N_1^\top$, $N_2\setminus N_1^T$, $N_1^\bot \setminus (N_1^\top \cup N_2)$ partition $A\setminus Z$.
Thus, the sets $S_1:= I$, $S_2:=N_1^T$, $S_3:= U_2 \cup (N_2\setminus N_1^T)$, $S_4:= U_1 \cup (N_1^\bot \setminus (N_1^\top \cup N_2))$ also partition $A \setminus Z$. 

For each of the above four sets 
$S_i$'s we will show that the average load of the nodes in them is at least $2(1+\frac 1{12})$ and we will be done.

\paragraph*{average load in $S_1$ and $S_2$:}By Conditions~\ref{lb2} and~\ref{lb5} in the definition of a valid assignment, we have that the load of every node in $N_1^\top$ or $I$ is at least $\frac{5}{2}$ or $\frac{7}{3}$, respectively. Since both numbers are larger than $2 \cdot (1 + \frac{1}{12})$, we are done for sets $S_1$ and $S_2$.

\paragraph*{average load in $S_3$:}
From Condition \ref{lb4}, the average load in $U_2$ is
\begin{align*}
    \frac 1{|U_2|}\sum_{x_i\in U_2}\load(x_i) &\geq \frac 1{|U_2|}\sum_{x_i\in U_2}\left(\frac{5}{2} + \frac{1}{2} \cdot \max(|N(x_i)| - 2, 0)\right)\\
    &= \frac 5 2 + \frac 1{2|U_2|}\sum_{x_i\in U_2}\max(|N(x_i)| - 2, 0)\\
    &\geq \frac 5 2 + \frac 1{2|U_2|}\sum_{x_i\in U_2}\max(|N(x_i)\cap (N_2\setminus N_1^T) | - 2, 0)
\end{align*}
Now, we will use Fact~\ref{fact:bipartite} to lower bound the last term. 
We build a bipartite graph $B$ by setting $L= U_2$, $R=N_2\setminus N_1^T$ and for each vertex $x_i\in L$ and $x_j \in R$ there is an edge iff $x_j\in N(x_i)$. 
Observe then $\deg_B(x_u) =|N(x_u)\cap (N_2\setminus N_1^T) |$. Also from the definition of $N_2$, every vertex in $N_2$ is adjacent to some vertex in $U_2$, thus every vertex in $(N_2\setminus N_1^T)$ has at least one edge incident on it in $B$. Thus, applying Fact~\ref{fact:bipartite} to the above inequality, we get

\begin{align*}
    \frac 1{|U_2|}\sum_{x_i\in U_2}\load(x_i) &\geq \frac 5 2 + \frac 1{2}\max\left(\frac{|N_2\setminus N_1^T |}{|U_2|} - 2, 0\right).
\end{align*}

Now from Condition~\ref{lb1} we have for each $x_i\in N_2\setminus N_1^T$, $\load(x_i)\geq 2$. 
Thus the average load in $N_2\setminus N_1^T$ is also at least $2$. For ease of readability set $N_2':=N_2\setminus N_1^T$, and let the average load of a set be given by the function $\mathsf{avg}(.)$.

Since $U_2$ and  $N_2\setminus N_1^T$ are disjoint, we have

\begin{align*}
\mathsf{avg}(S_3)=\mathsf{avg}(U_2 \cup N_2') 
&= \frac{|U_2|}{|U_2| + |N_2'|}\cdot \mathsf{avg}(U_2) + 
\frac{|N_2'|}{|U_2| + |N_2'|} \cdot \mathsf{avg}(N_2'),
\end{align*}

Let us set $p = \frac{|N_2'|}{|U_2|} $. 
Then we have $\mathsf{avg}(S_3)= \frac{1}{1+p}\cdot\mathsf{avg}(U_2) + (1 - \frac{1}{1+p}) \cdot  \mathsf{avg}(N_2').$
For $p = 2$, we have $\mathsf{avg}(S_3)= \frac{1}{3}\cdot\mathsf{avg}(U_2) +  \frac{2}{3}\cdot \mathsf{avg}(N_2')\geq \frac {1}{3} \cdot \frac 5 2 + \frac 2 3 \cdot2 = \frac {13}6=2(1+\frac 1{12})$.
Now, if $p<2$ the lower bound only increases (as we increase the proportion of $\mathsf{avg}(U_2)$).
For $p=2+\eps$, Observe that $\mathsf{avg}(U_2)\geq \frac 5 2 + \frac \eps 2 $. 

Thus, $\mathsf{avg}(S_3) = \frac 1{3+\eps}\cdot (\frac 5 2 + \frac \eps 2) + (1 - \frac 1{3+\eps})\cdot 2 =\frac{13+5\eps}{6+2\eps}\geq \frac{13}6= 2(1+\frac1{12})$, and we are done.

\paragraph*{average load in $S_4$} One can verify very easily that the above steps as in the case of $S_3$ can be followed to reach the same conclusion. 

This concludes the proof of the lemma.
\end{proof}

\subsection{A charging scheme for a valid assignment}

In this section, we define a load assignment and show that it is valid.
We first define $\xi(u, e)$ for each vertex $u\in \bigcup_{i\in[\ell]}(V(H_i)$ and each edge $e\in T$. Our load function then will be obtained as follows. $$\load(x_i) = \sum_{u\in x_i} 
\sum_{e\in T} \xi(u, e) \ .$$

We first partition $T= T_1\cup T_2$, and set $T_2 = \bigcup _{\ell+1 \leq i < j \leq k} E_R(V(H_i),V(H_j))$.
All edges in $T_2$ will transfer their weight to a fixed vertex $u_1^1$, i.e., $\xi(u_1^1, e)=1$ for all $e\in T_2$. Each unit-edge in $T_1$ will distribute a weight of $1$ between its two endpoints. Observe Conditions~\ref{cond1} and~\ref{cond2} in the definition of a valid assignment hold.

\paragraph*{Definition of $\xi$.}
We say a vertex $u$ is shortcut iff it is incident to a unit-edge of some $H_i$ which is shortcut by $R$.

We say that an edge $e \in T_1$ is \emph{crossing} if it is incident to two distinct components of $S$.
Let $e = \{ u, v \} \in T_1$ be some edge from component $H_i$ to component $H_j$ (not necessarily distinct).
Then define $\xi(u, e)$ (which represents the contribution of $e$ to $u$) as follows. 
\begin{itemize}
    \item[(1)] if $H_i=H_j$: $\xi(u,e)= \xi(v,e)=\frac 1 2$.
    \item[(2)] Else if $u$ shortcut, $v$ not shortcut: $\xi(u,e)=1$ and $\xi(v,e)=0$.
    \item[(3)] Else if $u$ and $v$ both not shortcut:
    $\xi(u,e)=\xi(v,e)=\frac 1 2.$
    \item[(4)] ({\bf note from now on: $H_i\neq H_j$  $u$ and $v$ are both shortcut})
    Else if neither $x_i$ nor $x_j$ is contained in $I$: $\xi(u,e)=\xi(v,e)=\frac 1 2.$
    \item[(5)] Else (at least one of $x_i$ or $x_j$ is contained in I)
    \item[a)] if $u$ has only one edge to $H_j$ and $v$ has only one edge to $H_i$ (namely $e$): $\xi(u,e)=\xi(v,e)=\frac 1 2$.
    \item[b)] Else if $u$ has at least $2$ outgoing edges to $H_j$ and $v$ has at least $2$ outgoing edges to $H_i$: $\xi(u,e)=\xi(v,e)=\frac 1 2$.
    \item[c)] Else if $u$ has only one edge, namely $e$ to $H_j$ and $v$ has at least $2$ edges to $H_i$: $\xi(u,e) = \frac{2}{3}$ and $\xi(v,e)= \frac 1 3$.
\end{itemize}

 Next, we will show that the remaining conditions, \ref{lb1}-\ref{lb5}, of the definition of valid assignment are  satisfied.

\subsection{Lower bounds on the $\load$}
In this section, we prove Lemma~\ref{lem:validAssignment}.

We now use this charging scheme $\xi$ and the fact that $S$ is a special configuration of some structured graph $G$ to show the  five Lemmas~\ref{lem:lb-1}-\ref{lem:lb-5}, that precisely establish Conditions~\ref{lb1} to~\ref{lb5} in the definition of a valid assignment, thus proving Lemma~\ref{lem:validAssignment}.
Depending on the type of node (defined earlier) we have five different lower bounds.
 First, we need a definition.
 \begin{definition}[$\delta(.)$] Given a vertex $v$, $\delta(v)$ is the set of unit-edges of $R$ incident on $v$.
 \end{definition}
 We now show the following useful lemma.

\begin{restatable}{lemma}{lemlb0}
\label{lem:lb-0}
For each vertex $v$ of some node $x_i \in A \setminus Z$ we have $\sum_{e \in \delta(v) \cap T_1} \xi(v, e) \geq \frac{1}{2}$.
\end{restatable}

\begin{proof}
Let $x_i \in A \setminus Z$ and let $v \in V(H_i)$.

First, if (5) c) applies to $v$, then $\sum_{e \in \delta(v) \cap T} \xi(v, e) \geq \frac{2}{3}$.
To see this, observe that in this case either $v$ is incident to one edge $e$ satisfying $\xi(v, e) = \frac{2}{3}$, or $v$ is incident to two edges $e_1, e_2$ satisfying $\xi(v, e_j) = \frac{1}{3}$, $j = 1,2$.
Hence, assume that (5) c) does not apply for $v$.

If $v$ is incident to a unit-edge $e \in E(H_i) \cap T_1$, then $\xi(v, e) = \frac{1}{2}$ by (1).
Hence, we assume that $v$ is not incident to such a unit edge. 
Thus, $v$ is shortcut.
By feasibility of $R$, we know that there must be at least one more unit-edge $e \in T$ that is incident to $v$.
But then either $\xi(v, e) = \frac{1}{2}$, if we are in case (3), (4), or (5) a), b), or $\xi(v, e) = 1$ if we are in case (2) (since $v$ is shortcut). This proves the claim.
\end{proof}

The above lemma together with $|V(H_i)| = 4$ immediately implies the following lemma.

\begin{restatable}{lemma}{lemlb1}
\label{lem:lb-1}
For each node $x_i \in A \setminus (Z \cup I)$ we have $\load(x_i) \geq 2$.
\end{restatable}

\begin{restatable}{lemma}{lemlb2}
\label{lem:lb-2}
For each node $x_i \in N_1^\top$ we have $\load(x_i) \geq \frac{5}{2}$.
\end{restatable}

\begin{proof}
Recall the definition of the nodes in $U_1$ and $N_1^\top$.
A node $x_i \in N_1^\top$ is adjacent to a single-shortcut node $x_j$ corresponding to some component $H_j$ that has an edge $\{u_1^j, u_2^j \}$ in $H_j \cap R$.
Furthermore, a vertex of $V(H_i)$ is adjacent to either $u_1^j$ or $u_2^j$ and, w.l.o.g., let this edge be $e= \{ u_1^j, u_1^i \}$.
Note that for edges between $x_j$ and $x_i$ only (1)-(4) apply, since $x_i, x_j \notin I$.

If the unit-edge $\{ u_1^i, u_2^i \}$ is contained in $R$, then by (3), we have that $\xi(u_1^i, e) = \frac{1}{2}$.
If the other unit-edge $\{ u_3^i, u_4^i \}$ is also contained in $R$, then we clearly have $\load(x_i) \geq \frac{5}{2}$, since then
$\xi(u_1^i, e) = \frac{1}{2}$, $\xi(u_1^i, \{ u_1^i, u_2^i \}) = \frac{1}{2}$, $\xi(u_2^i, \{ u_1^i, u_2^i \}) = \frac{1}{2}$, $\xi(u_3^i, \{ u_3^i, u_4^i \}) = \frac{1}{2}$, and $\xi(u_4^i, \{ u_3^i, u_4^i \}) = \frac{1}{2}$.

Else, if the unit-edge $\{ u_1^i, u_2^i \}$ is not contained in $R$, observe that in $R$ there have to be at least two more unit-edges outgoing from $u_3^i$ and $u_4^i$.
Let $u \in \{ u_3^i,  u_4^i\}$.
By Lemma~\ref{lem:lb-0}, we obtain $\sum_{e \in \delta(u) \cap T_1} \xi(u, e) \geq \frac{1}{2}$.
Then, since $\xi(u_1^i, e) = \frac{1}{2}$, $\xi(u_1^i, \{ u_1^i, u_2^i \}) = \frac{1}{2}$, and $\xi(u_2^i, \{ u_1^i, u_2^i \}) = \frac{1}{2}$, we have $\load(x_i) \geq \frac{5}{2}$.

Else, if the edge $\{ u_1^i, u_2^i \}$ is not present in $R$, then by (2), we have that $\xi(u_1^i, e) = 1$.
Furthermore, for $u \in \{u_2^i, u_3^i,  u_4^i\} $ we have $\sum_{e \in \delta(u) \cap T_1} \xi(u, e) \geq \frac{1}{2}$ by Lemma~\ref{lem:lb-0}.
Consequently, $\load(x_i) \geq \frac{5}{2}$. 
\end{proof}

\begin{restatable}{lemma}{lemlb3}
\label{lem:lb-3}
For each node $x_i \in U_1$ we have $\load(x_i) \geq \frac{5}{2} + \frac{1}{2} \cdot \max(| N^\bot_1(x_i)| - 2, 0)$.
\end{restatable}

\begin{proof}
Since $x_i \in U_1$, here only the definitions (1)-(4) apply, since by definition, $I$ neither contains $U_1$ nodes nor neighbors of $U_1$ nodes.
Recall the definition of a node $x_i \in U_1$.
The node $x_i$ corresponds to a small component $H_i$ such that $|| H[V(H_i)] \setminus E(R)|| = 1$, that is, $R$ contains precisely one unit-edge, say $\{u_1^i, u_2^i \}$, that is contained in $H_i$.
Furthermore, $x_i$ is an interior node of an open 2-augmenting path in $R / S$.
By definition, $\{u_3^i, u_4^i \}$ is the unit-edge of $E(H_i) \setminus E(R)$.
Recall that the set of nodes $N^\bot_1(x_i)$ is defined to be the set of neighbors of $x_i$ that are not incident to $u_1^i$ or $u_2^i$.
By definition of an open 2-augmenting path, there have to be at least one outgoing edge from $u_3^i$ to some component $H_p$ in $R$ and at least one outgoing edge from $u_4^i$ to $H_q$ in $R$ such that $p \neq q$.
W.l.o.g. let $u_1^p u_4^i u_1^i u_2^i u_3^i u_1^q$ be the 2-augmenting path in $R / S$ when expanded.

First, we show that the occurrence of $H_p$ and $H_q$ already lead to the fact that $\load(x_i) \geq \frac{5}{2}$. 
In a second step, we then show that any additional neighbor in $ N^\bot_1(x_i)$ contributes an additional load of at least $\frac{1}{2}$, which then proves the lemma.

If both unit-edges incident to $u_1^p$ and $u_1^q$ in $H_p$ and $H_q$, respectively, are contained in $R$, then it is straightforward to see that already this implies that $\load(x_i) \geq 3$, as follows:
$\xi(u_1^i, \{u_1^i, u_2^i \}) = \frac{1}{2}$ by (1), $\xi(u_2^i, \{u_1^i, u_2^i \}) = \frac{1}{2}$ by (1), $\xi( u_4^i, \{u_1^p, u_4^i \}) = 1$ by (3), and $\xi( u_3^i, \{u_1^q, u_4^i \}) = 1$ by (3).

Else, if only one of these edges is in $R$, say $\{u_1^p, u_2^p \}$, then observe that this implies  $\load(x_i) \geq \frac{5}{2}$, by (2), (3), and (4), since:
$\xi(u_1^i, \{u_1^i, u_2^i \}) = \frac{1}{2}$ by (1), $\xi(u_2^i, \{u_1^i, u_2^i \}) = \frac{1}{2}$ by (1), $\xi( u_4^i, \{u_1^p, u_4^i \}) = \frac{1}{2}$ by (4), and $\xi( u_3^i, \{u_1^q, u_4^i \}) = 1$ by (3).

Finally, consider the case that both unit-edges $\{u_1^p, u_2^p \}$ and $\{u_1^q, u_2^q \}$ are not contained in $R$.
In this case, we apply Lemma~\ref{lem:lower-bound-help-lemma} and obtain that there are additional edges $e_p$ from $H_p$ to $H_i$ and $e_q$ from $H_q$ to $H_i$. 
Since $x_p, x_q \in N^\bot_1(x_i)$ these edges are also incident to $u_3^i$ and $u_4^i$.
Now, also we have $\xi(u_1^i, \{u_1^i, u_2^i \}) = \frac{1}{2}$, $\xi(u_2^i, \{u_1^i, u_2^i \}) = \frac{1}{2}$, $\xi( u_4^i, \{u_1^p, u_4^i \}) = \frac{1}{2}$, and $\xi( u_3^i, \{u_1^q, u_4^i \}) = \frac{1}{2}$. 
Furthermore, either $\xi( u_3^i, e_p) \geq \frac{1}{2}$ or $\xi( u_4^i, e_p) \geq \frac{1}{2}$ and $\xi( u_3^i, e_q) \geq \frac{1}{2}$ or $\xi( u_4^i, e_q) \geq \frac{1}{2}$, all by (1) or (4).
Hence, this implies $\load(x_i) \geq 3$.

To see the second part, observe that each additional node $x_j \in N^\bot_1(x_i)$ is also incident to either $u_3^i$ or $u_4^i$ and hence contributes at least a load of $\frac{1}{2}$ to $x_i$.
Thus, combining we have $\load(x_i) \geq \frac{5}{2} + \frac{1}{2} \cdot \max(| N^\bot_1(x_i)| - 2, 0)$.
\end{proof}

\begin{restatable}{lemma}{lemlb4}
\label{lem:lb-4}
For each node $x_i \in U_2$ we have $\load(x_i) \geq \frac{5}{2} + \frac{1}{2} \cdot \max(|N(x_i)| - 2, 0)$.
\end{restatable}

\begin{proof}
Since $x_i \in U_2$, here only the definitions (1)-(4) apply, since by definition, $I$ neither contains $U_2$ nodes nor neighbors of $U_2$ nodes.
Recall the definition of a node $x_i \in U_2$.
The node $x_i$ corresponds to a small component $H_i$ such that $|| H[V(H_i)] \setminus E(R)|| = 2$, that is, $R$ does not contain any of the unit-edges in $H_i$ (but possibly a diagonal edge).
Furthermore, $x_i$ is an interior node of an open 2-augmenting path in $R / S$.

\textbf{Case A:} Assume that $R$ does not use a diagonal edge of $H_i$.
W.l.o.g. let $u_1^p u_4^i u_1^i u_2^i u_3^i u_1^q$ be the 2-augmenting path in $R / S$ when expanded, where $i, p, q$ are pairwise distinct.

We make a case distinction on which subset of vertices of $H_i$ are incident to the edge-sets $E_R(V(H_i),V(H_p))$ and $E_R(V(H_i),V(H_q))$.

\textbf{Case A1:} assume that in $R$ only the vertices $u_4^i$ and $u_3^i$ are connected to the vertices of $H_p$ or $H_q$.
Observe that in this case, by (2) and (4), we have that
$$ \sum_{e \in E_R( \{u_4^i\} ,V(H_p) \cup V(H_q)) } \xi(u_4^i, e) + \sum_{e \in E_R( \{u_3^i\} ,V(H_p) \cup V(H_q)) } \xi(u_3^i, e) \geq 2 \ . $$
To see this, note that either $\xi(u_4^i, \{u_4^i, u_1^p\}) = 1$ (by (2), if $u_1^p$ is not shortcut), or, by Lemma~\ref{lem:lower-bound-help-lemma}, there is one additional edge $e_p$ from $V(H_p)$ to $V(H_i)$, and therefore to $\{u_4^i, u_1^p\}$.
Without loss of generality assume that $e_p$ is incident to $u_4^i$.
In the latter case, we have $\xi(u_4^i, \{u_4^i, u_1^p\}) = \frac{1}{2}$, and $\xi(u_4^i, e_p) = \frac{1}{2}$.
Additionally, we have the same statement for $u_3^i$ and $H_q$.
Hence, in either case we obtain the above statement.

Furthermore, since $\{u_1^i, u_2^i\}$ is not present in $R$, we have that there is one edge $e_1$ outgoing from $u_1^i$ and one edge $e_2$ outgoing from $u_2^i$.
If $e_1$ and $e_2$ go to the same component, say $H_r$, then 
$$ \sum_{e \in E_R( \{u_1^i\} ,V(H_p) \cup V(H_q)) } \xi(u_1^i, e) + \sum_{e \in E_R( \{u_2^i\} ,V(H_p) \cup V(H_q)) } \xi(u_2^i, e) \geq 1 , $$
and hence
$$ \sum_{e \in E_R(V(H_i),V(H_p) \cup V(H_q) \cup V(H_r)) \ } \sum_{\ v \in V(H_i)} \xi(v, e) \geq 3 . $$
Furthermore, each additional neighbor of $x_i$ will contribute an additional load of at least $\frac{1}{2}$ to some vertex of $H_i$.
To see this, observe that $x_i$ is double-shortcut and hence for each neighbor either (4) applies ($x_i$ gets $\frac{1}{2}$ extra load) or (4) applies ($x_i$ gets $1$ extra load).

Hence, let us assume that $e_1$ and $e_2$ go to distinct components, say $H_1$ and $H_2$.
Then, $e_1$ and $e_2$ also form an open 2-augmenting path and by Lemma~\ref{lem:lower-bound-help-lemma} we have that there are two more edges from $H_1$ and $H_2$ to the vertices of $H_i$. 
Hence, we have that 
$$ \sum_{e \in E_R(V(H_i),V(H_p) \cup V(H_q) \cup V(H_1) \cup V(H_2)) \ } \sum_{ \ v \in V(H_i)} \xi(v, e) \geq 4 , $$ by (2) and (4) and since $x_i \in U_2$.
Furthermore, each additional neighbor of $x_i$ will contribute an additional load of at least $\frac{1}{2}$ (similar to before), which proves the claim.

\textbf{Case A2:} assume that in $R$ only the vertices $u_4^i$, $u_3^i$, and $u_2^i$ are connected to the vertices of $H_p$ or $H_q$. The case in which the vertices $u_4^i$, $u_3^i$, and $u_1^i$ are connected to the vertices of $H_p$ or $H_q$ is analogue.
Observe that in this case, similar to the previous case, by (2) and (4) and the application of Lemma~\ref{lem:lower-bound-help-lemma}, we have that
$$ \sum_{e \in E_R(V(H_i),V(H_p) \cup V(H_q))  \ } \sum_{ \ v \in V(H_i)} \xi(v, e) \geq 2 . $$
Furthermore, since $\{u_1^i, u_2^i\}$ is not present in $R$, we have that there is one edge $e_2$ outgoing from $u_2^i$ to some component $H_r$, distinct from $H_p$ and $H_q$.
Similar to before, then there is another open 2-augmenting path with $H_r$, $H_i$, and either $H_p$ or $H_q$ (depending on where the outgoing edge of $u_1^i$ is going), and hence by Lemma~\ref{lem:lower-bound-help-lemma} there has to be one more edge from $H_r$ to $H_i$.
Hence, by similar arguments as before, we have that 
$$ \sum_{e \in E_R(V(H_i),V(H_p) \cup V(H_q) \cup V(H_r))  \ } \sum_{ \ v \in V(H_i)} \xi(v, e) \geq 3 . $$
Furthermore, each additional neighbor of $x_i$ will contribute an additional load of at least $\frac{1}{2}$, which proves the claim.

\textbf{Case A3:} consider the case that in $R$ all vertices of $H_i$ are connected to the vertices of $H_p$ or $H_q$.
Note that Lemma~\ref{lem:lower-bound-help-lemma} guarantees the existence of the edges $e_p$ from $u_2^p$ to $H_i$ and/or $e_q$ from $u_2^q$ to $H_i$ if $\{ u_1^p, u_2^p \}$ and/or $\{ u_1^q, u_2^q \}$ is not present in $R$.
Furthermore, if $e_p$ is incident to $u_1^i$ and $e_q$ is incident to $u_2^i$, then $P$ together with the edges $e_p$ and $e_q$ contains a small to large merge, a contradiction to the fact that $G$ is special.
If $e_p$ is incident to $u_2^i$ and $e_q$ is incident to $u_1^i$, then $P$ together with the edges $e_p$ and $e_q$ contains a small to medium merge, again a contradiction to the fact that $G$ is special.
Hence, in any case, we directly obtain that 
$$ \sum_{e \in E_R(V(H_i),V(H_p) \cup V(H_q)) \ } \sum_{ \ v \in V(H_i)} \xi(v, e) \geq \frac{5}{2} , $$
since there is at least one vertex $v \in V(H_i)$ incident to two unit-edges, and since $x_i$ is a double-shortcut node.
Furthermore, similar as before, each additional neighbor of $x_i$ will contribute an additional load of at least $\frac{1}{2}$, which proves the claim.

\textbf{Case B:} assume that $R$ uses a diagonal edge in $H_i$.
W.l.o.g. let this edge be $\{ u_2^i, u_4^i \}$.
Note that Lemma~\ref{lem:lower-bound-help-lemma} guarantees the existence of the edges $e_p$ from $u_2^p$ to $H_i$ and/or $e_q$ from $u_2^q$ to $H_i$ if $\{ u_1^p, u_2^p \}$ and/or $\{ u_1^q, u_2^q \}$ is not present in $R$.
This directly implies that 
$$ \sum_{e \in E_R(V(H_i),V(H_p) \cup V(H_q))  \ } \sum_{ \ v \in V(H_i)} \geq 3 , $$
by similar arguments as used previously.
Furthermore, as before, each additional neighbor of $x_i$ will contribute an additional load of at least $\frac{1}{2}$, since $x_i$ is a double-shortcut node. This proves the lemma.
\end{proof}

\begin{restatable}{lemma}{lemlb5}
\label{lem:lb-5}
For each node $x_i \in I$ we have $\load(x_i) \geq \frac{7}{3}$.
\end{restatable}

\begin{proof}
Since $x_i \in I$, here only the definitions (1), (2), (3), and (5) a), b), c) apply. 
By definition of the set $I$, we know that $x_i \in B_1$ or $x_i \in B_2$ and hence there is at least one unit-edge, say $\{ u_3^i, u_4^i \}$, that is not contained in $R$. 

Let us first assume that $R$ does not use a diagonal edge in $H_i$.
Then in $R$, there have to be outgoing edges $e_3$ from $u_3^i$ and $e_4$ from $u_4^i$ that go to the same component $H_p$ (since $x_i \in I$).
Let $e_3 = \{ u_3^i, v \}$ and $e_4 = \{ u_3^i, w \}$, $v, w \in V(H_p)$.
First, assume that $v \neq w$.
Note that $v$ and $w$ can not be $u_1^p$ and $u_2^p$ or $u_3^p$ and $u_4^p$, since then $e_3$ and $e_4$ form a good cycle, a contradiction to the fact that $G$ is special. 
Furthermore, if in $G$ there is some diagonal edge $e_d$ in $H_p$, then $v$ and $w$ can not be both vertices that are not incident to $e_d$. 
Otherwise, again there is a good cycle.

Next, assume that $v$ and $w$ are any other vertices such that $v \neq w$.
If $v$ or $w$ are not shortcut, we obtain $\load(x_i) \geq \frac{5}{2}$.
To see this observe that in this case either $\sum_{z \in V(H_i)} \xi(z, e_3) = 1$ or $\sum_{z \in V(H_i)} \xi(z, e_4) = 1$ (that is, there is one vertex $z \in V(H_i)$ such that $\xi(z, e_3) = 1$, by (2)) and together with the remaining unit-edges incident to $H_i$ in $T_1$ we have the desired result.

Thus, we assume that the distinct unit-edges incident to $v$ and $w$ in $H_i$ are both shortcut.
Since $x_p$ is not a $U_2$-node (otherwise $x_i$ is not an $I$-node), there have to be at least 4 edges from the vertices of $H_p$ to the vertices of $H_i$ in $T_1$.
If all vertices of $H_i$ are incident to an edge that goes to $H_p$, then $R$ is not 2-edge-connected.
To see this, observe that both $x_i$ and $x_p$ are no $U_1$ or $U_2$ nodes, and hence none of them can be connected to some other node $x_q$, a contradiction.
Hence, at least one vertex of $H_i$ is incident to two unit-edges of $E_R(V(H_i),V(H_p))$.
Hence, by (2), (3), (4) and (5), we have that $\load(x_i) \geq \frac{5}{2}$.
To see this, observe that if $x_i \in B_1$, then the unit-edge of $R$ inside $H_i$ will give the remaining load.
Otherwise, if $x_i \in B_2$, then there has to be one more outgoing edge from $x_i$ to some other component $H_r$ (since 2 edges are shortcut in $x_i$) and the edge-set incident to this vertex and $H_r$ will give the remaining load.

Hence, w.l.o.g.\ we can assume that $v = w = u_1^p$.
In this case, observe that by (5) c) we have that $\xi(u_3^i, e_3) = \xi(u_4^i, e_4) = \frac{2}{3}$.
Furthermore, if $x_i \in B_1$, then the unit-edge of $T_1$ inside $H_i$ will give another load of 1 for $H_i$
Otherwise, if $x_i \in B_2$, then there have to be two more outgoing edges from $x_i$ to some other component (since 2 edges are shortcut in $x_i$) and these edges will give another load of 1 for $H_i$.
Hence, we have that $\load(x_i) \geq \frac{7}{3}$.

Finally, let us assume that $R$ uses a diagonal edge in $H_i$, say $\{ u_2^i, u_4^i \}$.
By doing the same case distinction as in the case in which $R$ does not use a diagonal edge, we also obtain $\load(x_i) \geq \frac{7}{3}$.
\end{proof}


\newpage
\section{The gluing algorithm of Cheriyan et al.\ (proof of Lemma~\ref{lem:glue})}
\label{sec:Glue}

In this section, we will prove the following lemma.

\lemglue*

Note that this lemma is already proven implicitly in~\cite{CheriyanCDZ21}. 
However, we still give the full proof here since our notation differs from theirs. 
Furthermore, they have some obstructions that we do not explicitly exclude in the preprocessing, but which are implicitly subsumed by our preprocessing, in particular by our contractable subgraphs.

Throughout this section, we assume that at the beginning of the algorithm we are given a special configuration $S$ of $G$. 
In particular, this means that $S$ consists only of small and large components.
During the algorithm, we always keep a 2-edge-cover $H$, where at the beginning of the algorithm we have $H = S$.

In order to prove Lemma~\ref{lem:lowerbound}, we use the same credit invariant as used in~\cite{CheriyanCDZ21}.
\begin{itemize}
    \item[(i)] Each large component receives a credit of at least 2, and 
    \item[(ii)] each small component receives a credit of $\frac{4}{3}$.
\end{itemize}
For a component $C$ of $H$, we write $\newcredit(C)$ to denote its credit.
Note that each component of $H$ corresponds to a node in $G / H$ and hence we sometimes say that a node of $G / H$ has certain credit.
Furthermore, for a 2-edge-connected subgraph $K$ in $G / H$ we keep the notation of $\savings(H, K)$ from Section~\ref{sec:ComputingSpecialConfiguration}.

The goal is to \emph{glue} the components of $H$ together while maintaining this credit invariant until eventually we only have a single component left, which is then a 2-ECSS.
It is straightforward to verify that such an algorithm satisfies the conditions of Lemma~\ref{lem:lowerbound}.
In order to define the algorithm, first recall the definitions of a good cycle and an open/closed augmenting path from Section~\ref{sec:ComputingSpecialConfiguration}.
Furthermore, we need the following definition.

\begin{definition}[stacked closed augmenting paths]
Let $H$ be a bridgeless 2-edge-cover of some structured graph $G$ and let $P_1$ and $P_2$ be two closed augmenting paths in $G / H$.
We say that $(P_1, P_2)$ is a \textbf{stacked closed augmenting path} if $|V(P_1) \cap V(P_2) | = 1$ and $V(P_1) \cap V(P_2)$ is not an interior node of $P_1$.
We say that $(P_1, P_2)$ is a \textbf{stacked closed 2-augmenting path} if we additionally have that $|E(P_1)| = |E(P_2) | = 2$.
\end{definition}

We are now ready to define Algorithm $\glue(G, S)$.
The Merge-steps are defined in Lemma~\ref{lem:glue:merge-enough-credit}.

\begin{algorithm}
\caption{Gluing Special Configuration}\label{alg:gluespecial}
\begin{algorithmic} 
\Function{Glue}{$(G, S)$} 
\State $H = S$.
\While{$H$ has at least 2 connected components}
\If{$G / H$ contains a good cycle $C$} 
\State Merge $H$ and $C$ to obtain graph $H'$.
\Return $\glue(G, H')$.
	\EndIf
\If{$G / H$ contains an open 2-augmenting path $P$} 
\State Merge $H$ and $P$ to obtain graph $H'$.
\Return $\glue(G, H')$.
	\EndIf
\If{$G / H$ contains a stacked closed 2-augmenting path $(P_1, P_2)$} 
\State Merge $H$ and $(P_1, P_2)$ to obtain graph $H'$.
\Return $\glue(G, H')$.
	\EndIf
\EndWhile
\Return $H$.
\EndFunction
\end{algorithmic}
\end{algorithm}

In order to prove that the Algorithm $\glue(G, S)$ outputs a feasible solution in polynomial time and additionally satisfies the credit invariant, we show the following lemmas.
First, we show that any of the three structures, which we are searching for in the algorithm, can be found in polynomial time.
Second, we show that whenever we are given such a structure, then we can merge them into a single component using a polynomial-time algorithm, that also satisfies the credit invariant.
Finally, we show that if there are no good cycles and no open 2-augmenting paths in $G / H$, then there is a stacked closed 2-augmenting path. 
This guarantees that the algorithm indeed outputs a feasible solution.
Note that in the last part we crucially exploit the fact that $G$ is structured.
Hence, the following three lemmas together with the new credit invariant prove Lemma~\ref{lem:glue}.

\begin{lemma}
\label{lem:glue:find-obstructions}
Let $H$ be a bridgeless 2-edge-cover of some structured graph $G$.
Then we can find in polynomial time one of the obstructions specified in Algorithm $\glue(G, S)$, namely a good cycle, an open 2-augmenting path, or a stacked closed 2-augmenting path.
\end{lemma}

\begin{proof}
Note that the proof for good cycles is already given in Section~\ref{sec:ComputingSpecialConfiguration} in Lemma~\ref{lem:find-special-obstruction-poly}.
Furthermore, note that each of the remaining obstructions only has a constant size, and hence we can find such objects in polynomial time.
\end{proof}

We now show the second part.

\begin{lemma}
\label{lem:glue:merge-enough-credit}
Let $H$ be a bridgeless 2-edge-cover of some structured graph $G$ that satisfies the new credit invariant and has a good cycle, an open 2-augmenting path, or a stacked closed 2-augmenting path. 
Then there is a bridgeless 2-edge-cover $H'$ of $G$ that satisfies the new credit invariant and has strictly fewer components than $H$.
\end{lemma}

\begin{proof}
We consider each obstruction one by one.
First, note that good cycles have already been handled in Section~\ref{sec:ComputingSpecialConfiguration} with a stronger credit invariant in Lemma~\ref{lem:handling-good-cycles}. 
Hence, it remains to prove the result for the remaining two obstructions.

Let $P$ be an open 2-augmenting path $x_1, x_2, x_3$ of $G / H$, where the components $H_1, H_2, H_3$ correspond to the nodes $x_1, x_2, x_3$, respectively.
In $G /H$ there has to be a path $P'$ from $x_1$ to $x_3$ such that $P' \cap x_2 = \emptyset$, as otherwise $H_2$ is a separator in $G$, a contradiction to the fact that $G$ is structured and does not contain an $S_{\{3,4\}}$.
First, observe that $P \cup P'$ is a cycle in $G/H$ and hence we can use these edges to merge all nodes in $P \cup P'$ to a single node.
Furthermore, note that $||E(P) \cup E(P') || = | V(P) \cup V(P')|$ and that, by the definition of an open 2-augmenting path, $\savings(H, E(P \cup E(P')) \geq 1$.
Let $C_S'$ be the newly created connected component containing all nodes in $P$ and $P'$.
Then we have that
\begin{align*}
\newcredit(C_S') & \geq \frac{4}{3} \cdot |V(P) \cup V(P')| + 1 - |V(P) \cup V(P')| \geq 2 \ ,
\end{align*}
since $|V(P) \cup V(P')| \geq 3$.

Now consider a stacked closed 2-augmenting path $(P_1, P_2)$ of $(G, H)$.
Note that, by definition, $E(P_1) \cup E(P_2)$ is 2-edge-connected in $G/H$.
Furthermore, $\savings(H, E(P \cup E(P')) \geq 2$ and $||E(P_1) \cup E(P_2)|| = 4$.
Let $C_S'$ be the newly created connected component containing all nodes in $P_1$ and $P_2$.
Then we have that
\begin{align*}
\newcredit(C_S') & \geq \frac{4}{3} \cdot |V(P_1) \cup V(P_2)| + 2 - |V(P_1) \cup V(P_2)| = 2 \ ,
\end{align*}
since $|V(P) \cup V(P')| = 3$.
This proves the lemma.
\end{proof}

Finally, we show the following lemma.
The proof heavily follows the proof of Lemma 28 in~\cite{CheriyanCDZ21} and Lemma 6.4 of the full version of their paper, respectively.

\begin{lemma}
\label{lem:glue:existence-of-stacked-path}
Let $H$ be a bridgeless 2-edge-cover of some structured graph $G$ such that $G / H$ does not contain a good cycle or an open 2-augmenting path.
Then $G / H$ contains a stacked closed 2-augmenting path.
\end{lemma}

\begin{proof}
For each component $H_i$ of $H$ let $x_i$ be its corresponding node in $G/H$.
First, we observe that each small component of $H$ corresponds to a node $x$ in $G / H$ which is an interior node of some augmenting path.
Otherwise, observe that any optimal solution picks at least 2 unit edges inside such a small component.
But then such a component forms a contractable subgraph, a contradiction to the fact that $G$ is structured.
Since $G / H$ does not contain any open 2-augmenting paths by assumption, we know that each augmenting path has length 2 and is closed.
We say that a closed 2-augmenting path $P$ \emph{belongs} to some component $H_i$, if and only if $x_i$ is the interior node of $P$. 
Note that this is unique since there is only one interior node.

We now consider the following auxiliary digraph $D^{aux}$ from~\cite{CheriyanCDZ21}: we have a node for each component of $H$ and we call the nodes corresponding to the small components the \emph{red} nodes and the other nodes the \emph{green} nodes.
For each small component $H_i$ of $H$ and each closed 2-augmenting path $P = x_j x_i x_j$ that belongs to $H_i$, we add the arc $(x_i, x_j)$.
By the above discussion, we have that each red node has at least one outgoing arc.

Note that a stacked closed 2-augmenting path corresponds to a path of length 2 in $D^{aux}$.
Furthermore, if a red node has an arc to a green node, then this corresponds to a good cycle.
Therefore, let us assume that there is no such path of length 2 in $D^{aux}$ and that there is no arc from a red node to a green node.

Hence, we can assume that in $D^{aux}$ each red node has precisely one outgoing arc such that each red node is contained in a directed cycle of length precisely 2.
We will show that this contradicts the fact that $G$ is structured.
Let $x_1$ and $x_2$ be two nodes of $D^{aux}$ that form a cycle and let $H_1$ and $H_2$ be its corresponding components in $H$.
Let $P_i$ be the closed 2-augmenting path that belongs to $x_i$, $i = 1,2$.
For $i = 1,2$ let $u_i$ and $w_i$ be the vertices of $H_i$ that are incident to $P_i$ when expanded and let $T_i \subseteq V(H_i)$ be the vertices that have edges to $V(G) \setminus (V(H_1) \cup V(H_2))$.
The proof is completed via the following claims.
$$ |T_1 \cup T_2| \geq 2, T_1 \neq \emptyset, \text{ and } T_1 \cap \{u_1, w_1 \} = \emptyset.   $$
We first prove this claim. 
Since $G/H$ has more than 2 nodes and $G$ does not have cut-vertices ($G$ is structured), we have that $|T_1 \cup T_2| \geq 2$.
If $T_1$ is empty, then all outgoing edges of $H_1$ go to $H_2$.
But then $H_2$ forms a $S_{\{3,4\}}$, a contradiction.
Finally, $T_1 \cap \{u_1, w_1 \} \neq \emptyset$ is a contradiction to the fact that $(G, H)$ does not contain an open 2-augmenting path. 

We now prove the next claim.
\begin{align*}
& \text{ $H_1$ is a 4-cycle $C_1$, $u_1$ and $w_1$ are connected by a unit-edge $e_1$ of $C_1$ such that all neighbors} \\
& \text{ (in $G$) of $u_1$ and $w_1$ are in $V(H_1) \cup V(H_2)$; moreover $|T_1| = 1$, $T_1 = v_1 \neq u_1, w_1$, and the } \\
& \text{ vertex $z_1 = V(H_1) \setminus \{u_1, v_1, w_1\}$ is incident to exactly two edges in $G$.}
\end{align*} 
We now prove this claim.
There are two cases for $u_1$ and $w_1$:
\begin{itemize}
    \item[(i)] $u_1$ and $w_1$ are not adjacent in $H_1$, so $H_1$ is a 4-cycle, and the other two vertices in $H_1$ are adjacent in $G$ ($E(V(H_i))$ has a "diagonal edge"), or
    \item[(ii)] $H_i$ has a unit-edge between $u_1$ and $w_1$.
\end{itemize}
Consider case (i).
Let $v_1$ be a vertex of $(V(H_i) \setminus \{u_1, w_1 \}) \cap T_1$.
Then the unit-edge $f_1$ of $H_1$ incident to $v_1$ is part of an open 2-augmenting path, a contradiction. 
Hence, case (i) can not occur.
Next, consider case (ii).
Recap that all neighbors of $u_1$ and $w_1$ are in $V(H_1) \cup V(H_2)$.
If $T_1$ contains exactly one vertex, say $v_1$, then $H_1$ cannot be a 3-cycle, as otherwise there exists an open 2-augmenting path with interior node $x_1$ corresponding to $H_1$.
Hence, $H_1$ is a 4-cycle such that $T_1$ contains exactly one vertex and this vertex is not $u_1$ and not $w_1$.
If both nodes of $V(H_1) \setminus \{u_1, v_1 \}$ are in $T_1$.
But then $x_1$ contains two outgoing arcs in $D^{aux}$, a contradiction (then $x_1$ would be part of an open 2-augmenting path or part of a stacked closed 2-augmenting path, as discussed at the beginning of this proof).
Finally, if $z_1 = V(H_1) \setminus \{u_1, v_1, w_1\}$ is incident to more than two edges in $G$, observe that $(G, H)$ contains an open 2-augmenting path, a contradiction.
Hence, we have finished the proof of the second claim.

Note that the analogue claims for $H_2$ also hold, by symmetry.
Now observe that the subgraph induced by the vertices of $H_1$ and $H_2$ forms a contractable subgraph: 
any optimal solution must buy the unit-edge incident to $z_i$, $i = 1,2$.
Furthermore, any optimal solution must buy 2 additional unit-edges incident to $u_i$ and $w_i$, $i = 1,2$.
Hence, any optimal solution must pick at least 4 unit-edges of $E(V(H_1) \cup V(H_2)$. 
However, 5 edges suffice in order to 2-edge-connect $V(H_1) \cup V(H_2)$:
Take all unit-edges of $H_1$ and $H_2$ except $\{u_1, w_1 \}$ and additionally add the closed 2-augmenting path that is incident to the vertices $u_1$ and $w_1$.
This is a contradiction to the fact that $G$ is structured and hence does not contain contractable subgraphs.
\end{proof}


\end{document}